\documentclass[
a4paper,reqno]{amsart}
\usepackage[T1]{fontenc}
\usepackage[english]{babel}
\usepackage{amssymb,bbm,enumerate}
\usepackage{color}
\usepackage[linktocpage=true,colorlinks=true, linkcolor=blue, citecolor=red, urlcolor=green]{hyperref}
\usepackage[all]{xy}

\renewcommand{\phi}{\varphi}

\newcommand{\mc}[1]{\mathcal{#1}}
\newcommand{\mf}[1]{\mathfrak{#1}}
\newcommand{\mb}[1]{\mathbb{#1}}

\newcommand{\id}{\mathbbm{1}}
\newcommand{\quot}[2] {\ensuremath{\raisebox{.40ex}{\ensuremath{#1}}
\! \big / \! \raisebox{-.40ex}{\ensuremath{#2}}}}
\newcommand{\tint}{{\textstyle\int}}

\DeclareMathOperator{\Mat}{Mat}
\DeclareMathOperator{\Hom}{Hom}
\DeclareMathOperator{\End}{End}
\DeclareMathOperator{\diag}{diag}
\DeclareMathOperator{\tr}{Tr}
\DeclareMathOperator{\res}{Res}
\DeclareMathOperator{\ad}{ad}
\DeclareMathOperator{\im}{Im}

\DeclareMathOperator{\Der}{Der}
\DeclareMathOperator{\Ker}{Ker}

\newcommand{\dx}{\delta_x}
\DeclareMathOperator{\wt}{wt}

\theoremstyle{plain}
\newtheorem{theorem}{Theorem}[section]
\newtheorem{lemma}[theorem]{Lemma}
\newtheorem{proposition}[theorem]{Proposition}
\newtheorem{corollary}[theorem]{Corollary}

\theoremstyle{definition}
\newtheorem{definition}[theorem]{Definition}
\newtheorem{example}[theorem]{Example}

\theoremstyle{remark}
\newtheorem{remark}[theorem]{Remark}

\numberwithin{equation}{section}

\definecolor{light}{gray}{.9}

\title{Classical \texorpdfstring{$\mc W$}{W}-algebras 
and generalized Drinfeld-Sokolov bi-Hamiltonian systems 
within the theory of Poisson vertex algebras}
\author{Alberto De Sole, Victor G. Kac, Daniele Valeri}
\address{Dept. of Math., Univ. of Rome 1,
P.le Aldo Moro 2, 00185 Roma, Italy.}
\email{desole@mat.uniroma1.it}
\address{Dept. of Math., MIT,
77 Massachusetts Ave, Cambridge, MA 02139, USA.}
\email{kac@math.mit.edu}
\address{
SISSA-ISAS,
Via Bonomea 265, 34136 Trieste, Italy.}
\email{daniele.valeri@sissa.it}

\begin{document}

\pagestyle{plain}

\begin{abstract}
We describe of the generalized Drinfeld-Sokolov Hamiltonian reduction 
for the construction of classical $\mc W$-algebras 
within the framework of Poisson vertex algebras.
In this context, the gauge group action on the phase space
is translated in terms of (the exponential of) a Lie conformal algebra action
on the space of functions.
Following the ideas of Drinfeld and Sokolov,
we then establish under certain sufficient conditions the applicability 
of the Lenard-Magri scheme of integrability 
and the existence of the corresponding integrable hierarchy of bi-Hamiltonian equations.
\end{abstract}

\maketitle

\tableofcontents

\section{Introduction}

In the seminal paper \cite{DS85}, 
Drinfeld and Sokolov defined a 1-parameter family of Poisson brackets
on the space $\quot{\mc W}{\partial\mc W}$ of local functionals 
on an infinite-dimensional Poisson manifold $\mc M$.
Such Poisson manifold is obtained, 
starting from an affine Kac-Moody algebra $\widehat{\mf g}$, via a Hamiltonian reduction,
and the corresponding differential algebra $\mc W$ of functions on $\mc M$,
with its Poisson bracket on the space $\quot{\mc W}{\partial\mc W}$ of local functionals,
is known as principal classical $\mc W$-algebra.
In the same paper Drinfeld and Sokolov 
constructed an integrable hierarchy of bi-Hamiltonian equations associated
to each principal classical $\mc W$-algebra.
The Korteweg-de Vries (KdV) equation appears in the case of $\mf g=\mf{sl}_2$.
For $\mf g=\mf{sl}_n$, the principal classical $\mc W$-algebra Poisson bracket
coincides with the Adler-Gelfand-Dickey Poisson bracket \cite{Adl79,GD87}
on the space of local functionals on the set of ordinary differential operators of the 
form $\partial^n+u_1\partial^{n-2}+\ldots+u_{n-1}$,
and the corresponding integrable hierarchy is the so called $n$-th Gelfand-Dickey hierarchy
(see \cite{Dic97} for a review).

In a few words, the construction of \cite{DS85} is as follows.
Let $\mf g$ be a simple finite-dimensional Lie algebra
with a non-degenerate symmetric invariant bilinear form $\kappa$,
and let $f$ be a principal nilpotent element in $\mf g$,
which we include in an $\mf{sl}_2$-triple $(f,h=2x,e)$ in $\mf g$.
Then $\mf g$ decomposes as a direct sum of $\ad x$-eigenspaces
$\mf g=\bigoplus_{i\in\mb Z}\mf g_i$.
To define the Poisson manifold $\mc M$, consider first the space $\widetilde{\mc M}$ 
of first order differential operators of the form
$$
L(z)=\partial_x+f+zs+q(x)\,,
$$
where $s$ is a fixed non-zero element of the center of $\mf n_+=\bigoplus_{i>0}\mf g_i$,
$q(x)$ is a smooth map $S^1\to\mf b=\mf g_{0}\oplus\mf n_+$,
and $z$ is an indeterminate.
On this space there is an action of the infinite-dimensional Lie group $N$,
whose Lie algebra is the space of smooth maps $S^1\to\mf n_+$,
by gauge transformations:
$$
L^{A}(z)=e^{\ad A}L(z)\,,
$$
for any smooth map $A:\, S^1\to\mf n_+$.
The Poisson manifold $\mc M$ is then obtained as the quotient of $\widetilde{\mc M}$
by the action of the gauge group.
As a differential algebra, the principal classical $\mc W$-algebra is therefore the space
of functions on $\widetilde{\mc M}$ which are gauge invariant.
The corresponding 1-parameter family of Poisson brackets on $\quot{\mc W}{\partial\mc W}$
($z$ being the parameter)
is obtained as a reduction of the affine algebra Lie-Poisson bracket.
An explicit formula for it is
$$
\{\tint g,\tint h\}_{z,\rho}
=\int\kappa\Big(\frac{\delta h}{\delta q}\,\Big|\,\big[L(z),\frac{\delta g}{\delta q}\big]\Big)\,,
$$
where $\frac{\delta g}{\delta q}$ denotes the variational derivative of the local functional 
$\tint g\in\quot{\mc W}{\partial\mc W}$
(the index $\rho$ will be explained in Section \ref{sec:3.2}).

In order to construct an integrable hierarchy of bi-Hamiltonian equations for $\mc W$,
one conjugates $L$ to an operator of the form
$$
L_0(z)=e^{\ad U(z)}L(z)=\partial_x+f+zs+h(z)\,,
$$
where $U(z)$ is a smooth function on $S^1$ with values in $\mf n_+\oplus\mf g[[z^{-1}]]z^{-1}$,
and $h(z)$ is a smooth function on $S^1$ with values in $\mf h\cap\mf g[[z^{-1}]]$,
where $\mf h=\Ker\ad(f+zs)$ (it is an abelian subalgebra of $\mf g((z^{-1}))$).
Then, for any element $a(z)\in\mf h$ we obtain an infinite sequence of Hamiltonian functionals in involution
defined by ($n\in\mb Z_+$):
$$
\tint \mc H_n=\tint\res_zz^{n-1}\kappa(a(z)\mid h(z))\,\in\quot{\mc W}{\partial\mc W}\,.
$$
The corresponding generalized KdV hierarchy of Hamiltonian equations is
$\frac{dp}{dt_n}=\{\tint\mc H_n,\tint p\}_{0,\rho}$, $n\in\mb Z_+$, $p\in\mc W$.

\bigskip

Since the original paper of Drinfeld and Sokolov, 
the construction of the classical $\mc W$-algebras
has been generalized by many authors to the case when $f\in\mf g$ is an arbitrary nilpotent element.
In the framework of Poisson vertex algebras,
they have been constructed in \cite{DSK06}.
In \cite{dGHM92,BdGHM93,FGMS95,FGMS96}
they constructed the corresponding generalized KdV hierarchies,
starting with a Heisenberg subalgebra $\mc H\subset\mf g((z^{-1}))$.
In this approach, they cover all classical $\mc W$-algebras associated 
to nilpotent elements $f\in\mf g$,
for which there exists a graded semisimple element of the form $f+zs\in\mc H$
(the existence of such a graded semisimple element is also studied, in the regular, or ``type I'', case, 
in \cite{FHM92,DF95}, using results in \cite{KP85}, and, for $\mf g$ of type $A_n$,
in \cite{FGMS95,FGMS96}).

\bigskip

In \cite{BDSK09} the theory of Hamiltonian equations and integrable bi-Hamiltonian hierarchies
has been naturally related to the theory of Poisson vertex algebras.

Recall that a \emph{Poisson vertex algebra} (PVA) is a differential algebra $\mc V$,
with a derivation $\partial$,
endowed with a $\lambda$-bracket 
$\{\cdot\,_\lambda\,\cdot\}:\,\mc V\otimes\mc V\to\mb F[\lambda]\otimes\mc V$
satisfying sesquilinearity \eqref{sesqui}, left and right Leibniz rules \eqref{lleibniz}-\eqref{rleibniz},
skew-symmetry \eqref{skewsim}, and Jacobi identity \eqref{jacobi},
displayed in Section \ref{sec:pva}.
Given a PVA structure on an algebra $\mc V$ of smooth functions $u:\,S^1\to\mb R$,
or, in a more algebraic context, on an algebra of differential polynomials $\mc V$
over a field $\mb F$ of characteristics zero,
and a local Hamiltonian functional $\tint h\in\quot{\mc V}{\partial\mc V}$,
the corresponding \emph{Hamiltonian equation} is
\begin{equation}\label{eq:hameq}
\frac{du}{dt}=\{h_\lambda u\}\big|_{\lambda=0}\,.
\end{equation}
An \emph{integral of motion} for such evolution equation 
is a local functional $\tint g\in\quot{\mc V}{\partial\mc V}$
such that 
$$
\{\tint h,\tint g\}:=\tint\{h_\lambda g\}\big|_{\lambda=0}=0\,.
$$
Equation \eqref{eq:hameq} is said to be \emph{integrable}
if there exists an infinite sequence $\tint h_0=\tint h,\,\tint h_1,\,\tint h_2,\dots$
of integrals of motion in involution: 
$\{\tint h_m,\tint h_n\}=0$, for all $m,n\in\mb Z_+$,
which span an infinite dimensional subspace of $\quot{\mc V}{\partial\mc V}$.

The main tool to construct an infinite hierarchy of Hamiltonian equations
is the so called Lenard-Magri scheme (see \cite{Mag78}).
This scheme can be applied to a bi-Hamiltonian equation, that is an evolution equation
which can be written in two compatible Hamiltonian forms:
\begin{equation}\label{eq:bihameq}
\frac{du}{dt}=\{{h_0}_\lambda u\}_H\big|_{\lambda=0}=\{{h_1}_\lambda u\}_K\big|_{\lambda=0}\,,
\end{equation}
where $\{\cdot\,_\lambda\,\cdot\}_H$ and $\{\cdot\,_\lambda\,\cdot\}_K$
are compatible $\lambda$-brackets,
in the sense that any their linear combination defines a PVA structure on $\mc V$.
In this case, under some additional conditions, 
one can solve the recurrence equation
$$
\{{h_n}_\lambda u\}_H=\{{h_{n+1}}_\lambda u\}_K,
\qquad n\in\mb Z_+\,.
$$
Then, according to the Magri Theorem, the local functionals $\tint h_n,\,n\in\mb Z_+$,
are in involution,
so that equation \eqref{eq:bihameq} is integrable, provided that the $\tint h_n$'s
span an infinite-dimensional vector space.

\bigskip

The main aim of the present paper is to derive the Drinfeld-Sokolov construction of classical
$\mc W$-algebras and generalized KdV hierarchies,
as well as the generalizations 
mentioned above, within the context of Poisson vertex algebras.

In fact, it appears clear from the results in Section \ref{sec:dsred},
that Poisson vertex algebras provide the most natural framework 
to describe classical $\mc W$-algebras and the corresponding 
generalized Drinfeld-Sokolov Hamiltonian reduction.
In particular, Theorem \ref{20120511:thm1} (and the following Remark \ref{ref:confgroup}) 
shows that the action of the gauge group $N$
on the phase space $\widetilde{\mc M}$ coincides with an action of a ``Lie conformal group''
on the space $\widetilde{\mc W}$ of functions on $\widetilde{\mc M}$,
obtained by exponentiating the natural Lie conformal algebra action 
of $\mb F[\partial]\mf n$ on $\widetilde{\mc W}$,
where $\mf n$ is a certain subalgebra of $\mf n_+$.

\bigskip

The paper is organized as follows.
In Section \ref{sec:pva} we review, following \cite{BDSK09}, 
the basic definitions and notations of Poisson vertex algebra theory
and its application to the theory of integrable bi-Hamiltonian equations.

In Section \ref{sec:hom}, following the original ideas of Drinfeld and Sokolov,
we show how to apply the Lenard-Magri scheme of integrability 
for the affine Poisson vertex algebra $\mc V(\mf g)$ (defined in Example \ref{affinePVA}),
where $\mf g$ is a reductive Lie algebra.
This is known as the homogeneous case, as it corresponds to the choice $f=0$.
The main result here is Corollary \ref{cor:homog},
which provides an integrable hierarchy of Hamiltonian equations 
in $\mc V(\mf g)$,
associated to a regular semisimple element $s\in\mf g$ 
and an element $a\in Z(\Ker(\ad s))\backslash Z(\mf g)$.
Here and further, for a Lie algebra $\mf a$, $Z(\mf a)$ denotes its center.

Section \ref{sec:dsred} is the heart of the paper.
We first define the action of the Lie conformal algebra $\mb F[\partial]\mf n$ 
on a suitable differential subalgebra $\mc V(\mf p)$ of $\mc V(\mf g)$.
The classical $\mc W$-algebra is then defined as $\mc V(\mf p)^{\mb F[\partial]\mf n}$,
that is the subspace of $\mb F[\partial]\mf n$-invariants in $\mc V(\mf p)$.
As discussed above, in Section \ref{sec:3.3}, we show how
this Lie conformal algebra action is related to the action of the gauge group $N$
on the phase space $\widetilde{\mc M}$,
and we then show how our definition of a classical $\mc W$-algebra is related
to the original definition of Drinfeld and Sokolov.
Then, in Section \ref{sec:3.4}, we use this correspondence
to prove that the classical $\mc W$-algebra is an algebra of differential polynomials in $r=\dim(\Ker\ad f)$
variables, and to give an explicit set of generators for it.

Finally, in Section \ref{sec:non_hom}, following the ideas of Drinfeld and Sokolov,
we apply the Lenard-Magri scheme of integrability 
to derive integrable hierarchies for classical $\mc W$-algebras.
The main result here is Theorem \ref{final},
where we construct an integrable hierarchy of bi-Hamiltonian equations
associated to a nilpotent element $f\in\mf g$,
a homogeneous element $s\in\mf g$ such that $[s,\mf n]=0$ and $f+zs\in\mf g((z^{-1}))$
is semisimple, and an element $a(z)\in Z(\Ker\ad(f+zs))\backslash Z(\mf g((z^{-1}))$.
In Section \ref{sec:4.10} we discuss, in the case of $\mf{gl}_n$,
for which nilpotent elements $f$ Theorem \ref{final} can be applied,
obtaining the same restrictions as in \cite{FGMS95}.

\subsubsection*{Acknowledgments}
Part of the research was conducted at MIT, while the first and third author were visiting the Mathematics Department.
We also thank the Institute Henri Poincar\`e in Paris, where this paper was completed.
%
%
We wish to thank Corrado De Concini for useful discussions
and Andrea Maffei for (always) illuminating observations.
The third author wishes to thank Alessandro D'Andrea for revising his Ph.D. thesis,
from which the present paper originated.

\section{Poisson vertex algebras and Hamiltonian equations}\label{sec:pva}

In this section we review the connection between Poisson vertex algebras
and the theory of Hamiltonian equations as laid down in \cite{BDSK09}.
It is shown that Poisson vertex algebras provide a convenient framework 
for Hamiltonian PDE's
(similar to that of Poisson algebras for Hamiltonian ODE's).
As the main application we explain how to establish integrability
of such partial differential equations using the Lenard-Magri scheme.

\subsection{Algebras of differential polynomials}

By a \emph{differential algebra} we mean a unital commutative associative algebra
$\mc V$ over a field $\mb F$ of characteristic $0$, with a derivation $\partial$,
that is an $\mb F$-linear map from $\mc V$ to itself such that, for $a,b\in \mc V$
$$
\partial(ab)=\partial(a)b+a\partial(b)\,.
$$
In particular $\partial1=0$.

The most important examples we are interested in are
the \emph{algebras of differential polynomials} in the variables $u_1,\ldots,u_{\ell}$:
$$
\mc V=\mb F[u_i^{(n)}\mid i\in I=\{1,\ldots,\ell\},n\in\mb Z_+]\,,
$$
where $\partial$ is the derivation defined by $\partial(u_i^{(n)})=u_i^{(n+1)}$, $i\in I,n\in\mb Z_+$.
Note that we have in $\mc V$ the following commutation relations:
\begin{equation}\label{partialcomm}
\left[\frac{\partial}{\partial u_i^{(n)}},\partial\right]=\frac{\partial}{\partial u_i^{(n-1)}}\,,
\end{equation}
where the RHS is considered to be zero if $n=0$,
which can be written equivalently
in terms of generating series, as follows
\begin{equation}\label{partialcomm2}
\sum_{n\in\mb Z_+}z^n\frac{\partial(\partial g)}{\partial u_i^{(n)}}
=
(\partial+z) \sum_{n\in\mb Z_+} z^n\frac{\partial g}{\partial u_i^{(n)}}
\,,\,\,g\in\mc V
\,.
\end{equation}
We say that $f\in\mc V\backslash\mb F$ has \emph{differential order} $m\in\mb Z_+$ if
$\frac{\partial f}{\partial u_i^{(m)}}\neq0$ for some $i\in I$, 
and $\frac{\partial f}{\partial u_j^{(n)}}=0$ for all $j\in I$ and $n>m$.

The \emph{variational derivative} of $f\in\mc V$ with respect to $u_i$ is, by definition,
$$
\frac{\delta f}{\delta u_i}=\sum_{n\in\mb Z_+}(-\partial)^n\frac{\partial f}{\partial u_i^{(n)}}\,.
$$
It is immediate to check, using \eqref{partialcomm2}, that $\frac{\delta}{\delta u_i}\circ\partial=0$ for every $i$.
In the algebra $\mc V$ of differential polynomials the converse is true too:
if $\frac{\delta f}{\delta u_i}=0$ for every $i=1,\dots,\ell$,
then necessarily $f$ lies in $\partial \mc V\oplus\mb F$
(see e.g. \cite{BDSK09}). 
Letting $U=\bigoplus_{i\in I}\mb F u_i$ be the generating space of $\mc V$,
we define the \emph{variational derivative} of $f\in\mc V$ as
\begin{equation}\label{eq:varder}
\frac{\delta f}{\delta u}=
\sum_{i\in I}u_i\otimes\frac{\delta f}{\delta u_i}\in U\otimes\mc V\,.
\end{equation}


\subsection{Poisson vertex algebras}

\begin{definition}
Let $\mc V$ be a differential algebra. A \emph{$\lambda$-bracket} on $\mc V$ is an $\mb F$-linear map
$\mc V\otimes\mc V\to\mb F[\lambda]\otimes\mc V$,
denoted by $f\otimes g\to\{f_\lambda g\}$,
satisfying \emph{sesquilinearity} ($f,g\in\mc V$):
\begin{equation}\label{sesqui}
\{\partial f_\lambda g\}=
-\lambda\{f_\lambda g\},\qquad 
\{f_\lambda\partial g\}=
(\lambda+\partial)\{f_\lambda g\}\,,
\end{equation} 
and the \emph{left and right Leibniz rules} ($f,g,h\in\mc V$):
\begin{align}\label{lleibniz}
\{f_\lambda gh\}&=
\{f_\lambda g\}h+\{f_\lambda h\}g,\\
\{fh_\lambda g\}&=
\{f_{\lambda+\partial}g\}_{\rightarrow}h+\{h_{\lambda+\partial}g\}_{\rightarrow}f\,,\label{rleibniz}
\end{align}
where we use the following notation: if
$\{f_\lambda g\}=\sum_{n\in\mb Z_+}\lambda^n c_n$,
then
$\{f_{\lambda+\partial}g\}_{\rightarrow}h=\sum_{n\in\mb Z_+}c_n(\lambda+\partial)^nh$.
We say that the $\lambda$-bracket is \emph{skew-symmetric} if
\begin{equation}\label{skewsim}
\{g_\lambda f\}=-\{f_{-\lambda-\partial}g\}\,,
\end{equation}
where, now, 
$\{f_{-\lambda-\partial}g\}=\sum_{n\in\mb Z_+}(-\lambda-\partial)^nc_n$
(if there is no arrow we move $\partial$ to the left).
\end{definition}

\begin{definition}
A \emph{Poisson vertex algebra} (PVA) is a differential algebra $\mc V$ endowed 
with a $\lambda$-bracket which is skew-symmetric and satisfies 
the following \emph{Jacobi identity} in $\mc V[\lambda,\mu]$ ($f,g,h\in\mc V$):
\begin{equation}\label{jacobi}
\{f_\lambda\{g_\mu h\}\}=
\{\{f_\lambda g\}_{\lambda+\mu}h\}+
\{g_\mu\{f_\lambda h\}\}\,.
\end{equation}
\end{definition}

In this paper we consider PVA structures on an algebra 
of differential polynomials $\mc V$ in the variables $\{u_i\}_{i\in I}$.
In this case, thanks to sesquilinearity and Leibniz rules,
the $\lambda$-brackets $\{u_i{}_{\lambda}u_j\},\,i,j\in I$,
completely determine 
the $\lambda$-bracket on the whole algebra $\mc V$. 

\begin{theorem}[{\cite[Theorem 1.15]{BDSK09}}]\label{master}
Let $\mc V$ be an algebra of differential polynomials in the variables $\{u_i\}_{i\in I}$,
and let $H_{ij}(\lambda)\in\mb F[\lambda]\otimes\mc V,\,i,j\in I$.
\begin{enumerate}[(a)]
\item 
The Master Formula
\begin{equation}\label{masterformula}
\{f_\lambda g\}=
\sum_{\substack{i,j\in I\\m,n\in\mb Z_+}}\frac{\partial g}{\partial u_j^{(n)}}(\lambda+\partial)^n
H_{ji}(\lambda+\partial)(-\lambda-\partial)^m\frac{\partial f}{\partial u_i^{(m)}}
\end{equation}
defines the $\lambda$-bracket on $\mc V$
with given $\{u_i{}_\lambda u_j\}=H_{ji}(\lambda),\,i,j\in I$.
\item 
The $\lambda$-bracket \eqref{masterformula} on $\mc V$ satisfies the skew-symmetry 
condition \eqref{skewsim} provided that the same holds on generators ($i,j\in I$):
\begin{equation}\label{skewsimgen}
\{u_i{}_\lambda u_j\}=-\{u_j{}_{-\lambda-\partial}u_i\}\,.
\end{equation}
\item 
Assuming that the skew-symmetry condition \eqref{skewsimgen} holds, 
the $\lambda$-bracket \eqref{masterformula} satisfies the Jacobi identity \eqref{jacobi}, 
thus making $\mc V$ a PVA, provided that the Jacobi identity holds 
on any triple of generators ($i,j,k\in I$):
$$
\{u_i{}_\lambda\{u_j{}_\mu u_k\}\}=
\{\{u_i{}_\lambda u_j\}_{\lambda+\mu}u_k\}+\{u_j{}_\mu\{u_i{}_\lambda u_k\}\}\,.
$$
\end{enumerate}
\end{theorem}

\begin{example}\label{affinePVA}
Let $\mf g$ be a Lie algebra over $\mb F$ with a symmetric invariant bilinear form $\kappa$, 
and let $s$ be an element of $\mf g$.
The \emph{affine PVA} $\mc V(\mf g,\kappa,s)$, associated to the triple $(\mf g,\kappa,s)$,
is the algebra of differential polynomials $\mc V=S(\mb F[\partial]\mf g)$ 
(where $\mb F[\partial]\mf g$ is the free $\mb F[\partial]$-module generated by $\mf g$
and $S(R)$ denotes the symmetric algebra over the $\mb F$-vector space $R$)
together with the $\lambda$-bracket given by
\begin{equation}\label{affinebracket}
\{a_\lambda b\}=[a,b]+\kappa(s\mid[a,b])+\kappa(a\mid b)\lambda
\qquad\text{ for } a,b\in\mf g
\,,
\end{equation}
and extended to $\mc V$ by sesquilinearity and the left and right Leibniz rules.
\end{example}

In Section \ref{sec:dsred} we will define classical $\mc W$-algebras in terms 
of representations of Lie conformal algebras.
Let us recall here some definitions \cite{Kac98}.
\begin{definition}\phantomsection\label{def:lca}
\begin{enumerate}[(a)]
\item
A \emph{Lie conformal algebra} is an $\mb F[\partial]$-module $R$
with an $\mb F$-linear map  $\{\cdot\,_\lambda\,\cdot\}:\,R\otimes R\to\mb F[\lambda]\otimes R$ 
satisfying \eqref{sesqui}, \eqref{skewsim} and \eqref{jacobi}.
\item
A \emph{representation} of a Lie conformal algebra $R$ on an $\mb F[\partial]$-module $\mc V$
is a $\lambda$-action $R\otimes\mc V\to\mb F[\lambda]\otimes\mc V$,
denoted 
$a\otimes g\mapsto a\,^\rho_\lambda\,g$,
satisfying sesquilinearity,
$(\partial a)\,^\rho_\lambda\,g=-\lambda a\,^\rho_\lambda\,g,\,
a\,^\rho_\lambda\,(\partial g)=(\lambda+\partial)a\,^\rho_\lambda\,g$,
and Jacobi identity 
$a\,^\rho_\lambda\,(b\,^\rho_\mu\,g)-b\,^\rho_\mu\,(a\,^\rho_\lambda\,g)
=\{a_\lambda b\}\,^\rho_{\lambda+\mu}\,g$
($a,b\in R$, $g\in\mc V$).
\item
If, moreover, $\mc V$ is a differential algebra,
we say the the action of $R$ on $\mc V$ is by \emph{conformal derivations} if
$a\,^\rho_\lambda\,(gh)=(a\,^\rho_\lambda\,g)h+(a\,^\rho_\lambda\,h)g$.
\end{enumerate}
\end{definition}

\subsection{Poisson structures and Hamiltonian equations}

By Theorem \ref{master}(a), if $\mc V$ is an algebra of differential polynomials
in the variables $\{u_i\}_{i\in I}$, there is a bijective correspondence between
$\ell\times\ell$-matrices 
$H(\lambda)=
\left(H_{ij}(\lambda)\right)_{i,j\in I}\in\Mat_{\ell\times \ell}\mc V[\lambda]$
and the $\lambda$-brackets $\{\cdot\,_\lambda\,\cdot\}_H$ on $\mc V$.

Let $U=\bigoplus_{i\in I}\mb F u_i$ be the generating space of $\mc V$,
and let $\{\chi^i\}_{i\in I}$ be the dual basis of $U^*$.
We have a natural identification
\begin{equation}\label{eq:identif}
\Mat_{\ell\times\ell}\mc V[\lambda]
\hookrightarrow
\Hom(U\otimes\mc V,U^*\otimes\mc V)
\,,
\end{equation}
associating to the matrix $H=\left(H_{ij}(\lambda)\right)_{i,j\in I}$ 
the linear map 
$U\otimes\mc V\to U^*\otimes\mc V$ given by
$$
a\otimes f\mapsto
\sum_{i,j\in I}\chi^j(a) \chi^i\otimes H_{ij}(\partial)f\,.
$$
By an abuse of notation, from now on we will denote by the same letter
an element $H\in\Mat_{\ell\times\ell}\mc V[\lambda]$
or the corresponding linear map $H:\,U\otimes\mc V\to U^*\otimes\mc V$.

\begin{definition}\label{hamop}
A \emph{Poisson structure} on $\mc V$ is a matrix $H\in\Mat_{\ell\times\ell}\mc V[\lambda]$
such that the corresponding $\lambda$-bracket $\{\cdot\,_\lambda\,\cdot\}_H$,
given by equation \eqref{masterformula},
defines a PVA structure on $\mc V$.
\end{definition}

\begin{example}\label{affineHAMOP}
Consider the affine PVA $\mc V(\mf g,\kappa,s)$ defined in Example \ref{affinePVA}. 
Let $\{u_i\}_{i\in I}$ be a basis of $\mf g$ and
let $\{\chi^i\}_{i\in I}\subset\mf g^*$ be its dual basis. 
The corresponding Poisson structure 
$H=\left(H_{ij}(\lambda)\right)\in\Mat_{\ell\times\ell}\mc V[\lambda]$ 
to the $\lambda$-bracket defined in \eqref{affinebracket} is given by
$$
H_{ij}(\lambda)=
\{u_j{}_{\lambda}u_i\}=
[u_j,u_i]+\kappa(s\mid[u_j,u_i])+\kappa(u_i\mid u_j)\lambda\,.
$$
Via the identifications \eqref{eq:identif}, $H$ 
corresponds to the linear map 
$\mf g\otimes\mc V\to\mf g^*\otimes\mc V$ given by
\begin{equation}\label{H3}
H(a\otimes f)
=\sum_{i\in I}\chi^i\otimes[a,u_i]f+\kappa([s,a]\mid\cdot)\otimes f
+\kappa(a\mid\cdot)\otimes\partial f\,.
\end{equation}
In the special case when $\kappa$ is non-degenerate, we can identify 
$\mf g^*\stackrel{\sim}{\longrightarrow}\mf g$ via the isomorphism 
$\kappa(a\mid\cdot)\mapsto a$. 
Let $\{u^i\}_{i\in I}\subset\mf g$ be the dual basis of $\mf g$ with respect to $\kappa$:
$\kappa(u^i\mid u_j)=\delta_{ij},\,i,j\in I$. 
Then,
the map \eqref{H3} is identified with 
the linear map $\mf g\otimes\mc V\to\mf g\otimes\mc V$
given by
\begin{equation}\label{H4}
H(a\otimes f)=
\sum_{i\in I}[u^i,a]\otimes u_i f+[s,a]\otimes f+a\otimes\partial f\,.
\end{equation}
\end{example}

The relation between PVAs and Hamiltonian equations 
associated to a Poisson structure is based on the following simple observation.
\begin{proposition}\label{pvahamop}
Let $\mc V$ be a PVA. The $0$-th product on $\mc V$
induces a well defined Lie algebra bracket on the quotient space $\quot{\mc V}{\partial\mc V}$:
\begin{equation}\label{lambda=0}
\{\tint f,\tint g\}=\tint \left.\{f_\lambda g\}\right|_{\lambda=0}\,,
\end{equation}
where $\tint:\mc V\to\quot{\mc V}{\partial\mc V}$ is the canonical quotient map.
Moreover, we have a well defined Lie algebra action of $\quot{\mc V}{\partial\mc V}$ on $\mc V$
by derivations of the commutative associative product on $\mc V$, commuting with $\partial$,
given by 
$$
\{\tint f,g\}=\{f_\lambda g\}|_{\lambda=0}\,.
$$
\end{proposition}
In the special case when $\mc V$ is an algebra of differential polynomials in $\ell$ variables $\{u_i\}_{i\in I}$
and the PVA $\lambda$-bracket on $\mc V$ is associated to the Poisson structure 
$H\in\Mat_{\ell\times\ell}\mc V[\lambda]$,
the Lie bracket \eqref{lambda=0} on $\quot{\mc V}{\partial\mc V}$ takes the form 
(cf. \eqref{masterformula}):
\begin{equation}\label{liebrak}
\{\tint f,\tint g\}
=\sum_{i,j\in I}\int\frac{\delta g}{\delta u_j}H_{ji}(\partial)\frac{\delta f}{\delta u_i}
\,.
\end{equation}

\begin{definition}\label{hamsys}
Let $\mc V$ be an algebra of differential polynomials
with a Poisson structure $H$.
\begin{enumerate}[(a)]
\item 
Elements of $\quot{\mc V}{\partial\mc V}$ are called \emph{local functionals}. 
\item
Given a local functional $\int h\in\quot{\mc V}{\partial\mc V}$, 
the corresponding \emph{Hamiltonian equation} is
\begin{equation}\label{hameq}
\frac{du}{dt}
=\{\tint h,u\}_H
\qquad
\Big(\text{equivalently,}\quad
\frac{du_i}{dt}
=\sum_{j\in I}H_{ij}(\partial)\frac{\delta h}{\delta u_j},\ i\in I\Big)\,.
\end{equation}
\item 
A local functional $\int f\in\quot{\mc V}{\partial\mc V}$ is called an \emph{integral of motion} 
of equation \eqref{hameq} if $\frac{df}{dt}=0\mod\partial\mc V$
in virtue of \eqref{hameq}, or, equivalently, if $\tint h$ and $\tint f$ are \emph{in involution}:
$$
\{\tint h,\tint f\}_H=0\,.
$$
Namely, $\tint f$ lies in the centralizer of $\tint h$ in the Lie algebra $\quot{\mc V}{\partial\mc V}$
with Lie bracket \eqref{liebrak}.
\item 
Equation \eqref{hameq} is called \emph{integrable}
if there exists an infinite sequence $\tint f_0=\tint h,\,\tint f_1,\,\tint f_2,\dots$,
of linearly independent integrals of motion in involution.
The corresponding \emph{integrable hierarchy of Hamiltonian equations} is
\begin{equation}\label{eq:hierarchy}
\frac{du}{dt_n}=\{\tint f_n,u\}_H,\ n\in\mb Z_+
\,.
\end{equation}
(Equivalently,
$\frac{du_i}{dt_n}
=\sum_{j\in I}H_{ij}(\partial)\frac{\delta f_n}{\delta u_j}$, $n\in\mb Z_+,\,i\in I$)

\end{enumerate}
\end{definition}

\subsection{Bi-Poisson structures and integrability of Hamiltonian equations}\label{sec:1.3}

\begin{definition}
A \emph{bi-Poisson structure} $(H,K)$ is a pair of Poisson structures 
on an algebra of differential polynomials $\mc V$, which are compatible,
in the sense that any $\mb F$-linear combination of them is a Poisson structure.
\end{definition}

\begin{example}\label{affineHAMOP2}
Consider Example \ref{affineHAMOP}, with non-degenerate $\kappa$.
We identify $\mf g^*\stackrel{\sim}{\longrightarrow}\mf g$
via $\kappa(a\mid\cdot)\mapsto a$
and we describe a Poisson structure on $\mc V=\mc V(\mf g,\kappa,s)$ 
as an element of $\End(\mf g\otimes\mc V)$
via the identification \eqref{eq:identif}.
Then, the maps
$H,K:\,\mf g\otimes\mc V\to\mf g\otimes\mc V$
given by
\begin{equation}\label{H5}
H(a\otimes f)=
\sum_{i\in I}[u^i,a]\otimes u_i f+a\otimes\partial f,
\qquad
K(a\otimes f)=
[a,s]\otimes f\,,
\end{equation}
form a bi-Poisson structure on $\mc V$.
Comparing \eqref{H4} and \eqref{H5}, we get that the Poisson structure 
of Example \ref{affineHAMOP} is equal to $H-K$.
\end{example}

Let $\mc V$ be an algebra of differential polynomials,
with the generating space $U\subset\mc V$,
and we let $(H,K)$ be a bi-Poisson structure on $\mc V$.
According to the \emph{Lenard-Magri scheme of integrability}
\cite{Mag78} (see also \cite{BDSK09}),
in order to obtain an integrable hierarchy of Hamiltonian equations,
one needs to find a sequence 
of local functionals $\{\tint f_n\}_{n\in\mb Z_+}$ 
spanning an infinite-dimensional subspace of $\quot{\mc V}{\partial\mc V}$,
such that their variational derivatives
$F_n=\frac{\delta f_n}{\delta u}\in U\otimes\mc V,\,n\in\mb Z_+$, satisfy
\begin{equation}\label{lenardrecursion}
K(F_0)=0,
\qquad
H(F_n)=K(F_{n+1})
\,\,\Big(\in U^*\otimes\mc V\Big)\quad\text{ for every } n\in\mb Z_+\,.
\end{equation}
If this is the case, the elements $\tint f_n,\,n\in\mb Z_+$,
form an infinite sequence of local functionals in involution:
$\{\tint f_m,\tint f_n\}_H=\{\tint f_m,\tint f_n\}_K=0$, for all $m,n\in\mb Z_+$.
Hence, we get the corresponding integrable hierarchy of Hamiltonian equations \eqref{eq:hierarchy},
provided that the span of $\tint f_n$'s is infinite dimensional.

We note that, using the generating series
$F(z)=\sum_{n\in\mb Z_+}z^{-n}\frac{\delta f_n}{\delta u}\in (U\otimes\mc V)[[z^{-1}]]$,
we can rewrite the Lenard-Magri recursion \eqref{lenardrecursion} as:
\begin{equation}\label{lenardseries}
K(F_0)=0,\qquad
(H-zK)F(z)=0
\,.
\end{equation}

\section{Drinfeld-Sokolov hierarchies in the homogeneous case}\label{sec:hom}

As the first application of the Lenard-Magri scheme,
we construct integrable hierarchies of Hamiltonian equations \eqref{hameq}
for the bi-Poisson structure provided by Example \ref{affineHAMOP2}.
This is referred to as the homogeneous Drinfeld-Sokolov hierarchy.
The non-homogeneous case will be treated in the next sections, 
after giving the definition of classical $\mc W$-algebras.

Let $\mf g$ be a finite-dimensional Lie algebra 
with a non-degenerate symmetric invariant bilinear form $\kappa$,
let $s$ be an element of $\mf g$,
and let $\mc V=S(\mb F[\partial]\mf g)$.
Recall from Example \ref{affineHAMOP2} that
we have a bi-Poisson structure $H,K:\,\mf g\otimes\mc V\to\mf g\otimes\mc V$
on $\mc V$, given by \eqref{H5}
(as before, we are using the identification \eqref{eq:identif} for Poisson structures,
and the isomorphism $\mf g^*\stackrel{\sim}{\longrightarrow}\mf g$ associated
to the bilinear form $\kappa$).

We endow the space $\mf g\otimes\mc V$ with a Lie algebra structure letting,
for $a,b\in\mf g$ and $f,g\in\mc V$,
$$
[a\otimes f,b\otimes g]=[a,b]\otimes fg\in\mf g\otimes\mc V\,.
$$ 
We extend the bilinear form $\kappa$ of $\mf g$ to a bilinear map 
$\kappa:(\mf g\otimes\mc V)\times(\mf g\otimes\mc V)\to\mc V$, 
given by ($a,b\in\mf g,\,f,g\in\mc V$):
\begin{equation}\label{eq:forma}
\kappa(a\otimes f\mid b\otimes g)=\kappa(a\mid b)fg\in\mc V\,.
\end{equation}
Clearly, this bilinear map $\kappa$ is symmetric invariant and non-degenerate.
We extend $\partial\in\Der(\mc V)$ to a derivation of the Lie algebra $\mf g\otimes\mc V$
by
$$
\partial(a\otimes f)=a\otimes\partial f\,,
$$
for any $a\in\mf g$ and $f\in\mc V$. 
Thus we get the semidirect product Lie algebra $\mb F\partial\ltimes(\mf g\otimes\mc V)$, 
where
$[\partial,a\otimes f]=\partial(a\otimes f)=a\otimes\partial f$ for $a\in\mf g$ and $f\in\mc V$.

We set $\widetilde{\mf g}=(\mb F\partial\ltimes(\mf g\otimes\mc V))((z^{-1}))$,
the space of Laurent series in $z^{-1}$ with coefficients in $\mb F\partial\ltimes(\mf g\otimes\mc V)$,
endowed with the Lie algebra structure induced by that on $\mb F\partial\ltimes(\mf g\otimes\mc V)$.
We note that
$(\mf g\otimes\mc V)[[z^{-1}]]\subset\widetilde{\mf g}$
is a Lie subalgebra.
\begin{proposition}\label{int_hier1}
Let $L(z)=\partial+u+zs\otimes1\in\widetilde{\mf g}$,
where $u=\sum_{i\in I}u^i\otimes u_i\in\mf g\otimes\mc V$.
Then
$$
(H-zK)(a\otimes f)=[L(z),a\otimes f]\,,
$$
for any $a\in\mf g$ and $f\in\mc V$.
\end{proposition}
\begin{proof}
It follows immediately by \eqref{H5} and the definition of the Lie bracket on $\widetilde{\mf g}$.
\end{proof}

Recall from Section \ref{sec:1.3} that, in order to construct an integrable hierarchy 
of Hamiltonian equations,
we need to find $\tint f(z)\in(\quot{\mc V}{\partial\mc V})[[z^{-1}]]$
such that $F(z)=\frac{\delta f(z)}{\delta u}\in(\mf g\otimes\mc V)[[z^{-1}]]$
is a solution of equation \eqref{lenardseries}.
Hence, using Proposition \ref{int_hier1}, we conclude that 
the Lenard-Magri scheme can be applied if there exists
$F(z)\in(\mf g\otimes\mc V)[[z^{-1}]]$ satisfying the following three conditions:
\begin{enumerate}[(C1)]
\item
$[s\otimes1,F_0]=0$,
\item
$[L(z),F(z)]=0$.
\item
$F(z)=\frac{\delta f(z)}{\delta u}$, for some $\tint f(z)\in(\quot{\mc V}{\partial\mc V})[[z^{-1}]]$.
\end{enumerate}
The solution of the above problem will be achieved in Propositions \ref{int_hier2} and \ref{int_hier3}
below (which are essentially due to Drinfeld and Sokolov \cite{DS85}), 
under the assumption that $s\in\mf g$ is a semisimple element: 
in Proposition \ref{int_hier2} we find an element $F(z)\in (\mf g\otimes\mc V)[[z^{-1}]]$ 
satisfying conditions (C1) and (C2),
and in Proposition \ref{int_hier3} we show that this element satisfies condition (C3).

Before stating the results we need to introduce some notation.
For $U(z)\in (\mf g\otimes\mc V)[[z^{-1}]] z^{-1}$, 
we have a well-defined Lie algebra automorphism
$$
e^{\ad U(z)}:\widetilde{\mf g}\to\widetilde{\mf g}\,.
$$
By the Baker-Campbell-Hausdorff formula \cite{Ser92}, 
automorphisms of this type form a group.
Fix a semisimple element $s\in\mf g$, and denote $\mf h=\Ker(\ad s)\subset\mf g$
(it is clearly a subalgebra);
by invariance of the bilinear form $\kappa$ we have that $\mf h^{\perp}=\im(\ad s)$,
and that $\mf g=\Ker(\ad s)\oplus\im(\ad s)$.
%
\begin{remark}\label{rem:semisimple}
We can replace the assumption that $s\in\mf g$ is semisimple by the assumption that
$\mf g$ admits a vector space decomposition $\mf g=\Ker(\ad s)\oplus\im(\ad s)$.
It is not hard to show that then $s$ is semisimple, provided that $\mf g$
is a reductive Lie algebra.
\end{remark}

\begin{proposition}\phantomsection\label{int_hier2}
\begin{enumerate}[(a)]
\item
There exist unique formal series $U(z)\in (\mf{h}^\perp\otimes\mc V)[[z^{-1}]]z^{-1}$ 
and $h(z)\in(\mf{h}\otimes\mc V)[[z^{-1}]]$ such that
\begin{equation}\label{L0}
L_0(z)=e^{\ad U(z)}(L(z))=\partial+zs\otimes1+h(z)\,.
\end{equation}
\item
An automorphism $e^{\ad U(z)}\in (\mf{g}\otimes\mc V)[[z^{-1}]]z^{-1}$ 
solving \eqref{L0} for some $h(z)\in(\mf{h}\otimes\mc V)[[z^{-1}]]$
is defined uniquely up to multiplication on the left 
by automorphisms of the form $e^{\ad S(z)}$, 
where $S(z)\in (\mf h\otimes\mc V)[[z^{-1}]]z^{-1}$.
\item
Let $a\in Z(\mf h)$ (the center of $\mf h$),
and let $U(z)\in (\mf{g}\otimes\mc )V[[z^{-1}]]z^{-1}$ 
and $h(z)\in(\mf{h}\otimes\mc V)[[z^{-1}]]$ 
solve equation \eqref{L0}.
Then 
\begin{equation}\label{eq:F}
F(z)=e^{-\ad U(z)}(a\otimes1)\in(\mf g\otimes\mc V)[[z^{-1}]]
\end{equation}
is independent on the choice of $U(z)$
and it solves equations (C1) and (C2) above: 
$[s\otimes1,F_0]=0$ and $[L(z),F(z)]=0$.
\end{enumerate}
\end{proposition}
\begin{proof}
Writing $U(z)=\sum_{i\geq1}U_iz^{-i}$ and $h(z)=\sum_{i\in\mb Z_+}h_iz^{-i}$, 
with $U_i,h_i\in\mf g\otimes\mc V$, 
and equating coefficients of $z^{-i}$ in both sides of \eqref{L0}, we find an equation of the form 
$$
h_i+[s\otimes1,U_{i+1}]=A\,,
$$ 
where $A\in\mf g\otimes\mc V$ is expressed, inductively, 
in terms of the elements $U_1,U_2,\ldots, U_i$ and $h_0,h_1,\ldots,h_{i-1}$. 
For example, equating the constant term in \eqref{L0} gives the relation $h_0+[s\otimes1,U_1]=u$, 
while, equating the coefficients of $z^{-1}$, gives the relation 
$h_1+[s\otimes1,U_2]=-U_1^{\prime}+[U_1,u]+\frac{1}{2}[U_1,[U_1,s\otimes1]]$.
Decomposing $A=A_{\mf h}+A_{\mf h^\perp}$, according to $\mf g=\mf h\oplus\mf h^\perp$, 
we get $h_i=A_{\mf h}$ and $[s\otimes1,U_{i+1}]=A_{\mf h^{\perp}}$,
which uniquely defines $U_{i+1}\in\mf h^{\perp}\otimes\mc V$, proving (a).

Let $U(z)\in (\mf{h}^\perp\otimes\mc V)[[z^{-1}]]z^{-1},\,h(z)\in(\mf h\otimes\mc V)[[z^{-1}]]$ 
be the unique solution of \eqref{L0} given by part (a),
and let $\widetilde U(z)\in (\mf g\otimes\mc V)[[z^{-1}]]z^{-1},\,
\widetilde h(z)\in(\mf h\otimes\mc V)[[z^{-1}]]$ 
be some other solution of \eqref{L0}:
$e^{\ad \widetilde{U}(z)}(L(z))=\partial+zs\otimes1+\widetilde{h}(z)$.
By the observation before the statement of the proposition,
there exists $S(z)=\sum_{i=1}^\infty S_iz^{-i}\in(\mf g\otimes\mc V)[[z^{-1}]]z^{-1}$ 
such that
\begin{equation}\label{eq:ciao}
e^{\ad\widetilde U(z)}=e^{\ad S(z)}e^{\ad U(z)}\,.
\end{equation}
To conclude the proof of (b), we need to show that $S(z)\in(\mf h\otimes\mc V)[[z^{-1}]]z^{-1}$.
Applying equation \eqref{eq:ciao} to $L(z)$ we get
\begin{equation}\label{eq:ciao1}
\partial+zs\otimes1+\widetilde{h}(z)=e^{\ad S(z)}(\partial+zs\otimes1+h(z))\,.
\end{equation}
Comparing the constant terms in $z$ of both sides of \eqref{eq:ciao1}, we get
$\widetilde{h}_0=h_0+[s\otimes1,S_1]$,
which clearly implies $\widetilde{h}_0=h_0$ and $S_1\in\mf h\otimes\mc V$.
Assuming by induction that $S_1,\dots,S_i$ lie in $\mf h\otimes\mc V$,
and comparing the coefficients of $z^{-i}$ in both sides of \eqref{eq:ciao1}
we get $[s\otimes1,S_{i+1}]\in(\mf h^\perp\cap\mf h)\otimes\mc V=0$,
so that $S_{i+1}\in\mf h\otimes\mc V$, as desired.

To prove part (c),
note that, by part (b), 
$e^{-\ad\widetilde{U}(z)}(a\otimes1)=e^{-\ad U(z)}e^{-\ad S(z)}(a\otimes1)=F(z)$,
since, by construction $S(z)$ lies in $(\mf h\otimes\mc V)[[z^{-1}]]z^{-1}$ (hence, it commutes
with $a\otimes1\in Z(\mf h)\otimes\mc V$).
Moreover, $F_0=a\otimes1$, hence it commutes with $[s\otimes1]$, proving condition (C1).
Finally, condition (C2) follows from the facts that $[L_0(z),a\otimes1]=0$
and $e^{-\ad U(z)}$ is a Lie algebra automorphism of $\widetilde{\mf g}$.
\end{proof}

\begin{proposition}\label{int_hier3}
Let $U(z)\in(\mf g\otimes\mc V)[[z^{-1}]]z^{-1}$ and $h(z)\in(\mf h\otimes\mc V)[[z^{-1}]]$
be a solution of equation \eqref{L0}.
Then the formal power series $F(z)\in(\mf g\otimes\mc V)[[z^{-1}]]$ defined in \eqref{eq:F}
(where $a\in Z(\mf h)$)
satisfies condition (C3) above,
namely $F(z)=\frac{\delta f(z)}{\delta u}$, where
\begin{equation}\label{var_hier}
\tint f(z)=\tint\kappa(a\otimes1\mid h(z))\in\quot{\mc V}{\partial\mc V}[[z^{-1}]]\,.
\end{equation}
\end{proposition}
Before proving Proposition \ref{int_hier3} we introduce some notation and prove two 
preliminary lemmas.
We extend the partial derivatives $\frac{\partial}{\partial u_i^{(m)}}$ to derivations of the Lie algebra
$\mf g\otimes\mc V$ in the obvious way: 
$\frac{\partial}{u_i^{(m)}}(a\otimes f)=a\otimes\frac{\partial f}{u_i^{(m)}}$.
We also define the differential order of elements 
$F\in\mf g\otimes\mc V\backslash\mf g\otimes\mb F$ in the same way as before: 
$F$ has differential order $m\in\mb Z_+$ if
$\frac{\partial F}{\partial u_i^{(m)}}\neq0$ for some $i\in I$, and $\frac{\partial F}{\partial u_j^{(n)}}=0$ for all $j\in I$ and $n>m$.

Note that if $A\in\mf g\otimes\mc V$ and $B\in\mb F\partial\ltimes(\mf g\otimes\mc V)$, 
then $(\ad A)B\in\mf g\otimes\mc V$.
\begin{lemma}\label{lem1}
For $\alpha\in\mb F,\,A,U_1,\dots,U_k\in\mf g\otimes\mc V$, with $k\geq1$, we have
\begin{equation}\label{eq:induttiva}
\begin{array}{l}
\displaystyle{
\frac{\partial}{\partial u_i^{(m)}} \Big(\ad(U_1)\cdots\ad(U_k)(\alpha\partial+A)\Big)
} \\
\displaystyle{
=
\sum_{h=1}^k\ad(U_1)\cdots\ad\Big(\frac{\partial U_h}{\partial u_i^{(m)}}\Big)\cdots\ad(U_k)(\alpha\partial+A) 
} \\
\displaystyle{
+\ad(U_1)\cdots\ad(U_k)\Big(\frac{\partial A}{\partial u_i^{(m)}}\Big)
-\alpha\ad(U_1)\cdots\ad(U_{k-1})\Big(\frac{\partial U_k}{\partial u_i^{(m-1)}}\Big)\,.
}
\end{array}
\end{equation}
\end{lemma}
\begin{proof}
Equation \eqref{eq:induttiva} follows from the fact that 
$\frac{\partial}{\partial u_i^{(m)}}$ is a derivation of $\mf g\otimes\mc V$
and by the commutation rule \eqref{partialcomm}.
\end{proof}
\begin{lemma}\label{lem2}
For $U,V\in\mf g\otimes\mc V$ and $L\in\mb F\partial\ltimes(\mf g\otimes\mc V)$, we have
\begin{equation}\label{eq:commut}
\Big[\!\sum_{h\in\mb Z_+}\!
\frac1{(h+1)!}(\ad U)^h(V),e^{\ad U}(L)\Big]
=
\!\sum_{h,k\in\mb Z_+}\!
\frac1{(h+k+1)!}(\ad U)^h(\ad V)(\ad U)^k(L)
\,.
\end{equation}
\end{lemma}
\begin{proof}
Since $\ad U$ is a derivation of the Lie bracket in $\mb F\partial\ltimes(\mf g\otimes\mc V)$,
we have
$$
\begin{array}{l}
\displaystyle{
RHS\eqref{eq:commut}
=
\sum_{h,k\in\mb Z_+}\!\frac{(\ad U)^h}{(h+k+1)!}\big[V,(\ad U)^k(L)\big]
} \\
\displaystyle{
=
\sum_{h,k\in\mb Z_+}\sum_{l=0}^h\frac{\binom{h}{l}}{(h+k+1)!}\big[(\ad U)^l(V),(\ad U)^{h+k-l}(L)\big]
} \\
\displaystyle{
=
\sum_{m,n\in\mb Z_+}
\frac1{(m+1)!}\frac1{n!}
\big[(\ad U)^m(V),(\ad U)^n(L)\big]
\sum_{h=m}^{m+n}\frac{\binom{h}{m}}{\binom{m+n+1}{m+1}}
\,. }
\end{array}
$$
The RHS above is the same as the LHS of \eqref{eq:commut} thanks to the
simple combinatorial identity
$\sum_{h=m}^{m+n}\binom{h}{m}=\binom{m+n+1}{m+1}$.
\end{proof}
\begin{proof}[Proof of Proposition \ref{int_hier3}]
We need to compute $\frac{\delta f(z)}{\delta u}$.
By the definition \eqref{eq:varder} of the variational derivative
and the definition \eqref{var_hier} of $\tint f(z)$,
we have
\begin{equation}\label{eq:21mar-1}
\begin{array}{l}
\displaystyle{
\frac{\delta f(z)}{\delta u}
=\sum_{i\in I,m\in\mb Z_+}
u_i\otimes(-\partial)^m
\kappa\Big(a\otimes1\,\Big|\,\frac{\partial h(z)}{\partial u_i^{(m)}}\Big)
} \\
\displaystyle{
=\sum_{i\in I,m\in\mb Z_+}
u_i\otimes(-\partial)^m
\kappa\Big(
a\otimes1\,\Big|\,\frac{\partial}{\partial u_i^{(m)}}
\Big(
e^{\ad U(z)}(\partial+u+zs\otimes1)-\partial-zs\otimes1
\Big)\Big)
} \\
\displaystyle{
=
a\otimes1
+\!\!\sum_{i\in I,m\in\mb Z_+}\!\!\!
u_i\otimes(-\partial)^m
\kappa\Big(
a\otimes1\,\Big|\,
\sum_{k=1}^\infty\frac1{k!}
\frac{\partial}{\partial u_i^{(m)}}(\ad U(z))^k(\partial+u+zs\otimes1)
\Big)\,.
}
\end{array}
\end{equation}
In the second identity we used the definition \eqref{L0} of $h(z)$,
and in the last identity we used the Taylor series expansion for the exponential $e^{\ad U(z)}$.
By Lemma \ref{lem1}, the last term in the RHS of \eqref{eq:21mar-1} can be rewritten as
\begin{equation}\label{eq:21mar-2}
\begin{array}{l}
\displaystyle{
\sum_{i\in I,m\in\mb Z_+}\!\!\!\!
u_i\otimes(-\partial)^m
\kappa\Bigg(\!
a\otimes1\,\Bigg|\,
\sum_{k=1}^\infty\sum_{h=0}^{k-1}\frac1{k!}
(\!\ad U(z))^h \big(\ad \frac{\partial U(z)}{\partial u_i^{(m)}}\big) (\ad U(z))^{k-h-1}\! L(z)
} \\
\displaystyle{
+
\sum_{k=1}^\infty\frac1{k!}
(\ad U(z))^k\frac{\partial}{\partial u_i^{(m)}}(u+zs\otimes1)
-
\sum_{k\in\mb Z_+}\frac1{(k+1)!}
(\ad U(z))^{k}\frac{\partial U(z)}{\partial u_i^{(m-1)}}
\Bigg)
} \\
\displaystyle{
=
\sum_{k=1}^\infty\frac1{k!}
(-\ad U(z))^k (a\otimes1)
} \\
\displaystyle{
+
\sum_{i\in I,m\in\mb Z_+}
u_i\otimes(-\partial)^m
\kappa\Bigg(
a\otimes1\,\Bigg|\,
-\sum_{k\in\mb Z_+}\frac1{(k+1)!}
(\ad U(z))^{k}\frac{\partial U(z)}{\partial u_i^{(m-1)}}
} \\
\displaystyle{
\,\,\,\,\,\,\,\,\,\,\,\,\,\,\,\,\,\,\,\,\,\,\,\,\,\,\,\,\,\,\,\,\,\,\,\,
+\sum_{h,k\in\mb Z_+}^\infty\frac1{(h+k+1)!}
(\ad U(z))^h \big(\ad \frac{\partial U(z)}{\partial u_i^{(m)}}\big) (\ad U(z))^k L(z)
\Bigg)
\,.
}
\end{array}
\end{equation}
For the first term in the RHS we used the invariance of the bilinear map
$\kappa:\,(\mf g\otimes\mc V)\times(\mf g\otimes\mc V)\to\mc V$.
Combining the first term in the RHS of \eqref{eq:21mar-1} 
and the first term in the RHS of \eqref{eq:21mar-2},
we get $e^{-\ad U(z)}(a\otimes1)$, which is the same as $F(z)$ by \eqref{eq:F}.
Hence, in order to complete the proof of the proposition,
we are left to show that the last two terms in the RHS of \eqref{eq:21mar-2} cancel.
Let
$$
A_{i,m}(z)=\sum_{k\in\mb Z_+}\frac1{(k+1)!}
(\ad U(z))^{k}\frac{\partial U(z)}{\partial u_i^{(m)}}\,.
$$
Using this notation, the second term of the RHS of \eqref{eq:21mar-2} can be rewritten as
\begin{equation}\label{eq:21mar-3}
-\sum_{i\in I,m\in\mb Z_+}
u_i\otimes(-\partial)^m
\kappa\big(
a\otimes1\mid
A_{i,m-1}(z)
\big)\,,
\end{equation}
while, by Lemma \ref{lem2}, the third term of the RHS of \eqref{eq:21mar-2} is
$$
\sum_{i\in I,m\in\mb Z_+}
u_i\otimes(-\partial)^m
\kappa\Big(
a\otimes1\,\Big|\,
\big[A_{i,m}(z),e^{\ad U(z)}L(z)\big]
\Big)
\,.
$$
By equation \eqref{L0}, the invariance of the bilinear map $\kappa$
and the assumption that $a$ lies in the center of $\mf h$,
the above expression is equal to
$$
\sum_{i\in I,m\in\mb Z_+}
u_i\otimes(-\partial)^{m+1}
\kappa\Big(
a\otimes1\,\Big|\, A_{i,m}(z)
\Big)
\,,
$$
which, combined with \eqref{eq:21mar-3}, gives zero.
\end{proof}
\begin{remark}\label{rem:grado2}
Consider the usual polynomial 
grading of the algebra of differential polynomials $\mc V=S(\mb F[\partial]\mf g)$.
We can compute the part of $\tint f(z)\in(\quot{\mc V}{\partial\mc V})[[z^{-1}]]$ of degree less or equal than $2$,
using equations \eqref{L0} and \eqref{var_hier}.
For $n\in\mb Z_+$, we denote by $U(z)(n)$, $h(z)(n)$ and $\tint f(z)(n)$
the homogeneous components of degree $n$ in $U(z)$, $h(z)$ and $\tint f(z)$ respectively.
Using equation \eqref{L0} and the fact that $\mf h\cap\mf h^\perp=0$,
it is easy to show, inductively on the negative powers of $z$, 
that $U(z)(0)=0$ and $h(z)(0)=0$, so that $\tint f(z)(0)=0$.
Similarly, equating the homogeneous components of degree 1 in equation \eqref{L0},
we get
$h(z)(1)=\pi_{\mf h}u$, where $\pi_{\mf h}:\,\mf g\otimes\mc V\to\mf h\otimes\mc V$ 
is the canonical quotient map (with respect to the decomposition $\mf g=\mf h\oplus\mf h^\perp$).
Hence, $\tint f(z)(1)=\tint \kappa(a\otimes 1\mid u)$.
Moreover, $U(z)(1)$ solves the equation
$$
[zs\otimes 1,U(z)(1)]=\pi_{\mf h^\perp}u-U^\prime(z)(1)\,.
$$
More explicitly, the coefficient of $z^{-n}$ in $U(z)(1)$ is
given by $(\ad s)^{-n}(-\partial)^{n-1}\pi_{\mf h^\perp}u$,
where $\ad s$ is considered as an invertible endomorphism of $\mf h^\perp$.
Finally, equating the homogeneous components of degree 2 in \eqref{L0}, we get
$h(z)(2)=\frac12\pi_{\mf h}\big[U(z)(1),u\big]$.
Hence,
$$
\tint f(z)(2)=\frac12\sum_{n=1}^\infty z^{-n}\tint
\kappa\big(a\otimes1\,\mid\,[(\ad s)^{-n}(-\partial)^{n-1}\pi_{\mf h^\perp}u,u]\big)\,.
$$
\end{remark}
\begin{remark}\label{totder}
Let $U(z),h(z)$ and $\widetilde U(z), \widetilde h(z)$ be two solutions of \eqref{L0}.
Recall by Proposition \ref{int_hier2}(b) that
$e^{\ad\widetilde U(z)}=e^{\ad S(z)}e^{\ad U(z)}$ 
for some $S(z)\in(\mf h\otimes\mc V)[[z^{-1}]]z^{-1}$.
By Proposition \ref{int_hier2}(c),
$F(z)=e^{-\ad U(z)}(a\otimes1)=e^{-\ad\widetilde U(z)}(a\otimes1)$.
Hence, by Proposition \ref{int_hier3},
$f(z)=\tint\kappa(a\otimes1\mid h(z))$ and $\widetilde f(z)=\tint\kappa(a\otimes1\mid\widetilde h(z))$
differ by a total derivative.
In particular, if $\mf h$ is abelian 
(this is the case when $s\in\mf g$ is regular semisimple), 
then $h(z)-\widetilde h(z)=\partial S(z)$.
\end{remark}
\begin{corollary}\label{cor:homog}
Let $\mf g$ be a finite-dimensional Lie algebra
with a non-degenerate symmetric invariant bilinear form $\kappa$, and let $s\in\mf g$ be a semisimple element. Let
$\mc V=S(\mb F[\partial]\mf g)=\mb F[u_i^{(n)}\mid
i\in I,n\in\mb Z_+]$ (where $\{u_i\}_{i\in I}\subset\mf g$ is a basis of $\mf g$),
and let us extend $\kappa$ to a bilinear map 
$\kappa:\,\mf g\otimes\mc V\times\mf g\otimes\mc V\to\mc V$.
Let $U(z)\in(\mf g\otimes\mc V)[[z^{-1}]]z^{-1},\,h(z)\in(\mf h\otimes\mc V)[[z^{-1}]]$
be a solution of equation \eqref{L0}, where $\mf h=\Ker(\ad s)$.
Given an element $a\in Z(\mf h)\backslash Z(\mf g)$,
we have an infinite hierarchy of integrable bi-Hamiltonian equations
associated to the bi-Poisson structure $(H,K)$ on $\mc V$,
defined in \eqref{H5}:
$$
\frac{du_i}{dt_n}
=\sum_{j\in I}H_{ij}(\partial)\frac{\delta f_n}{\delta u_j}
=\sum_{j\in I}K_{ij}(\partial)\frac{\delta f_{n+1}}{\delta u_j},\,i\in I, n\in\mb Z_+\,,
$$
where $\tint f_n\in\quot{\mc V}{\partial\mc V}$ is the coefficient of $z^{-n}$ 
in $\tint f(z)=\tint \kappa(a\otimes1\mid h(z))\in(\quot{\mc V}{\partial\mc V})[[z^{-1}]]$.
\end{corollary}
\begin{proof}
By Propositions \ref{int_hier2} and \ref{int_hier3},
the formal power series $F(z)\in(\mf g\otimes\mc V)[[z^{-1}]]$, defined in \eqref{eq:F},
satisfies conditions (C1), (C2) and (C3) above.
Hence, according to the Lenard-Magri scheme, we only need to check that
the local functionals $\tint f_n\in\quot{\mc V}{\partial\mc V},\,n\in\mb Z_+$,
are linearly independent.
This is is obtained by the following simple observations.
By the definitions \eqref{H5} of  $H$ and $K$,
it is clear that, if $F=a\otimes1$ with $a\not\in Z(\mf g)$, 
then $H(F)\neq0$ and it is of differential order $0$;
if $F\in\mf g\otimes\mc V$ has differential order $m\in\mb Z_+$,
then necessarily $H(F)\in\mf g\otimes\mc V$ has differential order $m+1$,
and if $K(F)\in\mf g\otimes\mc V$ has differential order $m$,
then necessarily $F\in\mf g\otimes\mc V$ has differential order at least $m$.
Hence, by the recursion formula \eqref{lenardrecursion}
we immediately get that the elements $F_{n}\in\mf g\otimes\mc V$, $n\in\mb Z_+$, 
have distinct differential orders.
In particular, the elements $F_n=\frac{\delta f_n}{\delta u}\in\mf g\otimes\mc V$
are linearly independent, and therefore the elements $\tint f_n\in\quot{\mc V}{\partial\mc V}$
are linearly independent as well.
\end{proof}

\begin{example}\label{ex:nwave}
\emph{The $N$-wave equation}.
Let $\mf g=\mf{gl}_N$, with the bilinear form $\kappa(A\mid B)=\tr(AB)$, 
and let $s=\diag(s_1,\dots,s_N)$ be a diagonal matrix with distinct eigenvalues.
Then $\mf h=\Ker(\ad s)$ is the abelian subalgebra of diagonal $N\times N$ matrices,
and $\mf h^\perp=\im(\ad s)$ consists of $N\times N$ matrices with zeros along the diagonal.
We also have $u=\sum_{i,j=1}^NE_{ij}\otimes E_{ji}\in\mf{gl}_N\otimes\mc V$,
where $\mc V$ is the algebra of differential polynomials generated by $\mf{gl}_N$.
In this case, for $U(z)\in(\mf{h}^\perp\otimes\mc V)[[z^{-1}]]z^{-1}$,
there exists a unique $T(z)=\sum_{n\in\mb Z_+}T_nz^{-n}\in(\mf{gl}_N\otimes\mc V)[[z^{-1}]]$,
with $T_0=\id_N$ and $T_n\in\mf h^\perp$ for all $n\geq1$,
such that $e^{\ad U(z)}:\,\widetilde{\mf{gl}}_N\to\widetilde{\mf{gl}}_N$ 
coincides with conjugation by $T(z)$.
Hence, equation \eqref{L0} reduces, in this case, to finding
$T(z)$ as above,
and $h(z)=\sum_{n\in\mb Z_+}h_nz^{-n}\in\mf h\otimes\mc V[[z^{-1}]]$,
such that
$$
T(z)\big(\partial+u+zs\otimes 1\big)=\big(\partial+zs\otimes1+h(z)\big)T(z)\,.
$$
The above equation gives rise to the following recursive formula for $h_n\in\mf h\otimes\mc V$
and $T_{n+1}\in\mf h^\perp\otimes\mc V$ ($n\in\mb Z_+$):
$$
h_n+[s\otimes1,T_{n+1}]=T_nu-\partial T_n-\sum_{k=0}^{n-1}h_kT_{n-k}\,.
$$
Clearly, the above equation determines $h_n$ and $T_{n+1}$ uniquely.
The first few terms in the recursion are given by
(in the sums below the terms with zero denominator are dropped):
$$
\begin{array}{l}
\displaystyle{
h_0=\sum_{k}E_{kk}\otimes E_{kk}\,,
\qquad
T_1=\sum_{i,j}E_{ij}\otimes\frac{E_{ji}}{s_i-s_j}\,,
} \\
\displaystyle{
h_1=\sum_{k}E_{kk}\otimes
\Bigg(\sum_{\alpha}
\frac{E_{k\alpha}E_{\alpha k}}{s_k-s_\alpha}\Bigg)\,,
} \\
\displaystyle{
T_2=
\sum_{i,j}E_{ij}\otimes
\Bigg(\sum_{k}
\frac{E_{jk}E_{ki}}{(s_i-s_j)(s_i-s_k)}
-\frac{E_{ji}^{\prime}+E_{ji}E_{ii}}{(s_i-s_j)^2}
\Bigg)\,,
} \\
\displaystyle{
h_2=
\sum_{k}E_{kk}\otimes
\Bigg(\sum_{\alpha,\beta}
\frac{E_{k\alpha}E_{\alpha\beta}E_{\beta k}}{(s_k-s_\alpha)(s_k-s_\beta)}
-\sum_{\alpha}
\frac{E_{k\alpha}E_{\alpha k}^{\prime}+E_{k\alpha}E_{\alpha k}E_{kk}}{(s_k-s_\alpha)^2}
\Bigg)\,.
}
\end{array}
$$
The integrable hierarchy associated to 
$a=\diag(a_1,\dots,a_N)\in Z(\mf h)=\mf h$,
with not all $a_i$'s equal,
is defined in terms of the Hamiltonian functionals in involution
$\tint f_n=\tint \tr((a\otimes1) h_n),\,n\in\mb Z_+$.
The first few elements are
(again the terms with zero denominator are dropped from the sums):
$$
\begin{array}{l}
\displaystyle{
\int f_0
=\int\sum_{k}a_kE_{kk},
\qquad
\int f_1
=\int\sum_{k,\alpha}
\frac{a_kE_{k\alpha}E_{\alpha k}}{s_k-s_\alpha}\,,
}\\
\displaystyle{
\int f_2
=\int\sum_{\alpha,\beta,k}
\frac{a_kE_{k\alpha}E_{\alpha\beta}E_{\beta k}}{(s_k-s_\alpha)(s_k-s_\beta)}
-\sum_{\alpha,k}
\frac{a_k\big(E_{k\alpha}E_{\alpha k}^{\prime}+E_{k\alpha}E_{\alpha k}E_{kk}\big)}{(s_k-s_\alpha)^2}
\,.
}
\end{array}
$$
The corresponding hierarchy of Hamiltonian equations is 
$\frac{du}{dt_n}=H(F_{n})=\partial F_{n}+[u,F_n],\,n\in\mb Z_+$,
where $F(z)=\sum_{n\in\mb Z_+}F_nz^{-n}=T(z)^{-1}(a\otimes1)T(z)$ (see equation \eqref{eq:F}).
In particular, $F_0=a\otimes1$ and $F_1=[a\otimes1,T_1]
=\sum_{i\neq j}\frac{a_i-a_j}{s_i-s_j}E_{ij}\otimes E_{ji}\in\mf{gl}_N\otimes\mc V$.
Hence, 
the first two equations of the hierarchy are ($1\leq i,j\leq N$):
$$
\frac{d E_{ij}}{dt_0}=(a_i-a_j)E_{ij},
\qquad
\frac{d E_{ij}}{dt_1}
=\gamma_{ij}E_{ij}^\prime
+\sum_{k}(\gamma_{ik}-\gamma_{kj})E_{ik}E_{kj}
\,,
$$
where
$\gamma_{ij}=\frac{a_i-a_j}{s_i-s_j}$ for $i\neq j$ and $\gamma_{ij}=0$ for $i=j$.
The last equation is known as the $N$-wave equation.
\end{example}

\section{Classical \texorpdfstring{$\mc W$}{W}-algebras}\label{sec:dsred}

In this section we give the definition of classical $\mc W$-algebras in the language of Poisson vertex algebras.
We also show how this definition is related to the original definition
of Drinfeld and Sokolov \cite{DS85}.
We thus obtain a bi-Poisson structure for classical $\mc W$-algebras,
that we will use in the next section to apply successfully the Lenard-Magri scheme of integrability.

\subsection{Setup}\label{sec:3.1}
Throughout the rest of the paper we make the following assumptions.

Let $\mf g$ be a reductive finite-dimensional Lie algebra over the field $\mb F$
with a non-degenerate symmetric invariant bilinear form $\kappa$,
and let $f\in\mf g$ be a non-zero nilpotent element.
By the Jacobson-Morozov Theorem \cite[Theorem 3.3.1]{CMG93},
$f$ is an element of an $\mf{sl}_2$-triple $(f,h=2x,e)\subset\mf g$.
Then we have the $\ad x$-eigenspace decomposition
\begin{equation}\label{dec}
\mf g=\bigoplus_{i\in\frac{1}{2}\mb Z}\mf g_i\,,
\end{equation}
so that $f\in\mf g_{-1}$, $h\in\mf g_0$ and $e\in\mf g_1$.

There is a well-known skew-symmetric bilinear form
$\omega$ on $\mf g$ defined by
$$
\omega(a,b)=\kappa(f\mid[a,b])
\,\,,\,\,\,\,a,b\in\mf g\,.
$$
Its restriction to $\mf g_{\frac12}$ 
is non-degenerate since $\ad f:\,\mf g_{\frac12}\to\mf g_{-\frac12}$ is bijective.
Fix an isotropic subspace $\mf l\subset\mf g_{\frac12}$
(with respect to $\omega$)
and denote by $\mf l^{\perp\omega}=
\{a\in\mf g_{\frac12}\mid \omega(a,b)=0\text{ for all }b\in\mf l\}
\subset\mf g_{\frac12}$ its symplectic complement
with respect to $\omega$.
We consider the following nilpotent subalgebras of $\mf g$:
\begin{equation}\label{subalgebras}
\mf m=\mf l\oplus\mf g_{\geq1}\subset
\mf n=\mf l^{\perp\omega}\oplus\mf g_{\geq1}\,,
\end{equation}
where $\mf g_{\geq1}=\bigoplus_{i\geq1}\mf g_i$.
Clearly, $\omega(\cdot\,,\,\cdot)$ restricts to a skew-symmetric bilinear form on $\mf n$,
and $\mf m\subset\mf n$ is the kernel of this restriction.
Hence, we have an induced non-degenerate skew-symmetric bilinear form on $\quot{\mf n}{\mf m}$.

Let $\mf p\subset\mf g_{\leq\frac12}$
be a subspace complementary to $\mf m$ in $\mf g$:
$\mf g=\mf m\oplus\mf p$.
Since $\kappa$
is non-degenerate, we also have the corresponding decomposition with the orthogonal complements
$\mf g=\mf p^\perp\oplus\mf m^\perp$.
Hence, identifying $\mf g\simeq\mf g^*$ via $\kappa$,
we get the isomorphisms $\mf p^*\simeq\quot{\mf g}{\mf p^\perp}\simeq\mf m^\perp$.
Let $\{q_i\}_{i\in P}$ be a basis of $\mf p$, and let
$\{q^i\}_{i\in P}$ be the dual (with respect to $\kappa$) basis of $\mf m^{\perp}$,
namely, such that $\kappa(q^j\mid q_i)=\delta_{ij}$.
These dual bases are equivalently defined by the completeness relations
\begin{equation}\label{complete}
\sum_{j\in P}\kappa(q^j\mid a)q_j=\pi_{\mf p} a,
\qquad
\sum_{j\in P}\kappa(a\mid q_j)q^j=\pi_{\mf m^\perp} a
\qquad
\text{for all }a\in\mf g\,,
\end{equation}
where $\pi_{\mf p}:\,\mf g\twoheadrightarrow\mf p$ 
and $\pi_{\mf m^\perp}:\,\mf g\twoheadrightarrow\mf m^{\perp}$
are the projection maps with kernels $\mf m$ and $\mf p^\perp$ respectively.

For $a\in\mf n$, we have $\pi_{\mf p}(a)\in\mf p$, and $a-\pi_{\mf p}(a)\in\mf m\subset\mf n$.
Hence, $\pi_{\mf p}(a)\in\mf p\cap\mf n$.
It follows that
\begin{equation}\label{20130206:eq1}
\pi_{\mf p}(\mf n)=\mf n\cap\mf p\subset\mf g_{\frac12}\,,
\end{equation}
and this space is naturally isomorphic to $\quot{\mf n}{\mf m}\simeq\quot{\mf l^{\perp\omega}}{\mf l}$.
It is then clear that the bilinear form $\omega$
restricts to a non-degenerate skew-symmetric 
bilinear form on $\mf n\cap\mf p$.
Let $\{v^r\}_{r\in R}$, $\{v_r\}_{r\in R}$
(where $\#(R)=\dim\mf n-\dim\mf m=\dim\mf l^{\perp\omega}-\dim\mf l$),
be bases of $\mf n\cap\mf p$
dual with respect to $\omega(\cdot\,,\,\cdot)$, i.e.
\begin{equation}\label{20130206:eq2}
\omega(v^q,v_r)=\delta_{q,r}\,\,
\text{ for all } q,r\in R\,.
\end{equation}
We have the following completeness relations:
\begin{equation}\label{20130206:eq3}
\pi_{\mf p}(a)
=\sum_{r\in R}\kappa(f\mid [a,v^r])v_r
=-\sum_{r\in R}\kappa(f\mid[a,v_r])v^r
\,\,\text{ for all } a\in\mf n\,.
\end{equation}

Finally, we fix an element $s\in\Ker(\ad\mf n)\subset\mf g$.
In the next section, when applying the Lenard-Magri scheme of integrability,
we will need some further assumptions on the element $s$ (see Section \ref{sec:4.2}).

\subsection{Definition of classical \texorpdfstring{$\mc W$}{W}-algebras}\label{sec:3.2}

Let us consider the affine PVA $\mc V(\mf g)=\mc V(\mf g,\kappa,zs)$, where $z\in\mb F$,
from Example \ref{affinePVA}. 
As a differential algebra, it is $\mc V(\mf g)=S(\mb F[\partial]\mf g)$,
and the $\lambda$-bracket on it is given by
\begin{equation}\label{lambda}
\{a_\lambda b\}_z=[a,b]+\kappa(a\mid b)\lambda+z\kappa(s\mid[a,b])
\quad,\qquad a,b\in\mf g\,,
\end{equation}
and extended to $\mc V(\mf g)$ by the Master Formula \eqref{masterformula}.
Note that since, by assumption, $[s,\mf n]=0$,
$\mb F[\partial]\mf n\subset\mc V(\mf g)$ is a Lie conformal subalgebra (see Definition \ref{def:lca}),
with the $\lambda$-bracket $\{a_\lambda b\}_z=[a,b]$, $a,b\in\mf n$
(it is independent on $z$).

Consider the differential subalgebra
$\mc V(\mf p)=S(\mb F[\partial]\mf p)$ of $\mc V(\mf g)$,
and denote by $\rho:\,\mc V(\mf g)\twoheadrightarrow\mc V(\mf p)$,
the differential algebra homomorphism defined on generators by
\begin{equation}\label{rho}
\rho(a)=\pi_{\mf p}(a)+\kappa(f\mid a),
\qquad a\in\mf g\,.
\end{equation}
Note that, since by assumption $\mf p\subset\mf g_{\leq\frac12}$,
we have that $\rho$ acts as the identity on $\mc V(\mf p)$.
\begin{lemma}\phantomsection\label{20120511:lem1}
\begin{enumerate}[(a)]
\item
For every $a\in\mf n$ and $g\in\mc V(\mf m)=S(\mb F[\partial]\mf m)\subset\mc V(\mf g)$,
we have $\rho\{a_\lambda g\}_z=0$.
\item
For every $a\in\mf n$ and $g\in\mc V(\mf g)$, we have
$\rho\{a_\lambda \rho(g)\}_z=\rho\{a_\lambda g\}_z$.
\item
We have a representation of the Lie conformal algebra $\mb F[\partial]\mf n$ 
on the differential subalgebra $\mc V(\mf p)\subset\mc V(\mf g)$ given by
($a\in\mf n$, $g\in\mc V(\mf p)$):
\begin{equation}\label{20120511:eq1}
a\,^\rho_\lambda\,g=\rho\{a_\lambda g\}_z
\end{equation}
(note that the RHS is independent of $z$ since, by assumption, $s\in\Ker(\ad\mf n)$).
\item
The $\lambda$-action of $\mb F[\partial]\mf n$ on $\mc V(\mf p)$ given by \eqref{20120511:eq1}
is by conformal derivations (see Definition \ref{def:lca}(c)).
\end{enumerate}
\end{lemma}
\begin{proof}
To prove (a), we can use the Master Formula \eqref{masterformula} 
to reduce to the case when $g\in\mf m$.
In this case, (a) is immediate since, by the definitions \eqref{subalgebras} of $\mf m$ and $\mf n$,
we have $[\mf m,\mf n]\subset\mf m$ and $\kappa(f\mid[\mf m,\mf n])=0$.
Next, let us prove part (b). Since, by construction, 
$\mf g=\mf m\oplus\mf p$, we have $\mc V(\mf g)=\mc V(\mf m)\otimes\mc V(\mf p)$.
Part (b) then follows immediately by part (a) and the left Leibniz rule \eqref{lleibniz},
using the fact that $\rho$ acts as the identity on $\mc V(\mf p)$.
As for parts (c) and (d), 
clearly the $\lambda$-action \eqref{20120511:eq1} satisfies sesquilinearity and
the Leibniz rule, since $\rho$ is a differential algebra homomorphism.
We are left to prove the Jacobi identity for this $\lambda$-action.
For $a,b\in\mf n$ and $g\in\mc V(\mf p)$ we have, by part (b),
$$
\begin{array}{c}
\displaystyle{
\vphantom{\Big(}
a\,^\rho_\lambda\,(b\,^\rho_\mu\,g)
-b\,^\rho_\mu\,(a\,^\rho_\lambda\,g)
=\rho\{a_\lambda\{b_\mu g\}_z\}_z-\rho\{b_\mu\{a_\lambda g\}_z\}_z
} \\
\displaystyle{
\vphantom{\Big(}
=\rho\{{\{a_\lambda b\}_z}_{\lambda+\mu} g\}_z
={\{a_\lambda b\}_z}\,^\rho_{\lambda+\mu}\,g\,.
} 
\end{array}
$$
\end{proof}

We let $\mc W\subset\mc V(\mf p)$ be the subspace
annihilated by the Lie conformal algebra action of $\mb F[\partial]\mf n$:
\begin{equation}\label{20120511:eq2}
\mc W=\mc V(\mf p)^{\mb F[\partial]\mf n}
=\big\{g\in\mc V(\mf p)\,\mid\,a\,^\rho_\lambda\,g=0\,\text{ for all }a\in\mf n\}\,.
\end{equation}
\begin{lemma}\phantomsection\label{20120511:lem2}
\begin{enumerate}[(a)]
\item
$\mc W\subset\mc V(\mf p)$ is a differential subalgebra.
\item
For every $g\in\mc W$ and $h\in\mc V(\mf m)$, we have
$\rho\{g_\lambda h\}_z=\rho\{h_\lambda g\}_z=0$.
\item
For every elements $g\in\mc W$ and $h\in\mc V(\mf g)$, we have
$\rho\{g_\lambda\rho(h)\}_z=\rho\{g_\lambda h\}_z$,
and $\rho\{\rho(h)_\lambda g\}_z=\rho\{h_\lambda g\}_z$.
\item
For every $g,h\in\mc W$, we have $\rho\{g_\lambda h\}_z\in\mb F[\lambda]\otimes\mc W$.
\item
The map $\{\cdot\,_\lambda\,\cdot\}_{z,\rho}:\,\mc W\otimes\mc W\to\mb F[\lambda]\otimes\mc W$
given by 
\begin{equation}\label{20120511:eq3}
\{g_\lambda h\}_{z,\rho}=\rho\{g_\lambda h\}_z
\end{equation}
defines a PVA structure on $\mc W$.
\end{enumerate}
\end{lemma}
\begin{proof}
Part (a) follows from the fact that the $\lambda$-action \eqref{20120511:eq1}
of $\mb F[\partial]\mf n$ on $\mc V(\mf p)$ is by conformal derivations.
As for part (b), 
we can use the Master Formula \eqref{masterformula} to reduce to the case when $h\in\mf m$,
and in this case the statement is obvious by the definition of $\mc W$
(and the fact that $\mf m\subset\mf n$).
Since $\mc V(\mf g)=\mc V(\mf m)\otimes\mc V(\mf p)$,
part (c) follows immediately by part (b) and the left and right Leibniz rules \eqref{lleibniz}-\eqref{rleibniz},
using the fact that $\rho$ acts as the identity on $\mc V(\mf p)$.
Next, let us prove part (d).
For $a\in\mf n$ and $g,h\in\mc W$,
$$
\begin{array}{c}
\rho\{a_\lambda\rho\{g_\mu h\}_z\}_z
=\rho\{a_\lambda\{g_\mu h\}_z\}_z
=\rho\{{\{a_\lambda g\}_z}_{\lambda+\mu} h\}_z+\rho\{g_\mu\{a_\lambda h\}_z\}_z
\\
=\rho\{\rho{\{a_\lambda g\}_z}_{\lambda+\mu} h\}_z+\rho\{g_\mu\rho\{a_\lambda h\}_z\}_z=0\,.
\end{array}
$$
In the first equality we used Lemma \ref{20120511:lem1}(b),
while in the third equality we used part (c).
It follows that $\rho\{g_\mu h\}_z$ lies in $\mb F[\mu]\otimes\mc W$, proving (d).
Finally, let us prove part (e).
Since $\rho$ is a differential algebra homomorphism,
the $\lambda$-bracket \eqref{20120511:eq3} obviously satisfies
sesquilinearity, skewsymmetry, and the left and right Leibniz rules.
We are left to check the Jacobi identity.
For $g,h,k\in\mc W$ we have, by part (c),
$$
\begin{array}{c}
\{h_\lambda\{k_\mu g\}_{z,\rho}\}_{z,\rho}
=\rho\{h_\lambda\{k_\mu g\}_{z}\}_{z}
=\rho\{{\{h_\lambda k\}_{z}}_{\lambda+\mu} g\}_{z}
+\rho\{k_\mu\{h_\lambda g\}_{z}\}_{z}
\\
=\{{\{h_\lambda k\}_{z,\rho}}_{\lambda+\mu} g\}_{z,\rho}
+\{k_\mu\{h_\lambda g\}_{z,\rho}\}_{z,\rho}\,.
\end{array}
$$
\end{proof}
\begin{corollary}\label{20130206:cor}
\begin{enumerate}[(a)]
\item
An element $g\in\mc V(\mf g)$ is such that $\rho(g)\in\mc W$
if and only if $\rho\{a_\lambda g\}_z=0$ for every $a\in\mf n$.
\item
If $g\in\mc V(\mf g)$ is such that $\rho(g)\in\mc W$ and $h\in\mc V(\mf m)$, 
we have $\rho\{g_\lambda h\}_z=\rho\{h_\lambda g\}_z=0$.
\item
If $g\in\mc V(\mf g)$ is such that $\rho(g)\in\mc W$ and $h\in\mc V(\mf g)$, 
we have
$\rho\{g_\lambda\rho(h)\}_z=\rho\{g_\lambda h\}_z$,
and $\rho\{\rho(h)_\lambda g\}_z=\rho\{h_\lambda g\}_z$.
\item
If $g,h\in\mc V(\mf g)$ are such that $\rho(g),\rho(h)\in\mc W$,
then
$\rho\big(\{\rho(g)_\lambda\rho(h)\}_z\big)=\rho\big(\{g_\lambda h\}_z\big)$.
\end{enumerate}
\end{corollary}
\begin{proof}
By Lemma \ref{20120511:lem1}(b) we have, for $a\in\mf n$ and $g\in\mc V(\mf g)$,
$\rho\{a_\lambda\rho(g)\}_z=\rho\{a_\lambda g\}_z$.
Part (a) follows from the definition \eqref{20120511:eq2} of the space $\mc W$.
The proofs of (b) and (c) are the same as the proofs 
of Lemma \ref{20120511:lem2}(b) and (c) respectively,
using part (a) instead of the definition of $\mc W$.
Finally, 
part (d) is an immediate corollary of Lemma \ref{20120511:lem2}(c)
and of part (c).
\end{proof}

\begin{definition}\label{walg3}
The classical $\mc W$-algebra 
is the differential algebra $\mc W$ defined by \eqref{20120511:eq2}
with the PVA structure given by \eqref{20120511:eq3}.
\end{definition}
\begin{lemma}\label{20130207:lem}
We have $\partial\mc V(\mf g)\cap\mc W=\partial\mc W$.
In particular, 
we have a natural embedding 
$\quot{\mc W}{\partial\mc W}\subset\quot{\mc V(\mf g)}{\partial\mc V(\mf g)}$.
\end{lemma}
\begin{proof}
If $p\in\mc V(\mf g)$ is such that $\partial p\in\mc W$,
first it is clear that $p\in\mc V(\mf p)$.
Moreover, by definition of $\mc W$ we have 
$a_\lambda^\rho\partial p=(\lambda+\partial)\big(a_\lambda^\rho p\big)=0$ for $a\in\mf n$,
which immediately implies that $a_\lambda^\rho p=0$, i.e. $p\in\mc W$.
\end{proof}

\begin{remark}\label{wz}
The Poisson vertex algebra $\mc W$ can be constructed in the same way
for an arbitrary choice of $s$ in $\mf g$ 
(taking the Lie conformal algebra $\mb F[\partial]\mf n\oplus\mb F$ of $\mc V$).
However the differential algebra $\mc W$ is independent of the choice of $z\in\mb F$
if and only if $[s,\mf n]=0$.
This independence of $z$ will be very important in the next section,
where we construct integrable hierarchies of Hamiltonian equations,
since there we need to view $z$ as a formal parameter.
\end{remark}

\begin{remark}
In literature, the name classical $\mc W$-algebra is referred to the Poisson 
structure corresponding to the case $z=0$. 
As we will see, the whole family of PVAs $\mc W$, 
parametrized by $z\in\mb F$, plays an important role in obtaining an integrable hierarchy 
of Hamiltonian equations associated to the classical $\mc W$-algebra.
\end{remark}

Recall that we fixed a basis $\{q_i\}_{i\in P}$ of $\mf p$
and the dual basis $\{q^i\}_{i\in P}$ of $\mf m^\perp$.
We can find an explicit formula for the $\lambda$-bracket in $\mc W$ as follows.
Recalling the Master Formula \eqref{masterformula} and using \eqref{rho} 
and the definition \eqref{lambda} of the $\lambda$-bracket in $\mc V$, we get
($g,h\in\mc W$):
\begin{equation}\label{wbrack}
\{g_{\lambda}h\}_{z,\rho}
=\{g_{\lambda}h\}_{H,\rho}-z\{g_{\lambda}h\}_{K,\rho}\,,
\end{equation}
where
\begin{equation}\label{wbrackX}
\{g_{\lambda}h\}_{X,\rho}
=\sum_{\substack{i,j\in P\\m,n\in\mb Z_+}}
\frac{\partial h}{\partial q_j^{(n)}}(\lambda+\partial)^n
X_{ji}(\lambda+\partial)
(-\lambda-\partial)^m\frac{\partial g}{\partial q_i^{(m)}}\,,
\end{equation}
for $X$ one of the two matrices 
$H,K\in\Mat_{k\times k}\mc V(\mf p)[\lambda]$ ($k=\#(P)$),
given by
\begin{equation}\label{HKds}
H_{ji}(\lambda)
=\pi_{\mf p}[q_i,q_j]+\kappa(q_i\mid q_j)\lambda+\kappa(f\mid[q_i,q_j]),
\qquad
K_{ji}(\partial)=\kappa(s\mid[q_j,q_i])
\,,
\end{equation}
for $i,j\in P$.

Recall that a $k\times k$ matrix with entries in $\mc V(\mf p)[\lambda]$
corresponds, via \eqref{eq:identif}, to a linear map 
$\mf p\otimes\mc V(\mf p)\to \mf p^*\otimes\mc V(\mf p)$,
and that we can identify $\mf p^*\simeq\mf m^\perp$ via the bilinear form $\kappa$.
Therefore, we can describe the above matrices $H$ and $K$
as the following linear maps $\mf p\otimes\mc V(\mf p)\to \mf m^\perp\otimes\mc V(\mf p)$:
\begin{equation}\label{HKmaps}
\begin{array}{l}
\displaystyle{
\vphantom{\Big(}
H(a\otimes g)
=\sum_{i\in P}\pi_{\mf m^\perp}[q^i,a]\otimes q_ig
+\pi_{\mf m^\perp}[f,a]\otimes g
+\pi_{\mf m^\perp}(a)\otimes \partial g
\,,} \\
\displaystyle{
\vphantom{\Big(}
K(a\otimes g)
=\pi_{\mf m^\perp}[a,s]\otimes g
\,,
}
\end{array}
\end{equation}
for every $a\in\mf p$ and $g\in\mc V(\mf p)$.

Note that, even though the $\lambda$-bracket \eqref{wbrack} on the PVA $\mc W$
is formally associated to the matrices $H$ and $K$ in \eqref{HKds} 
via the Master Formula \eqref{masterformula},
$H$ and $K$ are NOT Poisson structures (on $\mc V(\mf p)$).
Indeed, $\mc V(\mf p)$ is not a PVA, namely the $\lambda$-bracket $\{\cdot\,_\lambda\,\cdot\}_{z,\rho}$
on $\mc V(\mf p)$ (given by the same formula \eqref{wbrack})
is not a PVA $\lambda$-bracket.

\begin{remark}\label{hamred}
The classical $\mc W$-algebra can be equivalently defined,
without fixing a complementary subspace $\mf p\subset\mf g$ of $\mf m$,
via the so called ``classical Hamiltonian reduction'' (see \cite{DSK06}).
The general construction is as follows.
Let $\mc V$ be a Poisson vertex algebra,
let $\mu:\,R\to\mc V$ be a Lie conformal algebra homomorphism,
and let $I_0\subset S(R)$ be a differential algebra ideal 
such that $[R_\lambda I_0]\subset \mb F[\lambda]\otimes I_0$.
Let $I\subset\mc V$ be the differential algebra ideal generated by $\mu(I_0)$.
The corresponding classical \emph{Hamiltonian reduction} is defined as the differential algebra
$$
\mc W(\mc V,R,I_0)=(\mc V/I)^{\mu(R)}
=\big\{f+I\,\big|\ \{\mu(a)_\lambda f\}\in\mb F[\lambda]\otimes I \text{ for all } a\in R\big\}
\,,
$$
endowed with the $\lambda$-bracket $\{f+I_\lambda g+I\}=\{f_\lambda g\}+\mb F[\lambda]\otimes I$.
It is not hard to show that this $\lambda$-bracket is well defined.
The classical $\mc W$-algebra is obtained by taking
$\mc V=\mc V(\mf g,\kappa,s)$, 
$R=\mb F[\partial]\mf n$, 
and the differential algebra ideal $I_0$ of $S(R)$
generated by $\big\{m-\kappa(f|m)\,\big|\,m\in\mf m\big\}$,
so that $I=\ker\rho$
(note that $\Ker\rho$ is independent of the choice of $\mf p$).
%
%
To see this, let
\begin{equation}\label{wtilde}
\widetilde{\mc W}
=\big\{g\in\mc V(\mf g)\,\mid\,\{a_\lambda g\}_z\in\mb F[\lambda]\otimes\Ker\rho
\text{ for all }a\in\mf n\big\}\subset\mc V(\mf g)\,.
\end{equation}
Since $[s,\mf n]=0$, the space $\widetilde{\mc W}$ is independent of $z$.
Clearly, the map $\rho:\,\mc V(\mf g)\to\mc V(\mf p)$
induces differential algebra isomorphism $\quot{\mc V(\mf g)}{\Ker\rho}\simeq\mc V(\mf p)$,
which restricts to a differential algebra isomorphism
$\quot{\widetilde{\mc W}}{\Ker\rho}\simeq\mc W$.
%
\end{remark}

\begin{remark}\label{uhirin}
The PVA $\mc W$ was constructed in \cite{DSK06}
as a quasiclassical limit of a family of vertex algebras,
obtained by a cohomological construction in \cite{KW04}.
The isomorphism of this construction with the construction in the present paper
via classical Hamiltonian reduction is proved in \cite{Suh12}.
\end{remark}

\subsection{Gauge transformations and Drinfeld-Sokolov approach to classical \texorpdfstring{$\mc W$}{W}-algebras}\label{sec:3.3}

In this section we show that the definition of the classical $\mc W$-algebra
given in Section \ref{sec:3.2}
is equivalent to the original definition of Drinfeld and Sokolov \cite{DS85},
given in terms of gauge invariance.

Recall from Section \ref{sec:hom} 
the definition of the Lie algebra $\mb F\partial\ltimes(\mf g\otimes\mc V(\mf g))$.
%
Let
\begin{equation}\label{q}
q=\sum_{i\in P}q^i\otimes q_i
\,\,\in\mf m^{\perp}\otimes\mc V(\mf p)\,.
\end{equation}
Note that $q=(\pi_{\mf m^\perp}\otimes1)u$,
where $u\in\mf g\otimes\mc V(\mf g)$ was defined in Proposition \ref{int_hier1}.
Let
$$
L=\partial+q+f\otimes1\in\mb F\partial\ltimes(\mf g\otimes\mc V(\mf p))\,.
$$
A \emph{gauge transformation} is, by definition, 
a change of variables formula $q\mapsto q^A\in\mf m^\perp\otimes\mc V(\mf p)$,
for $A\in\mf n\otimes\mc V(\mf p)$, given by
\begin{equation}\label{gauge}
e^{\ad A}L=\partial+q^A+f\otimes1\,.
\end{equation}
In \cite{DS85}, Drinfeld and Sokolov defined the classical $\mc W$-algebra
as the subspace $\mc W\subset\mc V(\mf p)$ consisting of \emph{gauge invariant}
differential polynomials $g$, that is, such that $g(q^A)=g(q)$ for every $A\in\mf n\otimes\mc V(\mf p)$.
Here and further we use the following notation:
for $g\in\mc V(\mf p)$ and $r=\sum_i q^i\otimes r_i\in\mf m^\perp\otimes\mc V(\mf p)$,
we let $g(r)$  be the differential polynomial in $q_1,\dots,q_k$
obtained replacing $q_i^{(m)}$ by $\partial^mr_i$ in the differential polynomial $g$.

In this section we will prove that the space of gauge invariant polynomials coincides
with the space $\mc W$ defined in \eqref{20120511:eq2}.
The key observation is that the action of the gauge group $g\mapsto g(q^A)\in\mc V(\mf p)$
is obtained by exponentiating the Lie conformal algebra action of $\mb F[\partial]\mf n$
on $\mc V(\mf p)$ given by \eqref{20120511:eq1}.
This is stated in the following
\begin{theorem}\label{20120511:thm1}
For every $a\otimes h\in\mf n\otimes\mc V(\mf p)$ and $g\in\mc V(\mf p)$, we have
\begin{equation}\label{20120511:eq4}
g(q^{a\otimes h})=
\sum_{n\in\mb Z_+}\frac{(-1)^n}{n!}
(a\,^\rho_{\lambda_1}\,\dots a\,^\rho_{\lambda_n}\,g)\,
\big(\big|_{\lambda_1=\partial}h\big)\dots\big(\big|_{\lambda_n=\partial}h\big)
\,,
\end{equation}
where,
for a polynomial $p(\lambda_1,\dots,\lambda_n)
=\sum c\, \lambda_1^{i_1}\dots\lambda_n^{i_n}$,
we denote
$$
p(\lambda_1,\dots,\lambda_n)
\big(\big|_{\lambda_1=\partial}h_1\big)\dots\big(\big|_{\lambda_n=\partial}h_n\big)
=\sum c\,(\partial^{i_1}h_1)\dots(\partial^{i_n}h_n)
\,.
$$
\end{theorem}
\begin{proof}
First, by Lemma \ref{20120511:lem1}(b), the RHS of equation \eqref{20120511:eq4} is
$$
\sum_{n\in\mb Z_+}^\infty\frac{(-1)^n}{n!}
\rho\{a_{\lambda_1}\dots \{a_{\lambda_n}g\}_z\dots\}_z
\big(\big|_{\lambda_1=\partial}h\big)\dots\big(\big|_{\lambda_n=\partial}h\big)
\,.
$$
Next, we expand the LHS of equation \eqref{20120511:eq4}
in Taylor series, using the definition \eqref{gauge} of $q^{a\otimes h}$:
$$
\begin{array}{l}
\displaystyle{
g(q^{a\otimes h})=
g\Big(q+\sum_{n\geq1}\frac1{n!}(\ad a\otimes h)^n(\partial+f\otimes1+q)\Big)
} \\
\displaystyle{
=
\!\!\!\!\!\sum_{\substack{s\in\mb Z_+,i_1,\dots,i_s\in P \\ m_1,\dots,m_s\in\mb Z_+ \\ n_1,\dots,n_s\geq1}}\!\!
\frac{1}{s!n_1!\dots n_s!}
\frac{\partial^s g}{\partial q_{i_1}^{(m_1)}\dots\partial q_{i_s}^{(m_s)}}
\big(\partial^{m_1}(\ad a\otimes h)^{n_1}(\partial+f\otimes1+q)\big)_{i_1}
\dots
} \\
\displaystyle{
\vphantom{\Big(}
\,\,\,\,\,\,\,\,\,\,\,\,\,\,\,\,\,\,\,\,\,\,\,\,\,\,\,\,\,\,\,\,\,\,\,\,\,\,\,\,\,\,\,\,\,\,\,\,\,\,\,\,\,\,\,\,\,\,\,\,\,\,\,\,\,\,\,\,\,\,\,\,\,\,\,\,\,\,\,\,\,
\dots
\big(\partial^{m_s}(\ad a\otimes h)^{n_s}(\partial+f\otimes1+q)\big)_{i_s}\,.
}
\end{array}
$$
Combining the above equations, we deduce that equation \eqref{20120511:eq4}
is equivalent to (for every $N\in\mb Z_+$):
\begin{equation}\label{20120516:eq1}
\begin{array}{l}
\displaystyle{
\rho\{a_{\lambda_1}\dots \{a_{\lambda_N}g\}_z\dots\}_z
\big(\big|_{\lambda_1=\partial}h\big)\dots\big(\big|_{\lambda_N=\partial}h\big)
} \\
\displaystyle{
=
\sum_{\substack{s\in\mb Z_+,\,i_1,\dots,i_s\in P \\ m_1,\dots,m_s\in\mb Z_+ \\ 
n_1,\dots,n_s\geq1 \\ (n_1+\dots+n_s=N)}}
\frac{N!}{s!n_1!\dots n_s!} (-1)^N
\frac{\partial^s g}{\partial q_{i_1}^{(m_1)}\dots\partial q_{i_s}^{(m_s)}}
} \\
\displaystyle{
\vphantom{\Big(}
\times
\big(\partial^{m_1}(\ad a\otimes h)^{n_1}(\partial+f\otimes1+q)\big)_{i_1}
\dots\big(\partial^{m_s}(\ad a\otimes h)^{n_s}(\partial+f\otimes1+q)\big)_{i_s}
\,.}
\end{array}
\end{equation}

We start by proving equation \eqref{20120516:eq1} when $g=q_i$, for every $i\in P$.
Namely, we need to prove that, for every $n\geq1$, we have
\begin{equation}\label{20120516:eq2}
\rho\{a_{\lambda_1}\dots \{a_{\lambda_n}q_i\}_z\dots\}_z
\big(\big|_{\lambda_1=\partial}h\big)\dots\big(\big|_{\lambda_n=\partial}h\big)
=
(-1)^n
\big((\ad a\otimes h)^n(\partial+f\otimes1+q)\big)_i
\,.
\end{equation}
By the second completeness relation \eqref{complete} and the invariance of the bilinear form $\kappa$, 
we have
$$
\begin{array}{l}
\displaystyle{
\vphantom{\Big(}
(-1)^n
\big((\ad a\otimes h)^n(\partial+f\otimes1+q)\big)_i
} \\
\vphantom{\Big(}
\displaystyle{
=(-1)^n\kappa\big((\ad a\otimes h)^n(\partial+f\otimes1+q)\,\mid\,q_i\otimes1\big)
} \\
\vphantom{\Big(}
\displaystyle{
=
\delta_{n,1}\kappa(a\mid q_i)\partial h
+\kappa\big(f\,\mid\,(\ad a)^n(q_i)\big)h^n
+\pi_{\mf p}(\ad a)^n(q_i)h^n
\,,}
\end{array}
$$
which is the same as the LHS of \eqref{20120516:eq2}.
Note that, in view of the above computation, the LHS of \eqref{20120516:eq2}
is the same as $\big((a^\rho_\partial)^nq_i\big)_\to h^n$,
where, as usual, the arrow means that $\partial$ should be moved to the right (to act on $h^n$).

In view of \eqref{20120516:eq2} (and the above observation), 
equation \eqref{20120516:eq1} can be rewritten as follows
\begin{equation}\label{20120516:eq3}
\begin{array}{l}
\vphantom{\Bigg(}
\displaystyle{
\rho\{a_{\lambda_1}\dots \{a_{\lambda_N}g\}_z\dots\}_z
\big(\big|_{\lambda_1=\partial}h\big)\dots\big(\big|_{\lambda_N=\partial}h\big)
} \\
\displaystyle{
=
\sum_{\substack{s\in\mb Z_+,\,i_1,\dots,i_s\in P \\ m_1,\dots,m_s\in\mb Z_+ \\ 
n_1,\dots,n_s\geq1 \\ (n_1+\dots+n_s=N)}}
\frac{N!}{s!n_1!\dots n_s!}
\frac{\partial^s g}{\partial q_{i_1}^{(m_1)}\dots\partial q_{i_s}^{(m_s)}}
} \\
\displaystyle{
\times
\Big(\partial^{m_1}
\big((a^\rho_\partial)^{n_1}q_{i_1}\big)_\to h^{n_1}
\Big)
\dots
\Big(\partial^{m_s}
\big((a^\rho_\partial)^{n_s}q_{i_s}\big)_\to h^{n_s}
\Big)
\,.}
\end{array}
\end{equation}
Let us denote the two sides of equation \eqref{20120516:eq3}
by $LHS_N(g)$ and $RHS_N(g)$.
The identity $LHS_N(q_i)=RHS_N(q_i)$ for every $i\in P$ and every $N\in\mb Z_+$
is given by the above observations.
Moreover, it is easy to check that
$LHS_N(\partial g)=\partial LHS_N(g)$ and, using equation \eqref{partialcomm2} $s$ times,
we also have that $RHS_N(\partial g)=\partial RHS_N(g)$.
Furthermore, it is not difficult to prove, using the left Leibniz rule \eqref{lleibniz},
that $LHS_N(g)$ satisfies the functional equation
$$
LHS_N(g_1g_2)=\sum_{n=0}^N\binom{N}{n} LHS_n(g_1)LHS_{N-n}(g_2)\,.
$$
In order to prove that \eqref{20120516:eq3} holds for every $g\in\mc V(\mf p)$,
it suffices to show that $RHS_N(g)$ satisfies the same functional equation.
We have, by the Leibniz rule for partial derivatives,
$$
\begin{array}{l}
\vphantom{\Bigg(}
\displaystyle{
RHS_N(g_1g_2)
=
\sum_{\substack{s\in\mb Z_+,\,i_1,\dots,i_s\in P \\ m_1,\dots,m_s\in\mb Z_+ \\ 
n_1,\dots,n_s\geq1 \\ (n_1+\dots+n_s=N)}}
\frac{N!}{s!n_1!\dots n_s!}
\frac{\partial^s g_1g_2}{\partial q_{i_1}^{(m_1)}\dots\partial q_{i_s}^{(m_s)}}
} \\
\displaystyle{
\,\,\,\,\,\,\,\,\,\,\,\,\,\,\,\,\,\,\,\,\,\,\,\,\,\,\,\,\,\,\,\,\,\,\,\,
\times
\Big(\partial^{m_1}
\big((a^\rho_\partial)^{n_1}q_{i_1}\big)_\to h^{n_1}
\Big)
\dots
\Big(\partial^{m_s}
\big((a^\rho_\partial)^{n_s}q_{i_s}\big)_\to h^{n_s}
\Big)
} \\
\displaystyle{
=
\sum_{\substack{s\in\mb Z_+,\,i_1,\dots,i_s\in P \\ m_1,\dots,m_s\in\mb Z_+ \\ 
n_1,\dots,n_s\geq1 \\ (n_1+\dots+n_s=N)}}
\frac{N!}{s!n_1!\dots n_s!}
\sum_{a=0}^s
\binom{s}{a}
\frac{\partial^a g_1}{\partial q_{i_1}^{(m_1)}\dots\partial q_{i_a}^{(m_a)}}
\frac{\partial^{s-a} g_2}{\partial q_{i_{a+1}}^{(m_{a+1})}\dots\partial q_{i_s}^{(m_s)}}
} \\
\displaystyle{
\,\,\,\,\,\,\,\,\,\,\,\,\,\,\,\,\,\,\,\,\,\,\,\,\,\,\,\,\,\,\,\,\,\,\,\,
\times
\Big(\partial^{m_1}
\big((a^\rho_\partial)^{n_1}q_{i_1}\big)_\to h^{n_1}
\Big)
\dots
\Big(\partial^{m_s}
\big((a^\rho_\partial)^{n_s}q_{i_s}\big)_\to h^{n_s}
\Big)
} \\
\displaystyle{
=\sum_{n=0}^N\binom{N}{n}RHS_n(g_1)RHS_{N-n}(g_2)
\,.
}
\end{array}
$$
\end{proof}
\begin{remark}\label{ref:confgroup}
The gauge transformation $g(q)\mapsto g(q^A)$ is not a group action on $\mc V(\mf p)$.
In view of Theorem \ref{20120511:thm1}, it is rather a ``Lie conformal group'' action.
\end{remark}
As immediate consequence of Theorem \ref{20120511:thm1} we have the following result.
\begin{corollary}\label{20120511:cor}
The space of gauge invariant differential polynomials $g\in\mc V(\mf p)$
coincides with the differential algebra $\mc W$ defined in \eqref{20120511:eq2}.
\end{corollary}

\subsection{Generators of the classical \texorpdfstring{$\mc W$}{W}-algebra}\label{sec:3.4}

Using the description of the classical $\mc W$-algebras in terms of gauge invariance,
we will prove, following the ideas of Drinfeld and Sokolov, that the differential algebra
$\mc W$ is an algebra of differential polynomials in $r=\dim\Ker(\ad f)$ variables, 
and we will provide an algorithm to find explicit generators.
A cohomological proof of this for quantum $\mc W$-algebras,
which also works for classical $\mc W$-algebras,
was given in \cite{KW04,DSK06}.

Note that, by the definitions \eqref{subalgebras} of $\mf m$ and $\mf n$,
we have $[f,\mf n]\subset\mf m^\perp$.
Furthermore, from representation theory of $\mf{sl}_2$, we know that
$\ad f:\,\mf n\to\mf m^\perp$ is an injective map.
Fix a subspace $V\subset\mf m^\perp$ complementary to $[f,\mf n]$,
compatible with the direct sum decomposition \eqref{dec}:
$V=\bigoplus_{i\geq0}V_i$, where $V_i\subset\mf m^\perp\cap\mf g_i$
is a subspace complementary to $[f,\mf n\cap\mf g_{i+1}]$.
Clearly, $\dim(V)=\dim\Ker(\ad f)$
(by representation theory of $\mf{sl}_2$ we can choose, for example, $V=\Ker(\ad e)$).

Before stating the main result of this section, we introduce some important gradings.
In the algebra of differential polynomials $\mc V(\mf p)$
we have the usual \emph{polynomial grading}, 
$\mc V(\mf p)=\bigoplus_{n\in\mb Z_+}\mc V(\mf p)(n)$.
For a homogeneous polynomial $g\in\mc V(\mf p)$,
we denote by $\deg(g)$ its degree,
and for an arbitrary polynomial $g$ we let $g(n)$ be its homogeneous component of degree $n$.
Also, we extend this decomposition to $\mf m^\perp\otimes\mc V(\mf p)$
by looking only at the second factor in the tensor products.
We thus have the polynomial degree decomposition
\begin{equation}\label{deg1}
\mf m^\perp\otimes\mc V(\mf p)
=\bigoplus_{n\in\mb Z_+}\mf m^\perp\otimes\mc V(\mf p)(n)\,.
\end{equation}

The algebra $\mf g$, as well as its subspace $\mf m^\perp$ 
and its subalgebra $\mf n\subset\mf m^\perp$,
has the decomposition \eqref{dec} by $\ad x$-eigenspaces.
We extend this grading to
the Lie algebra $\mf g\otimes\mc V(\mf p)$
by looking only at the first factor in the tensor product:
\begin{equation}\label{deg2}
\mf g\otimes\mc V(\mf p)
=\bigoplus_{i\in\frac12\mb Z}\mf g_i\otimes\mc V(\mf p)\,.
\end{equation}
For a homogeneous element $X\in\mf g\otimes\mc V(\mf p)$
we denote by $\dx(X)$ its $\ad x$-eigenvalue,
and for an arbitrary element we let $X[i]$ be its homogeneous component 
of $\ad x$-eigenvalue $i$.

In the algebra of differential polynomials $\mc V(\mf p)$
we introduce a second  grading, which we call \emph{conformal weight}
and denote by $\Delta$, defined as follows.
For a monomial $g=a_1^{(m_1)}\dots a_s^{(m_s)}$,
product of derivatives of elements $a_i\in\mf p$ homogeneous with respect 
to the $\ad x$-eigenspace decomposition \eqref{dec},
we define its conformal weight as
\begin{equation}\label{degcw}
\Delta(g)=s-\dx(a_1)-\dots-\dx(a_s)+m_1+\dots+m_s\,.
\end{equation}
As we will see below, this grading restricts to the conformal weight w.r.t. 
an explicitly defined  Virasoro element of $\mc W$,
hence the name ``conformal weight''.
Thus we get the conformal weight space decomposition
$\mc V(\mf p)=\bigoplus_{i\in\frac12\mb Z}\mc V(\mf p)\{i\}$.

Finally, we define a grading of the Lie algebra
$\mb F\partial\ltimes(\mf g\otimes\mc V(\mf p))$,
which we call \emph{weight} and denote by $\wt$, as follows.
We let $\wt(\partial)=1$,
and, for $a\otimes g\in\mf g\otimes\mc V(\mf p)$, we let
$\wt(a\otimes g)=-\dx(a)+\Delta(g)$.
We thus get the corresponding weight space decomposition
\begin{equation}\label{deg3}
\mf g\otimes\mc V(\mf p)
=\bigoplus_{i\in\frac12\mb Z}(\mf g\otimes\mc V(\mf p))\{i\}\,,
\end{equation}
where $(\mf g\otimes\mc V(\mf p))\{i\}=\bigoplus_j\mf g_{j}\otimes(\mc V(\mf p)\{i+j\})$.
It is immediate to check that this is indeed a Lie algebra grading of the Lie algebra 
$\mb F\partial\ltimes(\mf g\otimes\mc V(\mf p))$.

\begin{lemma}\label{linalg}
Assume that $\mf g^f\subset\mf p$ (this is the case, for example, if $\mf p\subset\mf g$ is compatible 
with the $\ad x$-eigenspace decomposition \eqref{dec}).
Consider the direct sum decomposition $\mf m^\perp=[f,\mf n]\oplus V$,
and let $\{v^j\}_{j\in P}$ be a basis of $\mf m^\perp$,
such that $\{v^j\}_{j\in J}$ is a basis of $V$
and $\{v^j\}_{j\in P\backslash J}$ is a basis of $[f,\mf n]$.
Consider the dual (with respect to $\kappa$) basis $\{v_j\}_{j\in P}$ of $\mf p$.
Then $\{v_j\}_{j\in J}$ is a basis of $\mf g^f$.
\end{lemma}
\begin{proof}
First note that, if $\mf p\subset\mf g$ is compatible 
with the $\ad x$-eigenspace decomposition \eqref{dec}, then
$\mf p=\mf p_{\frac12}\oplus\mf g_{\leq0}$, where $\mf p_{\frac12}\subset\mf g_{\frac12}$.
Hence, $\mf g^f\subset\mf g_{\leq0}\subset\mf p$
(proving the statement in parenthesis).

We identify $\mf g$ with $\mf g^*$ via the bilinear form $\kappa$.
The decomposition $\mf g=\mf m\oplus\mf p$ corresponds,
via this identification, to the ``dual'' decomposition
$\mf g=\mf p^\perp\oplus\mf m^\perp$, 
with $\mf p^\perp\simeq\mf m^*$ and $\mf m^\perp\simeq\mf p^*$.
Similarly, the decomposition $\mf m^\perp=[f,\mf n]\oplus V$
corresponds, via the same identification, to the ``dual'' decomposition
$\mf p=(\mf p^\perp\oplus V)^\perp\oplus(\mf p^\perp\oplus[f,\mf n])^\perp$,
with $(\mf p^\perp\oplus V)^\perp\simeq[f,\mf n]^*$ and $(\mf p^\perp\oplus[f,\mf n])^\perp\simeq V^*$.
Hence, in order to prove the lemma, we only have to show that
$$
(\mf p^\perp\oplus[f,\mf n])^\perp=\mf g^f\,.
$$
Since, by assumption, $\mf g^f\subset\mf p$,
the inclusion $\mf g^f\subset(\mf p^\perp\oplus[f,\mf n])^\perp$ is obvious.
On the other hand, $(\mf p^\perp\oplus[f,\mf n])^\perp\simeq V^*$,
so that $\dim(\mf p^\perp\oplus[f,\mf n])^\perp=\dim(V)=\dim(\mf g^f)$,
proving the claim.
\end{proof}

\begin{theorem}\phantomsection\label{20120511:thm2}
\begin{enumerate}[(a)]
\item
There exists a unique $X\in\mf n\otimes\mc V(\mf p)$
such that $w=q^X$ lies in $V\otimes\mc V(\mf p)$.
Moreover, $X$ and $w$ are homogeneous elements of
$\mf g\otimes\mc V(\mf p)$ 
with respect to the weight space decomposition \eqref{deg3},
with weights 
$\wt(X)=0$ and $\wt(w)=1$.
\item
Let $X\in\mf n\otimes\mc V(\mf p)$ be as in (a).
Assume that $\mf p\subset\mf g$ is compatible 
with the $\ad x$-eigenspace decomposition \eqref{dec}.
Let $\{v^j\}_{j\in P}$ be a basis of $\mf m^\perp$ consisting of $\ad x$-eigenvectors,
such that $\{v^j\}_{j\in J}$ is a basis of $V$, 
and $\{v^j\}_{j\in P\backslash J}$ is a basis of $[f,\mf n]$.
Let $\{v_j\}_{j\in P}$ be the corresponding dual basis of $\mf p$,
so that, by Lemma \ref{linalg}, 
$\{v_j\}_{j\in J}$ is a basis of $\mf g^f$.
Then, writing
$$
q^X=w=\sum_{j\in J}v^j\otimes w_j\in V\otimes\mc V(\mf p)\,,
$$
we have that
$w_j\in\mc V(\mf p)$ is homogeneous with respect to the conformal weight decomposition,
of conformal weight
$\Delta(w_j)=1+\dx(v^j)$,
and it has the form
\begin{equation}\label{20120724:eq2}
w_j=v_j+g_j\,,
\end{equation}
where 
$g_j=\sum b_1^{(m_1)}\dots b_s^{(m_s)}\in\mc V(\mf p)\{1+\dx(v^j)\}$
is a sum with $s+m_1+\dots+m_s>1$.
\item
The differential algebra $\mc W$ is the algebra of differential polynomials in
the variables $w_1,\dots,w_r$.
\end{enumerate}
\end{theorem}
\begin{proof}
Consider the expansion the elements $q$, $X$, and $w$ according to
the $\ad x$-eigenspace decomposition \eqref{deg2} of $\mf g\otimes\mc V(\mf p)$:
$q=\sum_{i\geq-\frac12}q[i]$, where $q[i]\in(\mf m^\perp\cap\mf g_i)\otimes\mc V(\mf p)$;
$X=\sum_{i\geq\frac12}X[i]$, where $X[i]\in(\mf n\cap\mf g_i)\otimes\mc V(\mf p)$;
and $w=\sum_{i\geq0}w[i]$, where $w[i]\in V_i\otimes\mc V(\mf p)$.
We want to prove, by induction on $i\geq-\frac12$,
that the elements
$X[i+1]\in(\mf n\cap\mf g_{i+1})\otimes\mc V(\mf p),\,i\geq-\frac12$,
and $w[i]\in V_i\otimes\mc V(\mf p),\,i\geq0$,
are uniquely determined by the equation $q^X=w$,
and they are homogeneous of weights $\wt(X[i+1])=0$ and $\wt(w[i])=1$.

Equating the terms of $\ad x$-eigenvalue $-\frac12$ in both sides of the equation $q^X=w$,
we get the equation
$$
[f\otimes1,X[\frac12]]
=q[-\frac12]
\,.
$$
Since $\ad f$ restricts to a bijection $\mf g_{\frac12}\stackrel{\sim}{\longrightarrow}\mf g_{-\frac12}$,
this uniquely defines $X[\frac12]\in(\mf n\cap\mf g_{\frac12})\otimes\mc V(\mf p)$.
Moreover, since $q[-\frac12]$ is homogeneous of weight $1$
and $\ad(f\otimes1)(\mf g\otimes\mc V(\mf p))\{i\}\subset(\mf g\otimes\mc V(\mf p))\{i+1\}$,
we get that $\wt(X[\frac12])=0$.

Next, fix $i\geq0$ and suppose by induction that we (uniquely) determined all elements
$X[j+1]\in(\mf n\cap\mf g_{j+1})\otimes\mc V(\mf p)$
and $w[j]\in V_i\otimes\mc V(\mf p)$ for $j<i$,
and that $\wt(X[j+1])=0$ and $\wt(w[j])=1$.
Equating the terms of $\ad x$-eigenvalue $i$ in both sides of the equation $q^X=w$,
we get an equation in $w[i]$ and $X[i+1]$ of the form
$$
w[i]+[f\otimes1,X[i+1]]=A\,,
$$
where $A\in\mf g_i\otimes\mc V(\mf p)$ is certain complicated expression,
involving the adjoint action of $X[j+1]$, with $j<i$, on $\partial$, $q$ and $f\otimes 1$,
which is homogeneous of $\ad x$-eigenvalue $\dx(A)=i$,
and of conformal weight $\wt(A)=1$.
Since $\mf g_i=[f,\mf g_{i+1}]\oplus V_i$,
and since $\ad f$ restricts to a bijection 
$\mf g_{i+1}\stackrel{\sim}{\longrightarrow}[f,\mf g_{i+1}]$,
the above equation determines uniquely 
$X[i+1]\in\mf g_{i+1}\otimes\mc V(\mf p)$ (note that $\mf n\cap\mf g_{i+1}=\mf g_{i+1}$ for $i\geq0$)
and $w[i]\in V_i\otimes\mc V(\mf p)$.
Moreover, since $\wt(A)=1$, we also get that, necessarily, 
$\wt(x[i+1])=0$ and $\wt(w[i])=1$.
This proves part (a).

For part (b), consider first the homogeneous components of $X\in\mf n\otimes\mc V(\mf p)$ 
and $w\in V \otimes\mc V(\mf p)$
of degree $0$, with respect to the polynomial degree decomposition \eqref{deg1}:
$X(0)\in\mf n\otimes\mb F\simeq\mf n$ and $w(0)\in V\otimes\mb F\simeq V$.
Clearly, $\deg(f\otimes1)=0$ and $\deg(q)=1$.
By looking at the terms of polynomial degree $0$
in both sides of the equation $q^X=w$, we get
$$
\big(e^{\ad X(0)}-1\big)f=w(0)\,.
$$
By an inductive argument on the $\ad x$-eigenvalues,
similar to the one used in the proof of (a),
it is not hard to show that, necessarily, $X(0)=0$ and $w(0)=0$.
Next, 
we study the homogeneous components of polynomial degree $1$:
$X(1)\in\mf n\otimes\mc V(\mf p)(1)=\mf n\otimes\mb F[\partial]\mf p$ 
and $w\in V\otimes\mc V(\mf p)(1)=\in V \otimes\mb F[\partial]\mf p$.
By looking at the terms of polynomial degree $1$
in both sides of the equation $q^X=w$, we get
\begin{equation}\label{20120724:eq1}
w(1)+[f\otimes1,X(1)]=q-X(1)'\,.
\end{equation}
By definition, $w(1)=\sum_{j\in J}v^j\otimes w_j(1)$,
and $q=\sum_{j\in P}v^j\otimes v_j$.
Equation \eqref{20120724:eq1} thus implies that
$$
w_j(1)-v_j\in\partial\mb F[\partial]\mf p\,,
$$
namely $w_j$ admits the decomposition as in \eqref{20120724:eq2}.
The conditions on the conformal weights of the elements $w_j$ and $g_j$
immediately follow from the fact that $\Delta(w)=1$.

Finally, we prove part (c). 
First, we prove that the elements $\{w_j\}_{j\in J}$ are differentially algebraically independent,
that is they generate a differential polynomial algebra.
For this, introduce in $\mc V(\mf p)$ the differential polynomial degree
$\text{dd}(v_j^{(n)})=n+1$ for every basis element $v_j\in\mf p$ and $n\in\mb Z_+$.
Suppose, by contradiction, that $P(w_1,\dots,w_r)=\sum w_{i_1}^{(n_1)}\dots w_{i_s}^{(n_s)}=0$
is a non trivial differential polynomial relation among the $w_j$'s.
If we let $P_0$ be the homogeneous component of $P$ of minimal differential polynomial degree,
then this relation can be written as
$P_0(v_1,\dots,v_r)+$ (stuff of higher differential polynomial degree in the $v_i$'s) $=0$.
Hence $P_0(v_1,\dots,v_r)=0$, contradicting the fact that the elements 
$v_1,\dots,v_r\in\mf g^f\subset\mf p$
are differentially algebraically independent in $\mc V(\mf p)$.

Next, we prove that all the coefficients $w_j,\,j\in J$, lie in $\mc W$.
In view of Corollary \ref{20120511:cor}, this is equivalent to show that
the differential polynomials $w_j\in\mc V(\mf p)$ are gauge invariant:
$w_j(q^A)=w_j$ for every $A\in\mf n\otimes\mc V(\mf p)$.
First, note that, obviously, $w(q^A)=\sum_{j\in J}v^j\otimes w_j(q^A)$ lies in $V\otimes\mc V(\mf p)$.
On the other hand, by definition of gauge transformation, we have
$$
\begin{array}{l}
\displaystyle{
\vphantom{\Big(}
w(q^A)=q^{X(q^A)}(q^A)=
e^{\ad X(q^A)}(\partial+f\otimes1+q^A)-\partial-f\otimes1
} \\
\displaystyle{
\vphantom{\Big(}
=e^{\ad X(q^A)}e^{\ad A}(\partial+f\otimes1+q)-\partial-f\otimes1\,.
}
\end{array}
$$
By the Baker-Campbell-Hausdorff formula,
there exists $\widetilde A\in\mf n\otimes\mc V(\mf p)$ such that
$e^{\ad X(q^A)}e^{\ad A}=e^{\ad\widetilde A}$.
Hence, the above equation reads
$w(q^A)=q^{\widetilde A}\in V\otimes\mc V(\mf p)$.
By the uniqueness of $X\in\mf n\otimes\mc V(\mf p)$ and $w\in V\otimes\mc V(\mf p)$ in part (a),
it follows that, necessarily, $\widetilde A=X$ and $w(q^A)=w$, as we wanted.

To conclude the proof of part (b), we are left to show that all the elements of $\mc W$
are differential polynomials in $w_1,\dots,w_r$.
Indeed, if $g\in\mc W$, then by Corollary \ref{20120511:cor} it is gauge invariant.
Hence, in particular, $g=g(q^X)=g(w)$, namely it is expressed as a differential polynomial
in the elements $w_j,\,j\in J$
(here we are using the obvious fact that, if we write $w$ in basis $\{q^i\}_{i\in P}$ of $\mf m^\perp$
as $\sum_{i\in P}q^i\otimes h_i$,
then the elements $h_i$ are linear combinations of the $w_j$'s).
\end{proof}

\begin{corollary}\label{daniele}
The Poisson vertex algebra $\mc W$, with the $\lambda$-bracket $\{g_\lambda h\}_{z,\rho}$,
is independent of the choice of the isotropic subspace $\mf l\subset\mf g_{\frac12}$ for $z=0$,
and also for arbitrary $z$, 
provided that $s\in\Ker(\ad\mf g_{\geq\frac12})$ is the same.
\end{corollary}
\begin{proof}
Let $\mf l_1\subset\mf l_2\subset\mf g_\frac12$ be isotropic subspaces and
$\mf m_1\subset\mf n_1$ and $\mf m_2\subset\mf n_2$ the corresponding nilpotent
subalgebras of $\mf g$ defined in \eqref{subalgebras}. Then
\begin{equation}\label{inclusions}
\mf m_1\subset\mf m_2\subset\mf n_2\subset\mf n_1\,.
\end{equation}
Let $I_i=\left(m-\kappa(f\mid m)\mid m\in\mf m_i\right)\subset{\mc V(\mf g)}$ and
$\widetilde{\mc W}_i\subset\mc V(\mf g)$ be defined as in \eqref{wtilde}, for $i=1,2$.
Clearly, by \eqref{inclusions}, $I_1\subset I_2$, from which follows easily that
$\widetilde{\mc W}_1\subset\widetilde{\mc W}_2$.
Hence, by Remark \ref{hamred}, we have a differential algebra homomorphism
$\phi:\mc W_1\rightarrow\mc W_2$, where $\mc W_i$, $i=1,2$, is the classical
$\mc W$-algebra corresponding to $\mf l_i$.
For $z=0$, or, for arbitrary z, provided that
$s=s_1=s_2\in\Ker(\ad\mf g_{\geq\frac12})$, this is a PVA 
homomorphism. Indeed, in this case,
$\widetilde{\mc W}_1\subset\widetilde{\mc W}_2$ is a Poisson vertex subalgebra.

We want to show that $\phi$ is a differential algebra isomorphism.
Fix $\mf p_1,\mf p_2\subset\mf g_{\leq\frac12}$ be complementary spaces to $\mf m_1$
and $\mf m_2$ in $\mf g$ respectively: 
$\mf g=\mf m_1\oplus\mf p_1=\mf m_2\oplus\mf p_2$.
By the arbitrariness of the choice of these complementary subspaces (see Remark \ref{hamred}), 
we may assume $\mf p_1\supset\mf p_2$. Let us denote by $\rho_i$, $i=1,2$, the
differential algebra homomorphism defined in \eqref{rho} corresponding to $\mf p_i$.
We have a differential algebra homomorphism induced by the following diagram
\begin{equation}\label{comdiag}
\xymatrix
{\mc V(\mf g)\ar[r]^{\rho_1}\ar[d]_{\rho_2}&\mc V(\mf p_1)\ar[ld]^{\widetilde\phi}\\
\mc V(\mf p_2)&
}
\end{equation}
(it exists since $\Ker\rho_1=I_1\subset\Ker\rho_2=I_2$).
It is easy to check that $\widetilde\phi$ is the differential algebra homomorphism induced
by the projection map $\pi_{\mf p_2}:\mf p_1\rightarrow\mf p_2$, and the restriction
to $\mc W_1=\mc V(\mf p_1)^{\mb F[\partial]\mf n}\subset\mc V(\mf p_1)$
is the differential algebra homomorphism $\phi:\mc W_1\rightarrow\mc W_2$
constructed above.
Let $q_i\in\mf m_i^{\perp}\otimes\mc V(\mf p_i)$, for $i=1,2$, as in \eqref{q}.
We note that
the choice of the vector space $V$ in Theorem \ref{20120511:thm2} does not
depend on the choice of the isotropic subspace $\mf l$
(indeed $\ad f:\mf g_{\frac12}\rightarrow\mf g_{-\frac12}$ is an isomorphism).
Hence, we may choose $V\subset\mf m_i^\perp$ such that 
$\mf m_i^{\perp}=[f,\mf n_i]\oplus V$, for $i=1,2$. Since the restriction of the
differential map $\phi$ to $\mf p_1\subset\mc V(\mf p_2)$ is the projection
map $\pi_{\mf p_2}$, we also note that
$(\id\otimes\phi)q_1=q_2$. Let $X\in\mf n_1\otimes\mc V(\mf p_1)$ such that
$q_1^X\in V\otimes\mc V(\mf p_1)$.
Then $(\id\otimes\phi)q_1^X=q_2^{(\id\otimes\tilde\phi)X}\in V\otimes\mc V(\mf p_2)$.
By the uniqueness argument in the proof of Theorem \ref{20120511:thm2},
it follows that $(\id\otimes\phi)X\in\mf n_2\otimes\mc V(\mf p_2)$.
By Theorem \ref{20120511:thm2}(b),
$\phi$ maps the generators of $\mc W_1$ to the generators of $\mc W_2$.
Hence it is an isomorphism.

Finally, if we take $\mf l_1=0$, it follows that the PVA structure obtained for $z=0$,
or, for arbitrary $z$, provided that $s=s_1=s_2\in\Ker(\ad\mf g_{\geq\frac12})$,
on $\mc W_2$ does not depend on the choice of the isotropic subspace $\mf l_2$.
\end{proof}
\begin{remark}\label{rem:different_l}
It is proved in \cite{GG02} for finite $\mc W$-algebras, 
by a method not applicable in our setup,
that 
they are independent of the choice of the isotropic subspace $\mf l\subset\mf g_{\frac12}$.
\end{remark}

\begin{definition}
Let $\mc V$ be a PVA. An element $L\in\mc V$ is called a \emph{Virasoro element}
of \emph{central charge} $c\in\mb F$ if
\begin{equation}\label{virasorofield}
\{L_\lambda L\}=(\partial+2\lambda)L+c\lambda^3
+\alpha\lambda
\,,
\end{equation}
for some $\alpha\in\mb F$.
(The term $\alpha\lambda$ can be removed by replacing $L$ with $L+\frac{\alpha}{2}$.)
An element $a\in\mc V$ is called an \emph{eigenvector} of
\emph{conformal weight} $\Delta_a\in\mb F$ if 
\begin{equation}\label{quasiprimary}
\{L_\lambda a\}=(\partial+\Delta_a\lambda)a+o(\lambda^2)\,.
\end{equation}
It is called a \emph{primary element} of conformal weight $\Delta_a$
if $\{L_\lambda a\}=(\partial+\Delta_a\lambda)a$.
\end{definition}

\begin{proposition}\label{daniele2}
Consider the PVA $\mc W$ with $\lambda$-bracket $\{\cdot\,_\lambda\,\cdot\}_{z,\rho}$
defined by equation \eqref{20120511:eq3}.
\begin{enumerate}[(a)]
\item
The following element lies in $\mc W$:
$$
L=
\rho\Big(\frac12\sum_{i\in I}u^iu_i+x^{\prime}
+\frac12\sum_{r\in R}v^r\partial v_r\Big)\,\in\mc W\,,
$$
where, as before, 
$\{u_i\}_{i\in I}$ and $\{u^i\}_{i\in I}$ are bases of $\mf g$ 
dual with respect to $\kappa$,
$\{v^r\}_{r\in R}$ and $\{v_r\}_{r\in R}$
are bases of $\mf n\cap\mf p$ $\omega$-dual, 
i.e. such that \eqref{20130206:eq2} holds,
and $\rho$ is the map \eqref{rho}.
Furthermore, the $\lambda$-bracket of $L$ with itself is
\begin{equation}\label{virasoro}
\{L_\lambda L\}_{z,\rho}
=(\partial+2\lambda)L-\kappa(x\mid x)\lambda^3+2\kappa(f\mid s)z\lambda\,.
\end{equation}
In particular, $L\in\mc W$ is a Virasoro element
of central charge $c=-\kappa(x\mid x)$.
\item
Assume that  $\mf p\subset\mf g$ is compatible 
with the $\ad x$-eigenspace decomposition \eqref{dec},
and consider the generators $w_j=v_j+g_j\in\mc W,\,j\in J$,
provided by Theorem \ref{20120511:thm2}(b),
where $\{v_j\}_{j\in J}$ is a basis of $\mf g^f$ 
consisting of $\ad x$-eigenvectors:
$[x,v_j]=(1-\Delta_j)v_j,\,j\in J$ (with $\Delta_j\geq1$).
Then $w_j$ is an $L$-eigenvector of conformal weight $\Delta_j$ for $z=0$.
\item
The PVA $\mc W$ is graded, as an algebra, by conformal weights:
$\mc W=\mb F\oplus\mc W\{1\}\oplus\mc W\{\frac32\}\oplus\mc W\{2\}\oplus\dots$.
Moreover, for $i=1$ or $\frac32$, $\mc W\{i\}$ is spanned over $\mb F$ by the generators
$w_j$ such that $\Delta(w_j)=1+\dx(v^j)=i$,
and all of them are primary elements for $z=0$.
\end{enumerate}
\end{proposition}
\begin {proof}
It is straightforward to check that
$$
L^{\mf g}=
\frac12\sum_{i\in I}u^iu_i+x^{\prime}\in\mc V(\mf g)
$$
is a Virasoro element of the affine PVA $\mc V(\mf g)$
(with $\lambda$-bracket \eqref{lambda})
with central charge $-\kappa(x\mid x)$:
\begin{equation}\label{20130206:eq9}
\{L^{\mf g}{}_\lambda L^{\mf g}\}_z
=(\partial+2\lambda)L^{\mf g}-\kappa(x\mid x)\lambda^3
+z(\partial+2\lambda)[x,s]
\,,
\end{equation}
and that, for $z=0$, $a\in\mf g_j$ is an eigenvector of conformal weight $\Delta=1-j$
(primary provided that $\kappa(x\mid a)=0$), more precisely:
\begin{equation}\label{20120723:eq1}
\begin{array}{l}
\displaystyle{
\{a_\lambda L^{\mf g}\}_z=-ja'+(1-j)a\lambda+\kappa(a\mid x)\lambda^2
-z\left([a,s]+j\kappa(s\mid a)\lambda\right)
\,,}\\
\displaystyle{
\{L^{\mf g}{}_\lambda a\}_z=(\lambda+\partial)a-ja\lambda-\kappa(a\mid x)\lambda^2
-z\left([s,a]+j\kappa(s\mid a)\lambda\right)
\,.
}
\end{array}
\end{equation}
For $a\in\mf n\cap\mf g_j$, we have, by the definition \eqref{20120511:eq1}
of the action of $\mb F[\partial]\mf n$ on $\mc V(\mf p)$,
$$
a\,^\rho_\lambda\,\rho(L^\mf g)=\rho\{a_\lambda\rho(L^{\mf g})\}_z=
\rho\{a_\lambda L^{\mf g}\}_z
=\rho\Big(-ja'+(1-j)a\lambda+\kappa(a\mid x)\lambda^2\Big)\,,
$$
where in the second equality we used Lemma \ref{20120511:lem1}(b).
Since $a\in \mf n\subset\mf g_{\geq\frac12}$, we have $\kappa(a\mid x)=0$. 
Moreover,
$\rho(a)=\pi_{\mf p}(a)+\kappa(f\mid a)$,
so that $\rho(a')=\partial\pi_{\mf p}(a)$,
and $(1-j)\rho(a)=(1-j)\pi_{\mf p}(a)$.
Recall also that, if $a\in\mf n\cap\mf g_j$, 
then $\pi_{\mf p}(a)\in\mf n\cap\mf p$ is zero unless $j=\frac12$.
Hence,
\begin{equation}\label{20130206:eq4}
a\,^\rho_\lambda\,\rho(L^\mf g)
=
-\frac12\partial\pi_{\mf p}(a)+\frac12\lambda\pi_{\mf p}(a)\,.
\end{equation}
Next, if $a\in\mf n$ and $v\in\mf n\cap\mf p$, 
we have $\kappa(s\mid [a,v])=0$, since $s\ker(\ad\mf n)$,
and $\kappa(a\mid v)=0$, since $\mf n\subset\mf g_{\geq\frac12}$.
Hence,
$$
\{a_\lambda v^r\partial v_r\}_z
=
[a,v^r] v_r^\prime
+v^r(\lambda+\partial)[a,v_r]\,.
$$
Therefore,
if $a\in\mf n\cap\mf g_j$, we have, by the definition \eqref{20120511:eq1}
of the action of $\mb F[\partial]\mf n$ on $\mc V(\mf p)$ and Lemma \ref{20120511:lem1}(b),
$$
a\,^\rho_\lambda\,\rho\Big(\frac12\sum_{r\in R}v^r\partial v_r\Big)
=\frac12\sum_{r\in R}\rho\{a_\lambda v^r\partial v_r\}_z
=\frac12\sum_{r\in R}\rho
\big([a,v^r] v_r^\prime+v^r(\lambda+\partial)[a,v_r]\big)\,.
$$
Note that, if $a\in\mf n\cap\mf g_j$ with $j\geq1$
and $v\in\mf n\cap\mf p\subset\mf g_{\geq\frac12}$,
then $[a,v]\in\mf g_{\geq\frac32}$, so that $\rho[a,v]=0=\kappa(f\mid [a,v])$.
If instead $a\in\mf n\cap\mf g_{\frac12}$
and $v\in\mf n\cap\mf p$,
then $[a,v]\in\mf g_{\geq1}$, so that $\rho[a,v]=\kappa(f\mid [a,v])$.
Therefore,
$$
a\,^\rho_\lambda\,\rho\Big(\frac12\sum_{r\in R}v^r\partial v_r\Big)
=\frac12\sum_{r\in R}
\rho\big(\kappa(f\mid [a,v^r]) v_r^\prime+\lambda\kappa(f\mid [a,v_r])v^r\big)\,.
$$
In conclusion, by the completeness relations \eqref{20130206:eq3}, we get
\begin{equation}\label{20130206:eq5}
a\,^\rho_\lambda\,\rho\Big(\frac12\sum_{r\in R}v^r\partial v_r\Big)
=\frac12\partial\pi_{\mf p}(a)
-\frac12\lambda\pi_{\mf p}(a)\,.
\end{equation}
Equations \eqref{20130206:eq4} and \eqref{20130206:eq5}
imply that $\rho\Big(L^{\mf g}+\frac12\sum_{r\in R}v^r\partial v_r\Big)$ lies in $\mc W$,
proving the first claim in (a).
It remains to prove that $L\in\mc W$ satisfies equation \eqref{virasoro}. 
By definition of the PVA structure on $\mc W$ \eqref{20120511:eq3}, we have,
using Corollary \ref{20130206:cor}(d)
\begin{equation}\label{20130206:eq10}
\begin{array}{c}
\displaystyle{
\{L_\lambda L\}_{z,\rho}=\rho\{L_\lambda L\}_z
=\rho\Bigg(
\{L^{\mf g}{}_\lambda L^{\mf g}\}_z
+\frac12\sum_{r\in R}\{L^{\mf g}{}_\lambda v^r\partial v_r\}_z
} \\
\displaystyle{
+\frac12\sum_{r\in R}\{v^r\partial v_r{}_\lambda L^{\mf g}\}_z
+\frac14\sum_{q,r\in R}
\{v^q\partial v_q{}_\lambda v^r\partial v_r\}_z
\Bigg)\,.
}
\end{array}
\end{equation}
Recall that, for $r\in R$, we have $v^r,v_r\in\mf n\cap\mf p\subset\mf g_{\frac12}$.
Note also that, since $s$ lies in $\ker(\ad\mf n)\subset\mf g_{\geq0}$, 
we have $\kappa(s\mid v)=0$ for all $v\in\mf g_{\geq\frac12}$.
Hence, 
using \eqref{20120723:eq1}, we get
$$
\begin{array}{l}
\displaystyle{
\{L^{\mf g}{}_\lambda v^r\partial v_r\}_z
=
\{L^{\mf g}{}_\lambda v^r\}_z \partial v_r
+v^r(\lambda+\partial)\{L^{\mf g}{}_\lambda v_r\}_z
} \\
\displaystyle{
=
\big((\lambda+\partial)v^r-\frac12\lambda v^r\big)\partial v_r
+v^r(\lambda+\partial)
\big((\lambda+\partial)v_r-\frac12\lambda v_r\big)
} \\
\displaystyle{
=
(2\lambda+\partial)\big( v^r\partial v_r \big)
+\frac12\lambda^2 v^rv_r
\,,
}
\end{array}
$$
which implies, 
\begin{equation}\label{20130206:eq6}
\{L^{\mf g}{}_\lambda \sum_{r\in R} v^r\partial v_r\}_z
=
(2\lambda+\partial)\sum_{r\in R} v^r\partial v_r
\,,
\end{equation}
since, by skew-symmetry of $\omega$, 
we have $\sum_{r\in R}v^rv_r=0$.
By skew-symmetry of the $\lambda$-bracket we also get
\begin{equation}\label{20130206:eq7}
\{\sum_{r\in R}v^r\partial v_r {}_\lambda L^{\mf g}\}_z
=
(2\lambda+\partial)\sum_{r\in R} v^r\partial v_r
\,.
\end{equation}
Furthermore, we have, for $q,r\in R$,
$$
\begin{array}{l}
\displaystyle{
\vphantom{\Big(}
\{v^q\partial v_q\,_\lambda\,v^r\partial v_r\}_z
=
(\partial v_r)\{v^q\,{}_{\lambda+\partial}v^r\}_z(\partial v_q)
+v^r(\lambda+\partial)\{v^q\,{}_{\lambda+\partial}v_r\}_z(\partial v_q)
} \\
\displaystyle{
\vphantom{\Big(}
+(\partial v_r)\{v_q\,{}_{\lambda+\partial}v^r\}_z(-\lambda-\partial) v^q
+v^r(\lambda+\partial)\{v_q\,{}_{\lambda+\partial}v_r\}_z(-\lambda-\partial) v^q
} \\
\displaystyle{
\vphantom{\Big(}
=
[v^q,v^r](\partial v_r)(\partial v_q)
+v^r(\lambda+\partial)\big([v^q,v_r](\partial v_q)\big)
} \\
\displaystyle{
\vphantom{\Big(}
-[v_q,v^r](\partial v_r)(\lambda+\partial)v^q
-v^r(\lambda+\partial)\big([v_q,v_r](\lambda+\partial) v^q\big)
}
\end{array}
$$
Note that for $q,r\in R$, we have $v^q,v^r\in\mf n\cap\mf p\subset\mf g_{\frac12}$,
so that $\rho(v^r)=v^r$ and $\rho([v^q,v^r])=\kappa(f\mid [v^q,v^r])$ (and similarly with lower indices).
Hence, applying $\rho$ and summing over all indices in the above equation, we get,
using the completeness relations \eqref{20130206:eq3},
\begin{equation}\label{20130206:eq8}
\rho\Big\{\sum_{q\in R}v^q\partial v_q\,_\lambda\,\sum_{r\in R}v^r\partial v_r\Big\}_z
=-2(2\lambda+\partial)\sum_{r\in R}v^r\partial v_r
\,.
\end{equation}
Combining equations \eqref{20130206:eq9}, \eqref{20130206:eq6}, \eqref{20130206:eq7} 
and \eqref{20130206:eq8} into \eqref{20130206:eq10}, we get
\begin{equation}\label{20130206:eq11}
\{L_\lambda L\}_{z,\rho}
=
(\partial+2\lambda)L
-\kappa(x\mid x)\lambda^3+z(\partial+2\lambda)\rho([x,s])
\,.
\end{equation}
We claim that $[x,s]$ lies in $\mf m$.
Since, by assumption, $s\in\ker(\ad\mf n)$,
we have $s\in\mf g_{\geq0}$.
Let then $s=s_0+s_{\frac12}+s_1+\dots$, with $s_j\in\mf g_j$,
and clearly $s_j\in\ker(\ad\mf n)$ for all $j\geq0$.
We have $[x,s]=\frac12s_{\frac12}+s_1+\dots$. 
Therefore, to prove that $[x,s]$ lies in $\mf m$,
it suffices to prove that $s_\frac12\in\mf m$.
But $s_{\frac12}$ commutes with $\mf n$, hence with $\mf l^{\perp\omega}\subset\mf n$,
so that $\omega(s_{\frac12},n)=\kappa(f\mid [s_{\frac12},n])=0$ for all $n\in\mf l^{\perp\omega}$.
It follows that $s_{\frac12}\in(\mf l^{\perp\omega})^{\perp\omega}=\mf l\subset\mf m$,
proving the claim.
Therefore, $\rho([x,s])=\kappa(f\mid [x,s])$,
and equation \eqref{20130206:eq11}
reduces to equation \eqref{virasoro},
proving part (a).

Next, we first prove part (b) in the special case when $\mf l\subset\mf g_{\frac12}$
is maximal isotropic with respect to $\omega(\cdot\,,\,\cdot)$,
namely when $\mf m=\mf n$, and therefore $L=\rho(L^{\mf g})$.
Letting $z=0$ in the second equation of \eqref{20120723:eq1}
we have, for $a\in\mf g_{1-\Delta}$,
$$
\{L^{\mf g}\,{}_\lambda a\}_{z=0}=(\partial+\Delta\lambda)a+O(\lambda^2)\,.
$$
It immediately follows, by the Leibniz rule and Lemma \ref{20120511:lem2}(c),
that, if $w\in\mc W\{\Delta\}$, then
$$
\{L_\lambda w\}_{z=0,\rho}=\rho\{\rho(L^{\mf g})\,{}_\lambda w\}_{z=0}
=\rho\{L^{\mf g}\,{}_\lambda w\}_{z=0}
=(\partial+\Delta\lambda)w+O(\lambda^2)\,.
$$
Hence, if $\mf l=\mf l^{\perp\omega}$, 
part (b) follows from Theorem \ref{20120511:thm2}(b).
Next, let us prove (b) for arbitrary isotropic subspace $\mf l_1\subset\mf g_{\frac12}$.
Let $\mf l_2\subset\mf g_{\frac12}$ be a maximal isotropic subspace containing $\mf l_1$,
and let $\mf p_1$ and $\mf p_2$ be subspaces of $\mf g_{\leq\frac12}$
complementary to $\mf m_1=\mf l_1\oplus\mf g_{\geq\frac12}$
and to $\mf m_2=\mf l_2\oplus\mf g_{\geq\frac12}$ respectively.
We thus have
$$
0\subset\mf l_1\subset\mf l_2=\mf l_2^{\perp\omega}\subset\mf l_1^{\perp\omega}\subset\mf g_{\frac12}\,.
$$
Let $\mc W_1$ and $\mc W_2$ be the $\mc W$-algebras corresponding
to the choices $\mf l=\mf l_1$ and $\mf l_2$ respectively,
and consider the PVA isomorphism $\varphi:\,\mc W_{1}\to\mc W_{2}$
provided by Corollary \ref{daniele}.
By the proof of Corollary \ref{daniele},
these isomorphism maps the generators $w_j$ of $\mc W_{1}$, homogeneous 
with respect to conformal weight,
to homogeneous generators of $\mc W_{2}$.
We next want to see how the Virasoro elements are related by these isomorphisms.
Fix a basis $\{\tilde v^r\}_{r\in J_1}$ of $\mf l_1$,
extend it to a basis of $\mf l_2$ by adding elements $\{v^r\}_{r\in J_2}\subset\mf p_1$,
further extend it to a basis of $\mf l_1^{\perp}$ 
by adding elements $\{v_r\}_{r\in J_2}\subset\mf p_2$
which are $\omega$-dual to the $v^r$'s, i.e. $\omega(v^q,v_r)=\delta_{q,r}$,
and finally further extended this bases to a basis of $\mf g_{\frac12}$ 
by adding elements $\{\tilde v_r\}_{r\in J_1}\subset\mf p_1$
which are $\omega$-dual to the $\tilde v^r$'s, i.e. $\omega(\tilde v^q_i,\tilde v_r)=\delta_{q,r}$.
Then,
the Virasoro elements in $\mc W_1$ and $\mc W_2$ provided by part (a)
are respectively
$$
\begin{array}{l}
L_1=\rho_1\Big(L^{\mf g}+\sum_{r\in J_2}v^r\partial v_r-\sum_{r\in J_2}v_r\partial v^r\Big)
\,,\,\,
L_2=\rho_2(L^{\mf g})
\,,
\end{array}
$$
where $\rho_i,\,i=1,2$ are the differential algebra homomorphisms $\mc V(\mf g)\to\mc V(\mf p_i)$
in \eqref{rho}, for $\mf p=\mf p_i$.
By the definition of the isomorphism $\varphi:\,\mc W_1\to\mc W_2$ defined in the proof
of Corollary \ref{daniele}, we have
$$
\varphi(L_1)=\rho_2\Big(L^{\mf g}+\sum_{r\in J_2}v^r\partial v_r-\sum_{r\in J_2}v_r\partial v^r\Big)
\,,
$$
which is equal to $L_2$ since $v^r\in\mf l_2$ for every $r$, so that $\rho_2(v^r)=0$.
But then, take a  generator $w_j\in\mc W_1\{\Delta\}$,
which is mapped by $\varphi$ to a generator $\varphi(w_j)\in\mc W_2\{\Delta\}$.
Since claim (b) holds in the case of maximal isotropic subspace $\mf l$,
we have
$\{{L_2}_\lambda \varphi(w_j)\}_{z=0,\rho_2}=(\partial+\Delta\lambda)\varphi(w_j)+O(\lambda^2)$.
Therefore, since $\varphi$ is a PVA isomorphism mapping $L_1$ to $L_2$,
we conclude that
$\{{L_1}_\lambda w_j\}_{z=0,\rho_1}=(\partial+\Delta\lambda)w_j+O(\lambda^2)$,
proving the claim in (b).

For (c), one only needs to prove that all elements in $\mc W\{1\}$ and $\mc W\{\frac32\}$
are primary elements.
First, if $w_j\in\mc W_{\frac32}$ we have
$\{L_\lambda w_j\}=\partial w_j+\frac32\lambda w_j+L_{(2)}w_j\lambda^2+\dots$.
But for $n\geq2$ the element $L_{(n)}w_j\in\mc W$ has conformal weight
$\Delta(L_{(n)}w_j)=2+\frac32-n-1\leq\frac12$.
Hence it must be $L_{(n)}w_j=0$, proving that $w_j$ is primary.
Finally, if $w_j\in\mc W_1$, then, by Theorem \ref{20120511:thm2}(b),
it must be of the form
$w_j=v_j+\sum_k a_kb_k$, with $v_j\in\mf g^f_0=\mf g^f\cap\mf g_0$ 
and $a_k,b_k\in\mf p_{\frac12}=\mf p\cap\mf g_{\frac12}$.
But by equation \eqref{20120723:eq1}, for $z=0$
all the elements $v_j\in\mf g^f_0$ and $a_k,b_k\in\mf g_{\frac12}$ are primary
elements with respect to the Virasoro element $L^{\mf g}$ in $\mc V(\mf g)$
(we use the fact that $\kappa(v_j\mid x)=\frac12\kappa(v_j\mid [e,f])=0$ for $v_j\in\mf g^f_0$).
Hence, by the Leibniz rule $w_j$ is also a primary element
with respect $L^{\mf g}\in\mc V(\mf g)$ for $z=0$,
and therefore, by Lemma \ref{20120511:lem2}(c), $w_j$ is a primary element
with respect to $L\in\mc W$ (for $z=0$).
\end{proof}

\begin{corollary}\label{20120726:cor}
Assume that $\mf p\subset\mf g$ is compatible with the $\ad x$-eigenspace 
decomposition \eqref{dec} (in particular $\mf g^f\subset\mf p$).
Let $\{v_j\}_{j\in J}$ be a basis of $\mf g^f$ consisting of $\ad x$-eigenvectors: 
$[x,v_j]=(1-\Delta_j)v_j,\,j\in J$ (with $\Delta_j\geq1$).
Let $\{\tilde w_j\}_{j\in J}$ be an arbitrary collection of elements of $\mc W$
of the form
$\tilde w_j=v_j+\tilde g_j$,
where 
$$
\tilde g_j=\sum b_1^{(m_1)}\dots b_s^{(m_s)}\in\mc V(\mf p)\{\Delta_j\}\,,
$$
is a sum such of products of $\ad x$-eigenvectors $b_i\in\mf p$,
such that
$$
(1-\dx(b_1))+\dots+(1-\dx(b_s))+m_1+\dots+m_s=\Delta_j\,,
$$
and $s+m_1+\dots+m_s>1$.
(Such a collection of vectors exists by Theorem \ref{20120511:thm2}(b)).
Then:
\begin{enumerate}[(a)]
\item
The elements $\{\tilde w_j\}_{j\in J}$
form a set of generators for the algebra of differential polynomials $\mc W$.
\item
For $\Delta_j=1$ or $\frac32$,
the generators $\tilde w_j$ are uniquely determined 
by the corresponding basis elements $v_j\in\mf g^f_j$,
in particular, we have $\tilde w_j=w_j$ from Theorem \ref{20120511:thm2}(b).
\end{enumerate}
\end{corollary}
\begin{proof}
By Theorem \ref{20120511:thm2}(c), each element $\tilde w_j$ 
is a differential polynomial in the generators $\{w_k\}_{k\in J}$
defined in Theorem \ref{20120511:thm2}(b).
On the other hand,
if we order the generators according to their increasing conformal weights,
each $\tilde w_j$ is equal to $w_j$ plus a differential polynomial $P_j$ in the elements $w_k$'s
with $k<j$.
Hence, each $w_j$ can be expressed as a differential polynomial in the $\tilde w_k$'s,
proving part (a).
Part (b) follows from the same argument and the observation that the conformal weight 
of the generators of $\mc W$ are greater than or equal to $1$,
so that $P_j=0$ if $\Delta_j=1$ or $\frac32$.
\end{proof}

\subsection{Examples of classical \texorpdfstring{$\mc W$}{W}-algebras}

\begin{example}[Virasoro-Magri and Gardner-Faddeev-Zakharov PVAs]\label{W-sl2}
Let $\mf g=\mf{sl}_2$ with standard generators $f,h=2x,e$ and fix
$\kappa(a\mid b)=\tr(ab)$, for any $a,b\in\mf{sl}_2$.
With respect to the $\ad x$-eigenspaces decomposition \eqref{dec},
we have $\mf n=\mf m=\mb Fe$,
$\mf m^{\perp}=\mb Fh\oplus\mb Fe$,
$[f,\mf n]=\mb Fh\subset\mf m^\perp$.
We fix the subspace $\mf p=\mb Fh\oplus\mb Ff\subset\mf{sl}_2$ complementary to $\mf m$,
and the subspace $V=\mb F e\subset\mf m^\perp$ complementary to $[f,\mf n]$.
Then the element $X=e\otimes(-\frac{h}{2})\in\mf n\otimes\mc V(\mf p)$ brings
$q=\frac12 h\otimes h+e\otimes f\in\mf m^{\perp}\otimes\mc V(\mf p)$ to
$$
q^X=e\otimes\big(\frac{h^2}{4}+\frac{h^\prime}{2}+f\big)\in V\otimes\mc V(\mf p)\,.
$$
By Theorem \ref{20120511:thm2}, the corresponding $\mc W$-algebra is,
as differential algebra, equal to the algebra of differential polynomials 
$\mc W=\mb F[w,w',w'',\dots]\subset\mc V(\mf p)$,
where 
$$
w=\frac{h^2}{4}+\frac{h^\prime}{2}+f\in\mc V(\mf p)\,.
$$
We note that $w=L=\rho(L^{\mf{sl}_2})$ (see Proposition \ref{daniele2}),
namely it is a Virasoro element
for the $z=0$ $\lambda$-bracket. If we take $s=e\in\Ker(\ad\mf n)$,
by an easy computation we obtain
$$
\{w_\lambda w\}_{z,\rho}=(2\lambda+\partial)w-\frac{\lambda^3}{2}+2z\lambda\,.
$$
Hence, the two compatible Poisson structures $H,K\in\mc W[\lambda]$ associated 
to this family of $\lambda$-brackets via \eqref{masterformula} are
$H(\lambda)=(\partial+2\lambda)w-\frac{\lambda^3}{2}$,
known as the \emph{Virasoro-Magri} Poisson structure (of central charge $c=-\frac{1}{2}$),
and $K(\lambda)=-2\lambda$, known as the \emph{Gardner-Faddeev-Zakharov} 
Poisson structure (up to the factor $-2$).
\end{example}

\begin{example}\label{W-sl3-principale}
Let $\mf g=\mf{sl}_3$ and fix $\kappa(a\mid b)=\tr(ab)$ for $a,b\in\mf{sl}_3$.
Let $f\in\mf{sl}_3$ be its principal nilpotent element. In the matrix realization it is
$f=E_{21}+E_{32}$. 
We can extend $f$ to an $\mf{sl}_2$-triple $(f,h=2x,e)$,
with $x=E_{11}-E_{33}$.
In this case $\mf g_{\frac12}=0$, and hence $\mf n=\mf m\subset\mf{sl}_3$ 
is the nilpotent subalgebra of strictly upper triangular matrices.
Its orthogonal complement $\mf m^{\perp}$ consists of all upper triangular matrices.
We can fix $\mf p\subset\mf{sl}_3$ to be the subspace, complementary to $\mf m$,
consisting of all lower triangular matrices.
A basis of $\mf p$ is $\{h_1=E_{11}-E_{22},h_2=E_{22}-E_{33},E_{21},E_{31},E_{32}\}$.
Also, as a subspace $V\subset\mf m^\perp$
complementary to $[f,\mf n]\subset\mf m^\perp$ we choose, for example,
$V=\mf g^e=\Ker(\ad e)=\mb F(E_{12}+E_{23})\oplus\mb FE_{13}$. 
After a straightforward computation, one can find $X\in\mf n\otimes\mc V(\mf p)$ such that
$q^X=(E_{12}+E_{23})\otimes w_1+E_{13}\otimes w_2\in V\otimes\mc V(\mf p)$.
The answer is as follows:
$X=E_{12}\otimes a+E_{23}\otimes b+E_{13}\otimes c$, where
$$
a=-\frac13(2h_1+h_2)
\,\,,\,\,\,\,
b=-\frac13(h_1+2h_2)
\,\,,\,\,\,\,
c=\frac12(E_{32}-E_{21})-\frac16( h_1^2-h_2^2+h_1'- h_2')\,,
$$
and
\begin{align*}
w_1=&
\frac12(E_{21}+E_{32})+\frac16(h_1^2+h_1h_2+h_2^2)+\frac12(h_1'+h_2')\,, \\
w_2=&
E_{31}+\frac13h_1(E_{21}-2E_{32})+\frac13h_2(2E_{21}-E_{32})+\frac2{27}(h_1^3-h_2^3)
+\frac19 h_1h_2(h_1-h_2)
\\
&
+\frac12(E_{21}^\prime-E_{32}^\prime)+\frac16h_1(2h_1^\prime-h_2^\prime)
+\frac16h_2(h_1^\prime-2h_2^\prime)+\frac16( h_1''-h_2'')\,.
\end{align*}
It is easy to check that $w_1=\frac12 L=\rho(L^{\mf{sl}_3})$ (see Proposition \ref{daniele2}). 
Hence, by Theorem \ref{20120511:thm2},  
$L$ and $w_2$ generate the algebra of differential polynomials $\mc W$.
Letting
$s=E_{13}\in\Ker(\ad n)$, we get the following formulas for the
$\lambda$-bracket \eqref{wbrack} on the generators of the classical $\mc W$-algebra
\begin{align*}
\{L{}_\lambda L\}_{z,\rho}&=(2\lambda+\partial)L-2\lambda^3,\\
\{L{}_\lambda w_2\}_{z,\rho}&
=(3\lambda+\partial)w_2+3z\lambda,\\
\{w_2{}_\lambda w_2\}_{z,\rho}&
=
\frac13(2\lambda+\partial)L^2
-\frac16(\lambda+\partial)^3L-\frac16\lambda^3L
-\frac14\lambda(\lambda+\partial)(2\lambda+\partial)L+\frac16\lambda^5\,.
\end{align*}
In particular,  $w_2$ is a primary element of conformal weight $3$ for $z=0$.
The corresponding compatible Poisson structures $H,K\in\Mat_{2\times2}\mc W[\lambda]$ are
$$
\begin{array}{l}
\displaystyle{
H(\lambda)=
\left(\begin{array}{cc}
(2\lambda+\partial)L-2\lambda^3 & 
(3\lambda+2\partial)w_2 \\
\\
(3\lambda+\partial)w_2
& 
\begin{array}{c}
\frac13(2\lambda+\partial)L^2
-\frac16(\lambda+\partial)^3L-\frac16\lambda^3L\\
-\frac14\lambda(\lambda+\partial)(2\lambda+\partial)L+\frac16\lambda^5
\end{array}
\end{array}\right)
\,,} \\
\displaystyle{
K(\lambda)=
\left(\begin{array}{cc}
0 & -3\lambda \\
-3\lambda & 0
\end{array}\right)
\,.
}
\end{array}
$$
\end{example}

\begin{example}\label{W-sl3-minimal}
For $\mf g=\mf{sl}_3$, with the invariant bilinear form $\kappa(a\mid b)=\tr(ab)$,
consider the lowest root vector $f=E_{31}\in\mf{sl}_3$.
We can extend it to an $\mf{sl}_2$-triple $(f,h=2x,e)$,
with $x=\frac12E_{11}-\frac12E_{33}$, $e=E_{13}$.
In this case $\mf g_{\frac12}=\mb FE_{12}\oplus\mb FE_{23}$.
Let us choose $\mf l\subset\mf g_{\frac12}$ to be the maximal isotropic subspace $\mf l=\mb F(E_{12}+E_{23})=\mf l^{\perp\omega}$.
In this case, $\mf n=\mf m=\mb F(E_{12}+E_{23})\oplus\mb FE_{13}\subset\mf{sl}_3$, 
and its orthogonal complement 
$\mf m^{\perp}$ is generated by $E_{21}-E_{32}$ and all upper triangular matrices in $\mf{sl}_3$.
We can fix the subspace $\mf p\subset\mf{sl}_3$, complementary to $\mf m$,
with basis $\{g=E_{12}-E_{23},h_1=E_{11}-E_{22},h_2=E_{22}-E_{33},E_{21},E_{31},E_{32}\}$.
Also in this case, a subspace $V\subset\mf m^\perp$
complementary to $[f,\mf n]\subset\mf m^\perp$ is, for example,
$V=\mf g^e=\mb F(h_1-h_2)\oplus\mb FE_{12}\oplus\mb FE_{23}\oplus\mb FE_{13}$. 
Hence, we can find $X=(E_{12}+E_{23})\otimes a+E_{13}\otimes b\in\mf n\otimes\mc V(\mf p)$ 
such that $q^X=E_{13}\otimes w_1+E_{12}\otimes w_2+
E_{23}\otimes w_3+(h_1-h_2)\otimes w_4\in V\otimes\mc V(\mf p)$.
The answer is as follows:
$$
a=-\frac12g,
\qquad
b=-\frac12 h_1-\frac12 h_2\,,
$$
and
$$
\begin{array}{l}
\displaystyle{
w_1=
E_{31}-\frac3{64} g^4+\frac12g(E_{21}-E_{32})+\frac18 g^2(h_1-h_2)+\frac14( h_1+h_2)^2+\frac12 h_1'+\frac12 h_2'
\,,} \\
\displaystyle{
w_2=
E_{21}-\frac18g^3+\frac12gh_1+\frac12g'
\,\,,\,\,\,\,
w_3=
E_{32}+\frac18g^3+\frac12gh_2+\frac12g'
\,,} \\
\displaystyle{
w_4=
-\frac18 g^2+\frac16(h_1-h_2)
\,.
}
\end{array}
$$
It is not hard to check that $L=\rho(L^{\mf sl_3})=w_1+3w_4^2$
(see Proposition \ref{daniele2}). Hence, by Theorem \ref{20120511:thm2},
$\mc W$ is the algebra of differential polynomials in the generators $L,w_2,w_3,w_4$.
Letting $s=E_{12}+E_{23}\in\Ker(\ad \mf n)$, we get the following formulas for the
$\lambda$-brackets \eqref{wbrack} on the generators of $\mc W$:
\begin{align*}
\{L_\lambda L\}_{z,\rho}&=(2\lambda+\partial)L-\frac12\lambda^3,\\
\{L{}_\lambda w_2\}_{z,\rho}&=(\frac32\lambda+\partial)w_2+\frac32z\lambda,\\
\{L{}_\lambda w_3\}_{z,\rho}&=(\frac32\lambda+\partial)w_3+\frac32z\lambda,\\
\{L{}_\lambda w_4\}_{z,\rho}&=(\lambda+\partial)w_4,\\
\{w_2{}_\lambda w_2\}_{z,\rho}&=0,\\
\{w_2{}_\lambda w_3\}_{z,\rho}&=
-L+12w_4^2-3(2\lambda+\partial)w_4+\lambda^2,\\
\{w_2{}_\lambda w_4\}_{z,\rho}&=\frac12 w_2+\frac12 z,\\
\{w_3{}_\lambda w_3\}_{z,\rho}&=0,\\
\{w_3{}_\lambda w_4\}_{z,\rho}&=-\frac12 w_3-\frac12 z,\\
\{w_4{}_\lambda w_4\}_{z,\rho}&=\frac16 \lambda\,.
\end{align*}
%
We note that, with respect to the $z=0$ $\lambda$-bracket,
$w_2$ and $w_3$ are primary elements of conformal weight $\frac32$,
and $w_4$ is a primary element of conformal weight $1$.

We can also consider $\mf l=0$. Then $\mf l^{\perp\omega}=\mf g_{\frac12}$.
Hence, $\mf n$ consists of all strictly upper triangular matrices and
$\mf m=\mb FE_{13}\subset\mf n$.
The orthogonal complement 
$\mf m^{\perp}$ is spanned by $E_{21}$, $E_{32}$ and all upper triangular matrices in $\mf{sl}_3$.
We can fix the subspace $\mf p\subset\mf{sl}_3$, complementary to $\mf m$,
with basis $\{E_{12},E_{23},h_1=E_{11}-E_{22},h_2=E_{22}-E_{33},E_{21},E_{31},E_{32}\}$.
Again, as a subspace $V\subset\mf m^\perp$
complementary to $[f,\mf n]\subset\mf m^\perp$ we take $V=\mf g^e
=\mb F(h_1-h_2)\oplus\mb FE_{12}\oplus\mb FE_{23}\oplus\mb FE_{13}$. 
In this case, we can find $X=E_{12}\otimes a+E_{23}\otimes b+E_{13}\otimes c\in\mf n\otimes\mc V(\mf p)$ 
such that $q^X=E_{13}\otimes w_1+E_{12}\otimes w_2+
E_{23}\otimes w_3+(h_1-h_2)\otimes w_4\in V\otimes\mc V(\mf p)$.
We get the following answer:
$$
a=E_{23},
\qquad
b=-E_{12},
\qquad
c=-\frac12 h_1-\frac12 h_2
$$
and
$$
\begin{array}{l}
\displaystyle{
w_1=
E_{31}+E_{12}E_{21}+E_{23}E_{32}-\frac34E_{12}^2E_{32}^2+\frac14(h_1+h_2)^2
-\frac12E_{12}E_{23}(h_1-h_2)
} \\
\displaystyle{
\,\,\,\,\,\,\,\,\,\,\,\,
+\frac12E_{23}E_{12}^{\prime}
-\frac12E_{12}E_{23}^{\prime}+\frac12 h_1'+\frac12 h_2'
} \\
\displaystyle{
\vphantom{\bigg(}
w_2=E_{21}-E_{12}E_{23}^2-E_{23}h_1-E_{23}^{\prime}
\,\,,\,\,\,\,
w_3=
E_{32}-E_{12}^2E_{23}+E_{12}h_2+E_{12}^{\prime}
\,
} \\
\displaystyle{
w_4=
\frac12E_{12}E_{23}+\frac16(h_1-h_2)
\,.
}
\end{array}
$$
Also in this case we get $L=w_1+3w_4^2$.
Letting
$s=E_{13}\in\Ker(\ad \mf n)$, we get the following formulas for the
$\lambda$-brackets \eqref{wbrack} on the generators of $\mc W$:
\begin{align*}
\{L_\lambda L\}_{z,\rho}&=(2\lambda+\partial)L-\frac12\lambda^3+2z\lambda,\\
\{L{}_\lambda w_2\}_{z,\rho}&=(\frac32\lambda+\partial)w_2,\\
\{L{}_\lambda w_3\}_{z,\rho}&=(\frac32\lambda+\partial)w_3,\\
\{L{}_\lambda w_4\}_{z,\rho}&=(\lambda+\partial)w_4,\\
\{w_2{}_\lambda w_2\}_{z,\rho}&=0,\\
\{w_2{}_\lambda w_3\}_{z,\rho}&=
-L+12w_4^2-3(2\lambda+\partial)w_4+\lambda^2-z,\\
\{w_2{}_\lambda w_4\}_{z,\rho}&=\frac12 w_2,\\
\{w_3{}_\lambda w_3\}_{z,\rho}&=0,\\
\{w_3{}_\lambda w_4\}_{z,\rho}&=-\frac12 w_3,\\
\{w_4{}_\lambda w_4\}_{z,\rho}&=\frac16 \lambda\,.
\end{align*}
%
As stated in Corollary \ref{daniele} we note that the Poisson
structure corresponding to $z=0$ does not change for different choices
of the isotropic subspace $\mf l\subset\mf g_{\frac12}$
(but it does change for arbitrary $z$ with the change of $s$).
\end{example}


\section{Generalized Drinfeld-Sokolov hierarchies in the non-homogeneous case}\label{sec:non_hom}
In this section we construct,
using the Lenard-Magri scheme, an integrable hierarchy of Hamiltonian equations
for the classical $\mc W$-algebra defined in Definition \ref{walg3}.
We use the same setup and notation as in the previous section.
\subsection{Reformulation of the Lenard-Magri scheme}\label{sec:4.1}

\begin{proposition}\label{int_hier1_ds}
Letting $L(z)=\partial+q+(f+zs)\otimes1\in\mb F\partial\ltimes(\mf g\otimes\mc V(\mf p))$,
we have,
for $a\in\mf p$ and $g\in\mc V(\mf p)$,
$$
(H-zK)(a\otimes g)=(\pi_{\mf m^\perp}\otimes1)[L(z),a\otimes g]\,.
$$
\end{proposition}
\begin{proof}
It follows immediately from \eqref{HKmaps}.
\end{proof}
The variational derivative \eqref{eq:varder} in the algebra of differential polynomials $\mc V(\mf p)$,
denoted by $\frac{\delta}{\delta q}:\,\mc V(\mf p)\to\mf p\otimes\mc V(\mf p)$,
is given by ($g\in\mc V(\mf p)$):
\begin{equation}\label{variational_ds}
\frac{\delta g}{\delta q}=
\sum_{i\in P}q_i\otimes\frac{\delta g}{\delta q_i}
\,\in\mf p\otimes\mc V(\mf p)\,.
\end{equation}
\begin{lemma}\label{gauge1}
If $g\in\mc W\subset\mc V(\mf p)$, 
then $\left[L(z),\frac{\delta g}{\delta q}\right]\in\mf m^{\perp}\otimes\mc V(\mf g)$.
\end{lemma}
\begin{proof}
Let $\{q_i\}_{i\in M}$ be a basis of $\mf m$, so that, if $I=P\cup M$,
$\{q_i\}_{i\in I}$ is a basis of $\mf g$.
In order to prove the lemma, we only need to show that
$$
\kappa\Big(\Big[L(z),\frac{\delta g}{\delta q}\Big]\,\Big|\,q_k\otimes1\Big)=0
\qquad
\text{ for all } k\in M\,.
$$
By the definition \eqref{20120511:eq2} of the space $\mc W\subset\mc V(\mf p)$,
we have $\rho\{{q_k}_\lambda g\}_z=0$ for all $k\in M$ (since $\mf m\subset\mf n$).
By the skewsymmetry of the $\lambda$-bracket, and using the fact that $\rho$ is a homomorphism
of differential algebras, we thus get, using the Master Formula \eqref{masterformula} 
and the definition \eqref{lambda} of the $\lambda$-bracket $\{\cdot\,_\lambda\,\cdot\}_z$ on $\mc V(\mf g)$, 
that 
$$
0=\rho\left.\{g_\lambda q_k\}_z\right|_{\lambda=0}
=\sum_{i\in P}\Big(
\pi_{\mf p}[q_i,q_k]+\kappa(f+zs\mid[q_i,q_k])+\kappa(q_i\mid q_k)\partial\Big)\frac{\delta g}{\delta q_i}
\,,
$$
for every $k\in M$.
On the other hand, it is not hard to check that the RHS above is equal to
$\kappa\Big(\Big[L(z),\frac{\delta g}{\delta q}\Big]\,\Big|\,q_k\otimes1\Big)$,
proving the claim.
\end{proof}
\begin{corollary}\label{cor4apr}
If $g\in\mc W\subset\mc V(\mf p)$, 
then 
$$
(H-zK)\Big(\frac{\delta g}{\delta q}\Big)=\left[L(z),\frac{\delta g}{\delta q}\right]
\,.
$$
\end{corollary}
\begin{proof}
It is an immediate consequence of Proposition \ref{int_hier1_ds} and Lemma \ref{gauge1}.
\end{proof}

According to the so called Lenard-Magri scheme of integrability \cite{Mag78} (see also \cite{BDSK09}),
in order to construct an integrable hierarchy of bi-Hamiltonian equations in $\mc W$,
we need to find a sequence of local functionals
$\tint g_n\in\quot{\mc W}{\partial\mc W},\,n\in\mb Z_+$, such that
\begin{equation}\label{eq:lenard2}
\left.\tint \{{g_0}_\lambda p\}_{K,\rho}\right|_{\lambda=0}=0
\,\,\,\text{ and }\,\,\,
\left.\tint \{{g_n}_\lambda p\}_{H,\rho}\right|_{\lambda=0}
=\left.\tint \{{g_{n+1}}_\lambda p\}_{K,\rho}\right|_{\lambda=0}
\,,\,\, p\in\mc W\,.
\end{equation}
In this case it is not hard to prove (see e.g. \cite{BDSK09}) 
that we get the corresponding integrable hierarchy of Hamiltonian equations
(see \eqref{eq:hierarchy}):
$$
\frac{dp}{dt_n}
=\left. \{{g_n}_\lambda p\}_{H,\rho}\right|_{\lambda=0},
\qquad
n\in\mb Z_+\,,
$$
provided that the $\tint g_n$'s span an infinite dimensional subspace of 
$\quot{\mc W}{\partial\mc W}$.

We can reformulate the Lenard-Magri recursion relation \eqref{eq:lenard2}
in terms of the matrices $H$ and $K$ defined in \eqref{HKds}.
By \eqref{wbrackX}, equation \eqref{eq:lenard2} reads
($n\in\mb Z_+$, $p\in\mc W$)
\begin{equation}\label{eq:lenard3}
\int \sum_{i,j\in P}
\frac{\delta p}{\delta q_j} K_{ji}(\partial) \frac{\delta g_0}{\delta q_i}
=0
\,,\,\,
\int \sum_{i,j\in P}
\frac{\delta p}{\delta q_j} H_{ji}(\partial) \frac{\delta g_n}{\delta q_i}
=
\int \sum_{i,j\in P}
\frac{\delta p}{\delta q_j} K_{ji}(\partial) \frac{\delta g_{n+1}}{\delta q_i}
\,.
\end{equation}
Equivalently, we can rewrite equation \eqref{eq:lenard3}
in terms of the maps 
$H,K:\,\mf p\otimes\mc V(\mf p)\to\mf m^{\perp}\otimes\mc V(\mf p)$ 
defined in \eqref{HKmaps}.
For this, we consider the non-degenerate pairing
$(\mf m^{\perp}\otimes\mc V(\mf p))\times(\mf p\otimes\mc V(\mf p))
\to\quot{\mc V(\mf p)}{\partial\mc V(\mf p)}$,
defined by
$$
a\otimes g,b\otimes h\mapsto \tint\kappa(a\mid b)gh\,.
$$
In terms of the dual bases $\{q_i\}_{i\in P},\,\{q^i\}_{i\in P}$ of $\mf p$
and $\mf m^{\perp}$ respectively, 
the pairing is
$\sum_{i\in P}q^i\otimes g_i\,,\,\sum_{j\in P}q_j\otimes h_j
\mapsto\int\sum_{i\in P}g_ih_i$.
Then, the Lenard-Magri recursion relation \eqref{eq:lenard3}
can be rewritten as
$$
\int\kappa\Big(K\Big(\frac{\delta g_0}{\delta q}\Big)\,\Big|\,\frac{\delta p}{\delta q}\Big)=0
\,,\,\,
\int\kappa\Big( 
H\Big(\frac{\delta g_n}{\delta q}\Big)
-K\Big(\frac{\delta g_{n+1}}{\delta q}\Big)
\,\Big|\,\frac{\delta p}{\delta q}\Big)
=0
\,,
$$
for $p\in\mc W,\,n\in\mb Z_+$.
In terms of the generating series 
$\tint g(z)=\sum_{n\in\mb Z_+}\tint g_nz^{-n+N}\in\quot{\mc W}{\partial\mc W}((z^{-1}))$
($N\in\mb Z$ is arbitrary)
these relations can be equivalently rewritten, using Corollary \ref{cor4apr}, as
\begin{equation}\label{eq:lenard5}
\int\kappa\Big(
\Big[L(z),\frac{\delta g(z)}{\delta q}\Big]
\,\Big|\,\frac{\delta p}{\delta q}\Big)
=0
\,\,\text{ in }\,\,
\quot{\mc V(\mf p)}{\partial\mc V(\mf p)}((z^{-1}))
\,,\,\,
p\in\mc W\,.
\end{equation}
Here and below, we extend $\kappa$ to a bilinear map 
$(\mf g((z^{-1}))\otimes\mc V(\mf p))\times(\mf g((z^{-1}))\otimes\mc V(\mf p))\to\mc V(\mf p)((z^{-1}))$
as in \eqref{eq:forma} and linearly in $z$.

\subsection{Basic assumptions}\label{sec:4.2}

In the remainder of the paper we will assume that $s\in\Ker(\ad \mf n)\subset\mf g$ 
is a homogeneous element with respect to the $\ad x$-eigenspace decomposition \eqref{dec},
and that the Lie algebra $\mf g((z^{-1}))$ admits a decomposition 
\begin{equation}\label{eq:dech}
\mf g((z^{-1}))=\Ker\ad(f+zs)\oplus\im\ad(f+zs)
\end{equation}
(as pointed out in Remark \ref{rem:semisimple},
this is equivalent to semisemplicity of the element
$f+zs$ of the reductive Lie algebra $\mf g((z^{-1}))$
over the field $\mb F((z^{-1}))$).
Under these two assumptions, we will be able to construct, in the following sections,
the desired series $\tint g(z)\in\quot{\mc W}{\partial\mc W}((z^{-1}))$ solving \eqref{eq:lenard5},
thus providing an integrable hierarchy of bi-Hamiltonian equations in $\mc W$.

We denote $\mf h=\Ker\ad(f+zs)\subset\mf g((z^{-1}))$.
Then $\im\ad(f+zs)=\mf h^\perp$ is the orthogonal complement
to $\mf h$ with respect to the non-degenerate symmetric invariant
bilinear form $\kappa_0:\,\mf g((z^{-1}))\times\mf g((z^{-1}))\to\mb F$
(the constant term of $\kappa$ on $\mf g((z^{-1}))$)
given by 
$$
\kappa_0(a(z)\mid b(z))=\sum_{i\in\mb Z}\kappa(a_i\mid b_{-i})\,,
$$
for $a(z)=\sum_{i\in\mb Z}a_iz^{-i}$ and $b(z)=\sum_{i\in\mb Z}b_iz^{-i}\in\mf g((z^{-1}))$.

\subsection{Outline}

The applicability of the Lenard-Magri scheme of integrability will be achieved, 
following the ideas of Drinfeld and Sokolov \cite{DS85}, in four steps:
\begin{enumerate}[1.]
\item
In Section \ref{sec:4.4} we find $h(z)\in\mf h\otimes\mc V(\mf p)$
such that $e^{\ad U(z)}(L(z))=\partial+(f+zs)\otimes1+h(z)$ for some
$U(z)\in\mf g((z^{-1}))\otimes\mc V(\mf p)$.
\item
In Section \ref{sec:4.5} we prove that,
if $a(z)\in Z(\mf h)$, then 
$\tint g(z)=\tint \kappa(a(z)\otimes1\mid h(z))\in(\quot{\mc V(\mf p)}{\partial\mc V(\mf p)})((z^{-1}))$
solves the Lenard-Magri recursion condition \eqref{eq:lenard5}.
\item
In Section \ref{sec:4.6} we prove that $\tint g(z)$ defined above lies in
$(\quot{\mc W}{\partial\mc W})((z^{-1}))$
(namely, the coefficients of $g(z)$ lie in $\mc W$ up to total derivatives, 
cf. Lemma \ref{20130207:lem}).
\item
Finally, in Section \ref{sec:4.7} we prove that 
the coefficients $\tint g_n$ of the Laurent series $\tint g(z)$
span an infinite-dimensional subspace of $\quot{\mc W}{\partial\mc W}$.
\end{enumerate}

\subsection{Step 1}\label{sec:4.4}

We extend the gradation \eqref{dec} of $\mf g$
to a gradation of $\mf g((z^{-1}))$ by letting $f+zs$ be homogeneous of degree $-1$.
In other words, let $s$ have $\ad x$-eigenvalue $m\geq0$ 
($s$ is an eigenvector by assumption,
and it lies in the centralizer of $e$, 
hence it has non-negative eigenvalue),
then we let $z$ have degree $-m-1$.
\begin{lemma}
\begin{enumerate}[(a)]
\item
For $i\in\frac12\mb Z$,
let $\mf g((z^{-1}))_i\subset\mf g((z^{-1}))$ be the space of homogeneous elements of degree $i$.
We have the decomposition
\begin{equation}\label{decz}
\mf g((z^{-1}))=\widehat\bigoplus_{i\in\frac12\mb Z}\mf g((z^{-1}))_i\,,
\end{equation}
where the direct sum is completed by allowing infinite series in positive degrees.
\item
If $U(z)\in \mf g((z^{-1}))_{>0}\otimes\mc V(\mf p)$, 
then we have a well defined Lie algebra automorphism
$e^{\ad U(z)}$ of the Lie algebra $\mb F\partial\ltimes\big(\mf g((z^{-1}))\otimes\mc V(\mf p)\big)$.
\end{enumerate}
\end{lemma}
\begin{proof}
Let $\Delta$ be the maximal eigenvalue of $\ad x$ in $\mf g$.
Since $z$ has degree $-m-1<0$, we have
\begin{equation}\label{decz2}
\mf g((z^{-1}))_i\subset
\!\!\!\!\!\!
\bigoplus_{-\frac{i}{m+1}-\frac{\Delta}{m+1}\leq n\leq-\frac{i}{m+1}+\frac{\Delta}{m+1}}
\!\!\!\!\!\!
\mf gz^n
\,\,\,\text{ and }\,\,\,
\mf gz^n\subset
\!\!\!\!\!\!
\bigoplus_{-n(m+1)-\Delta\leq i\leq-n(m+1)+\Delta}
\!\!\!\!\!\!
\mf g((z^{-1}))_i
\,.
\end{equation}
The decomposition \eqref{decz} follows immediately by these inclusions.
Part (b) follows from (a).
\end{proof}

Note that, since $f+zs$ is homogeneous in $\mf g((z^{-1}))$,
then $\mf h=\Ker\ad(f+zs)\subset\mf g((z^{-1}))$
and $\mf h^\perp=\im\ad(f+zs)\subset\mf g((z^{-1}))$ 
are compatible with the decomposition \eqref{decz}:
$\mf h=\widehat\bigoplus_{i\in\frac12\mb Z}\mf h_i$
and $\mf h^\perp=\widehat\bigoplus_{i\in\frac12\mb Z}\mf h^\perp_i$,
where $\mf h_i=\mf h\cap\mf g((z^{-1}))_i$
and $\mf h^\perp_i=\mf h^\perp\cap\mf g((z^{-1}))_i$.
In the following, for $k\in\frac12\mb Z$ and a subspace $V=\widehat\bigoplus_{i\in\frac12\mb Z}V_i$
compatible with the decomposition \eqref{decz} (such as $\mf h$ or $\mf h^\perp$),
we denote $V_{>k}=\widehat\bigoplus_{i>k}V_i$.

\begin{proposition}\label{int_hier2_ds}
Let $r=\sum_{i\in P}q^i\otimes r_i\in\mf m^\perp\otimes\mc V(\mf p)$.
\begin{enumerate}[(a)]
\item
There exist unique formal Laurent series $U(z)\in\mf{h}^\perp_{>0}\otimes\mc V(\mf p)$
and $h(z)\in\mf{h}_{>-1}\otimes\mc V(\mf p)$ such that
\begin{equation}\label{L0_dsr}
e^{\ad U(z)}(\partial+(f+zs)\otimes1+r)=\partial+(f+zs)\otimes1+h(z)\,.
\end{equation}
Moreover, the coefficients of $U(z)$ and $h(z)$ are differential polynomials in
$r_1,\dots,r_k$.
\item
An automorphism $e^{\ad U(z)}$, with $U(z)\in\mf{g}((z^{-1}))_{>0}\otimes\mc V(\mf p)$,
solving \eqref{L0_dsr} for some $h(z)\in\mf{h}_{>-1}\otimes\mc V(\mf p)$,
is defined uniquely up to multiplication on the left by automorphisms of the form $e^{\ad S(z)}$, 
where $S(z)\in \mf h_{>0}\otimes\mc V(\mf p)$.
\end{enumerate}
\end{proposition}
\begin{proof}
Let us write $U(z)=\sum_{i\geq\frac12}U_i(z)$,
where $U_i(z)\in\mf h^{\perp}_i\otimes\mc V(\mf p),\,i\geq\frac12$,
and $h(z)=\sum_{i\geq-\frac12}h_i(z)$,
where $h_i(z)\in\mf h_i\otimes\mc V(\mf p),\,i\geq-\frac12$.
We will determine $U_{i+1}(z)\in\mf h^{\perp}_{i+1}\otimes\mc V(\mf p)$
and $h_i(z)\in\mf h_i\otimes\mc V(\mf p)$, inductively on $i\geq-\frac12$,
by equating the homogeneous components of degree $i$ 
in each sides of equation \eqref{L0_dsr}.

Recall that $\mf m^\perp\subset\mf g_{\geq-\frac12}$.
Equating the terms of degree $-\frac12$ in both sides of \eqref{L0_dsr},
we get the equation
$$
h_{-\frac12}(z)+[(f+zs)\otimes1,U_{\frac12}(z)]
=(\pi_{-\frac12}\otimes1)r\in\mf g_{-\frac12}\otimes\mc V(\mf p)\,,
$$
where $\pi_{-\frac12}:\,\mf g((z^{-1}))\twoheadrightarrow\mf g((z^{-1}))_{-\frac12}$
denotes the projection on the component of degree $-\frac12$.
Since we have the decomposition 
$\mf g_{-\frac12}\subset\mf g((z^{-1}))_{-\frac12}=\mf h_{-\frac12}\oplus\mf h^\perp_{-\frac12}$,
and since $\ad(f+zs)$ restricts to a bijection 
$\mf h^\perp_{\frac12}\stackrel{\sim}{\longrightarrow}\mf h^\perp_{-\frac12}$,
the above equation determines uniquely $h_{-\frac12}(z)\in\mf h_{-\frac12}\otimes\mc V(\mf p)$ 
and $U_{\frac12}(z)\in\mf h^\perp_{\frac12}\otimes\mc V(\mf p)$.
Moreover, the coefficients of $h_{-\frac12}(z)$ and $U_{\frac12}(z)$ 
are obviously differential polynomials in $r_1,\dots,r_k$.

Next, suppose by induction on $i$ that we determined all elements
$U_{j+1}(z)\in\mf h^\perp_{j+1}\otimes\mc V(\mf p)$
and $h_j(z)\in\mf h_{j}\otimes\mc V(\mf p)$ for $j<i$, $i>-\frac12$,
and that their coefficients are differential polynomials in $r_1,\dots,r_k$.
Equating the terms of degree $i$ in both sides of \eqref{L0_dsr},
we get an equation in $h_i(z)$ and $U_{i+1}(z)$ of the form
$$
h_i(z)+[(f+zs)\otimes1,U_{i+1}(z)]=A(z)\,,
$$
where $A(z)\in\mf g((z^{-1}))_i\otimes\mc V(\mf p)$ is certain complicated (differential polynomial) 
expression involving all the elements $U_{j+1}(z)$ and $h_j(z)$ for $j<i$.
As before, since $\mf g((z^{-1}))_i=\mf h_i\oplus\mf h^\perp_i$,
and since $\ad(f+zs)$ restricts to a bijection 
$\mf h^\perp_{i+1}\stackrel{\sim}{\longrightarrow}\mf h^\perp_{i}$,
the above equation determines uniquely $h_i(z)\in\mf h_i\otimes\mc V(\mf p)$ 
and $U_{i+1}(z)\in\mf h^\perp_{i+1}\otimes\mc V(\mf p)$,
and their coefficients are differential polynomials in $r_1,\dots,r_k$.
This proves part (a).

In the proof of part (b) we follow the same argument as in the proof of Proposition \ref{int_hier2}(b).
Let $U(z)\in \mf{h}^\perp_{>0}\otimes\mc V(\mf p),\,h(z)\in\mf h_{>-1}\otimes\mc V(\mf p)$ 
be the unique solution of \eqref{L0_dsr} given by part (a).
Let also 
$\widetilde U(z)\in \mf g((z^{-1}))_{>0}\otimes\mc V(\mf p),\,
\widetilde h(z)\in\mf h_{>-1}\otimes\mc V(\mf p)$ 
be some other solution of \eqref{L0}:
$e^{\ad \widetilde{U}(z)}(\partial+(f+zs)\otimes1+r)=\partial+(f+zs)\otimes1+\widetilde{h}(z)$.
By the Baker-Campbell-Hausdorff formula \cite{Ser92},
there exists $S(z)=\sum_{i>0}^\infty S_i(z)\in\mf g((z^{-1}))_{>0}\otimes\mc V(\mf p)$ 
such that
$e^{\ad\widetilde U(z)}e^{-\ad U(z)}=e^{\ad S(z)}$.
To conclude the proof of (b), we need to show that $S(z)\in\mf h_{>0}\otimes\mc V(\mf p)$.
By construction, we have
\begin{equation}\label{eq:ciaociao}
\partial+(f+zs)\otimes1+\widetilde{h}(z)=e^{\ad S(z)}(\partial+(f+zs)\otimes1+h(z))\,.
\end{equation}
Comparing the terms of degree $-\frac12$ in both sides of the above equation,
we get
$$
\mf h^\perp_{-\frac12}\otimes\mc V(\mf p)\ni[(f+zs)\otimes1,S_{\frac12}(z)]
=
h_{-\frac12}(z)-\widetilde{h}_{-\frac12}(z)\in\mf h_{-\frac12}\otimes\mc V(\mf p)\,.
$$
Since $\mf h^\perp_{-\frac12}\cap\mf h_{-\frac12}=0$,
we conclude that 
$\widetilde{h}_{-\frac12}(z)=h_{-\frac12}(z)$ and
$S_{\frac12}(z)\in\mf h_{\frac12}\otimes\mc V(\mf p)$.
Next,
assuming by induction that $S_j(z)\in\mf h_j\otimes\mc V(\mf p)$ for all $j<i$,
and comparing the terms of degree $i$ in both sides of equation \eqref{eq:ciaociao},
we easily get that 
$[(f+zs)\otimes1,S_{i+1}(z)]\in(\mf h^\perp_i\otimes\mc V(\mf p))\cap(\mf h_i\otimes\mc V(\mf p))=0$,
namely $S_{i+1}(z)\in\mf h_{i+1}\otimes\mc V(\mf p)$, as desired.
\end{proof}
Consider the special case when $r=q\in\mf m^\perp\otimes\mc V(\mf p)$.
In this case, equation \eqref{L0_dsr} reads
\begin{equation}\label{L0_ds}
L_0(z):=e^{\ad U(z)}(L(z))=\partial+(f+zs)\otimes1+h(z)\,.
\end{equation}
Proposition \ref{int_hier2_ds} states that there exist unique
$U_0(z)\in\mf{h}^\perp_{>0}\otimes\mc V(\mf p)$
and $h_0(z)\in\mf{h}_{>-1}\otimes\mc V(\mf p)$ solving \eqref{L0_ds},
and any other solution of \eqref{L0_ds} with $U(z)\in\mf g((z^{-1}))_{>0}\otimes\mc V(\mf p)$ 
and $h(z)\in\mf h_{>-1}\otimes\mc V(\mf p)$
is obtained from the unique one $(U_0(z),h_0(z))$
by taking $e^{\ad U(z)}=e^{\ad S(z)}e^{\ad U_0(z)}$
for some $S(z)\in\mf h_{>0}\otimes\mc V(\mf p)$.

\subsection{Step 2}\label{sec:4.5}

Throughout this section, we let $U(z)\in\mf g((z^{-1}))_{>0}\otimes\mc V(\mf p)$
and $h(z)\in\mf h_{>-1}\otimes\mc V(\mf p)$ be a solution of equation \eqref{L0_ds},
and we fix an element 
$a(z)\in Z(\mf h)$.
(For example, $a(z)=f+zs\in Z(\mf h)$).
We also denote
\begin{equation}\label{gz}
\tint g(z)=\tint \kappa(a(z)\otimes1\mid h(z))\,\in\quot{\mc V(\mf p)}{\partial\mc V(\mf p)}((z^{-1}))\,.
\end{equation}
The main result of the section will be Theorem \ref{main-step2} below, 
where we show that $\tint g(z)$ solves the Lenard-Magri recursion equation \eqref{eq:lenard5}.
In the following Sections \ref{sec:4.6} and \ref{sec:4.7} we will then show that, in fact, $\tint g(z)$
is independent of the choice of the solution $U(z),h(z)$ of \eqref{L0_ds},
that it lies in $(\quot{\mc W}{\partial\mc W})((z^{-1}))$,
and that its coefficients span an infinite-dimensional subspace of $\quot{\mc W}{\partial\mc W}$,
thus completing the proof of the applicability of the Lenard-Magri scheme of integrability.

Before proving the main theorem, we need some preliminary results.
First, as immediate consequence of Proposition \ref{int_hier2_ds}, we have the following result.
\begin{corollary}\label{int_hier3_ds}
We have $[L(z),e^{-\ad U(z)}(a(z)\otimes1)]=0$.
\end{corollary}
\begin{proof}
Since, by assumption, $a(z)\in Z(\mf h)\subset\mf h$, we have
$[\partial+(f+zs)\otimes1+h(z),a(z)\otimes1]=0$. Using the fact that
$e^{-\ad U(z)}$ is an automorphism of the Lie algebra $\mb F\partial\ltimes(\mf g((z^{-1}))\otimes\mc V(\mf p))$,
we have
$[L(z),e^{-\ad U(z)}(a(z)\otimes1)]
=e^{-\ad U(z)}[\partial+(f+zs)\otimes1+h(z),a(z)\otimes1]=0$.
\end{proof}

\begin{lemma}\label{step2-lemma1}
For $a\otimes g\in\mf g\otimes\mc V(\mf p)$ and $p\in\mc W$, we have
\begin{equation}\label{eq:gauge}
\tint\kappa\Big([L(z),a\otimes g]\,\Big|\,\frac{\delta p}{\delta q}\Big)
=\tint\rho{\{a_{\partial}p\}_z}_\rightarrow g\,.
\end{equation}
\end{lemma}
\begin{proof}
By the definition \eqref{variational_ds} of the variational derivative we have
$$
\int\kappa\left([L(z),a\otimes g]\left|\,\frac{\delta p}{\delta q}\right)\right.
=\int\sum_{i\in P}\frac{\delta p}{\delta q_i}\kappa([L(z),a\otimes g]\mid q_i\otimes1)\,.
$$
On the other hand, using Master Formula \eqref{masterformula} and integration by parts, we get
$$
\tint\rho{\{a_{\partial}p\}_z}_\rightarrow g
=\int\sum_{i\in P}\frac{\delta p}{\delta q_i}\rho{\{a_\partial q_i\}_z}_\rightarrow g\,.
$$
Hence, equation \eqref{eq:gauge} follows immediately from equation \eqref{20120516:eq2} with $n=1$.
\end{proof}

\begin{lemma}\label{step2-lemma2}
We have
\begin{equation} \label{var_der_ds}
\frac{\delta g(z)}{\delta q}
=(\pi_{\mf p}\otimes1)\Big(e^{-\ad U(z)}(a(z)\otimes1)\Big)
\,\in(\mf p\otimes\mc V(\mf p))((z^{-1}))\,,
\end{equation}
where the projection $\pi_{\mf p}:\,\mf g\to\mf p$ is extended to $\mf g((z^{-1}))$
in the obvious way.
\end{lemma}
\begin{proof}
The proof follows from a straightforward computation following the same steps as in  
the proof of  Proposition \ref{int_hier3}.
By the definition \eqref{variational_ds} of the variational derivative
and the definition \eqref{gz} of $\tint g(z)$, we have
\begin{equation}\label{eq:30apr-1}
\begin{array}{l}
\displaystyle{
\frac{\delta g(z)}{\delta q}
=\sum_{i\in P,m\in\mb Z_+}
q_i\otimes(-\partial)^m
\frac{\partial g(z)}{\partial q_i^{(m)}}
=\sum_{i\in P,m\in\mb Z_+}
q_i\otimes(-\partial)^m
\kappa\Big(a(z)\otimes1\,\Big|\,\frac{\partial h(z)}{\partial q_i^{(m)}}\Big)\,.
} \\
\displaystyle{
=\sum_{i\in P,m\in\mb Z_+}
q_i\otimes(-\partial)^m
\kappa\Big(
a(z)\otimes1\,\Big|\,\frac{\partial}{\partial q_i^{(m)}}
\Big(
e^{\ad U(z)} (L(z))-\partial-(f+zs)\otimes1
\Big)\Big)
} 
\end{array}
\end{equation}
In the last identity we used equation \eqref{L0_ds}.
We next expand $e^{\ad U(z)}$ in power series. The first term of the expansion is,
by the definition \eqref{q} of $q\in\mf m^\perp\otimes\mc V(\mf p)$ 
and the first completeness relation \eqref{complete},
\begin{equation}\label{eq:30apr-1a}
\sum_{i\in P,m\in\mb Z_+}
q_i\otimes(-\partial)^m
\kappa\Big(
a(z)\otimes1\,\Big|\,\frac{\partial q}{\partial q_i^{(m)}}\Big)
=
\sum_{i\in P}
\kappa(a(z)\mid q^i)
q_i\otimes1
=
\pi_{\mf p} a(z)\otimes 1\,.
\end{equation}
By Lemma \ref{lem1}, all the other terms in the power series expansion of the RHS of \eqref{eq:30apr-1} are
\begin{equation}\label{eq:30apr-1b}
\begin{array}{l}
\displaystyle{
\sum_{k=1}^\infty\frac1{k!}
\sum_{i\in P,m\in\mb Z_+}
q_i\otimes(-\partial)^m
\kappa\Big(
a(z)\otimes1\,\Big|\,
\frac{\partial}{\partial q_i^{(m)}}(\ad U(z))^kL(z)
\Big)
=
\sum_{k=1}^\infty\frac1{k!}
} \\
\displaystyle{
\times
\!\sum_{i\in P,m\in\mb Z_+}\!\!
q_i\otimes(-\partial)^m
\kappa\Bigg(
a(z)\otimes1\,\Bigg|\,
\sum_{h=0}^{k-1}
(\ad U(z))^h \big(\ad \frac{\partial U(z)}{\partial q_i^{(m)}}\big) (\ad U(z))^{k-h-1} L(z)
} \\
\displaystyle{
+(\ad U(z))^k\frac{\partial}{\partial q_i^{(m)}}(q+(f+zs)\otimes1)
-(\ad U(z))^{k-1}\frac{\partial U(z)}{\partial q_i^{(m-1)}}
\Bigg)
=
\!\!\!
\sum_{h,k\in\mb Z_+}^\infty
\!\!\!
\frac1{(h\!+\!k\!+\!1)!}
} \\
\displaystyle{
\times\sum_{i\in P,m\in\mb Z_+}
q_i\otimes(-\partial)^m
\kappa\Bigg(
a(z)\otimes1\,\Bigg|\,
(\ad U(z))^h \big(\ad \frac{\partial U(z)}{\partial q_i^{(m)}}\big) (\ad U(z))^k L(z)
\Bigg)
} \\
\displaystyle{
+\sum_{k=1}^\infty\frac1{k!}
(\pi_{\mf p}\otimes1)\left((-\ad U(z))^k (a(z)\otimes1)\right)
} \\
\displaystyle{
-\sum_{k\in\mb Z_+}\frac1{(k+1)!}
\sum_{i\in P,m\in\mb Z_+}
q_i\otimes(-\partial)^m
\kappa\Bigg(
a(z)\otimes1\,\Bigg|\,
(\ad U(z))^{k}\frac{\partial U(z)}{\partial q_i^{(m-1)}}
\Bigg)
\,.
}
\end{array}
\end{equation}
For the first and last terms in the RHS we just changed the summation indices,
while for the second term we used the first completeness relation \eqref{complete}
and the invariance of the bilinear map $\kappa$.
Combining \eqref{eq:30apr-1a} and the second term in the RHS of \eqref{eq:30apr-1b},
we get
$(\pi_{\mf p}\otimes1)\left(e^{-\ad U(z)}(a(z)\otimes1)\right)$,
which is the same as the RHS of \eqref{var_der_ds}.
Hence, in order to complete the proof of the proposition,
we are left to show that the first and last term in the RHS of \eqref{eq:30apr-1b} cancel out.
The last term of the RHS of \eqref{eq:30apr-1b} can be rewritten as
\begin{equation}\label{eq:30apr-3}
-\sum_{i\in P,m\in\mb Z_+}
q_i\otimes(-\partial)^m
\kappa(
a(z)\otimes1\mid
A_{i,m-1}(z))\,,
\end{equation}
where $A_{i,m}(z)=\sum_{k\in\mb Z_+}\frac1{(k+1)!}
(\ad U(z))^{k}\frac{\partial U(z)}{\partial q_i^{(m)}}$.
On the other hand, by Lemma \ref{lem2}, the first term of the RHS of \eqref{eq:30apr-1b} is equal to
$$
\sum_{i\in P,m\in\mb Z_+}
q_i\otimes(-\partial)^m
\kappa\big(
a(z)\otimes1\,\mid\,
\big[A_{i,m}(z),e^{\ad U(z)}L(z)\big]
\big)
\,.
$$
By equation \eqref{L0_ds}, the invariance of the bilinear map $\kappa$
and the assumption that $a(z)$ lies in the center of $\mf h$,
the above expression is equal to
$$
\sum_{i\in P,m\in\mb Z_+}
q_i\otimes(-\partial)^{m+1}
\kappa(
a(z)\otimes1\mid A_{i,m}(z))
\,,
$$
which, combined with \eqref{eq:30apr-3}, gives zero.
\end{proof}

\begin{theorem}\label{main-step2}
The formal Laurent series $\tint g(z)\in(\quot{\mc V(\mf p)}{\partial\mc V(\mf p)})((z^{-1}))$ in \eqref{gz}
solves the Lenard-Magri recursion equation \eqref{eq:lenard5}.
\end{theorem}
\begin{proof}
By \eqref{var_der_ds} we have
$$
\begin{array}{l}
\displaystyle{
\frac{\delta g(z)}{\delta q}
=(\pi_{\mf p}\otimes1)\Big(e^{-\ad U(z)}(a(z)\otimes1)\Big)
} \\
\displaystyle{
=e^{-\ad U(z)}(a(z)\otimes1)-(\pi_{\mf m}\otimes1)\Big(e^{-\ad U(z)}(a(z)\otimes1)\Big)\,,
}
\end{array}
$$
so that, by Corollary \ref{int_hier3_ds}, we get
$$
\left[L(z),\frac{\delta g(z)}{\delta q}\right]
=-\Big[L(z),(\pi_{\mf m}\otimes1)\Big(e^{-\ad U(z)}(a(z)\otimes1)\Big)\Big]\,.
$$
Hence, \eqref{eq:lenard5} holds 
by \eqref{eq:gauge} and the definition \eqref{20120511:eq2} of the space $\mc W\subset\mc V(\mf p)$
(recall that $\mf m\subset\mf n$).
\end{proof}

\subsection{Step 3}\label{sec:4.6}

In this section we will show that the Laurent series $\tint g(z)$ given by \eqref{gz}
has coefficients in $\quot{\mc W}{\partial\mc W}$ (Proposition \ref{20120515:thm} below).
\begin{lemma}\label{20120515:lem}
The Laurent series $\tint g(z)$ defined by \eqref{gz}
is independent of the choice of 
$U(z)\in\mf g((z^{-1}))_{>0}\otimes\mc V(\mf p)$, $h(z)\in\mf h_{>-1}\otimes\mc V(\mf p)$,
solving equation \eqref{L0_ds}.
\end{lemma}
\begin{proof}
Let $\widetilde{U}(z)\in\mf g((z^{-1}))_{>0}\otimes\mc V(\mf p)$, 
$\widetilde{h}(z)\in\mf h_{>-1}\otimes\mc V(\mf p)$ be any other solution of equation \eqref{L0_ds},
and let $\tint\widetilde{g}(z)=\tint\kappa(a(z)\otimes1\mid\widetilde{h}(z))$.
By Proposition \eqref{int_hier2_ds}(b) there exists
$S(z)\in\mf h_{>0}\otimes\mc V(\mf p)$ such that
$e^{\ad\widetilde{U}(z)}=e^{\ad S(z)}e^{\ad U(z)}$.
By Lemma \ref{step2-lemma2}, we then have
$$
\begin{array}{l}
\displaystyle{
\frac{\delta\widetilde{g}(z)}{\delta q}
=(\pi_{\mf p}\otimes1)\Big(e^{-\ad U(z)}e^{-\ad S(z)}(a(z)\otimes1)\Big)
} \\
\displaystyle{
=(\pi_{\mf p}\otimes1)\Big(e^{-\ad U(z)}(a(z)\otimes1)\Big)
=\frac{\delta g(z)}{\delta q}
\,.
}
\end{array}
$$
In the second equality we used the assumption that $a(z)\in Z(\mf h)$.
Since in the algebra of differential polynomials $\mc V(\mf p)$ we have
$\Ker\big(\frac{\delta}{\delta q}\big)=\partial\mc V(\mf p)\oplus\mb F$,
we deduce that $\tint \tilde g(z)$ and $\tint g(z)$ differ at most by a constant.
On the other hand, as explained in the proof of Lemma \ref{20120522:lem1},
the constant term $\tilde h(z)[0]\in\mf h\otimes1$ of $\tilde h(z)$
is always zero.
Therefore, the constant term of $\tint \tilde g(z)$ is zero as well.
\end{proof}
\begin{proposition}\label{20120515:thm}
We have $\tint g(z)\in(\quot{\mc W}{\partial\mc W})((z^{-1}))$.
\end{proposition}
\begin{proof}
Fix a subspace $V\subset\mf m^\perp$ complementary to $[f,\mf n]\subset\mf m^\perp$
compatible with the direct sum decomposition \eqref{dec},
and let $X\in\mf n\otimes\mc V(\mf p)$ and $w\in V\otimes\mc V(\mf p)$
be the unique element provided by Theorem \ref{20120511:thm2}(a).
By Proposition \ref{int_hier2_ds}(a) and Theorem \ref{20120511:thm2}(b),
there exist unique $U_w(z)\in\mf h^\perp_{>0}\otimes\mc W$ and $h_w(z)\in\mf h_{>0}\otimes\mc W$,
such that
$$
e^{\ad U_w(z)}(\partial+(f+zs)\otimes1+w)=\partial+(f+zs)\otimes1+h_w(z)\,.
$$
By the identity $w=q^X$, we can rewrite the above equation as
$$
e^{\ad U_w(z)}e^{-\ad X}(\partial+(f+zs)\otimes1+q)=\partial+(f+zs)\otimes1+h_w(z)\,.
$$
Since $\mf n\subset\mf g((z^{-1}))_{>0}$,
by the Baker-Campbell-Hausdorff formula there exists 
$\widetilde U(z)\in\mf g((z^{-1}))_{>0}\otimes\mc V(\mf p)$
such that $e^{\ad U_w(z)}e^{-\ad X}=e^{\ad\widetilde{U}(z)}$.
Hence, $\widetilde{U}(z),h_w(z)$ is another solution of equation \eqref{L0_ds}.
By Lemma \ref{20120515:lem} we thus conclude that
$\tint g(z)=\tint\kappa(a(z)\otimes1\mid h_w(z))\in(\quot{\mc W}{\partial\mc W})((z^{-1}))$.
\end{proof}

\subsection{Step 4}\label{sec:4.7}
Consider the Laurent series
$\tint g(z)=\sum_{n\in\mb Z_+}\tint g_nz^{-n+N}$ defined in \eqref{gz}.
In this section we prove that, if $a(z)\in Z(\mf h)$
does not lie in the center of $\mf g((z^{-1}))$,
then the local functionals $\{\tint g_n\}_{n\in\mb Z_+}$
span an infinite-dimensional subspace of $\quot{\mc W}{\partial\mc W}$.

We start by computing explicitly 
the linear (as polynomial in $\mc V(\mf p)$) part of $\frac{\delta g(z)}{\delta q}$.
Let $U(z)\in\mf h^\perp\otimes\mc V(\mf p)$ and $h(z)\in\mf h\otimes\mc V(\mf p)$
be the unique solution of equation \eqref{L0_ds}.
Let $U(z)=\sum_{k\in\mb Z_+}U(z)[k]$ be the decomposition of $U(z)$
according to the usual grading 
of the algebra of differential polynomials $\mc V(\mf p)=S(\mb F[\partial]\mf p)$.
\begin{lemma}\label{20120522:lem1}
The linear component of $U(z)$ is:
\begin{equation}\label{20120523:eq1}
U(z)[1]=
\sum_{n\in\mb Z_+}(-1)^n (\ad(f+zs)^{-n-1}\otimes1) (\pi_{\mf h^\perp}\otimes1) \partial^n q\,,
\end{equation}
where $\ad(f+zs)^{-1}$ denotes the inverse map 
of the bijection $\ad(f+zs)\big|_{\mf h^\perp}:\,\mf h^\perp\to\mf h^\perp$,
and $\pi_{\mf h^\perp}:\,\mf g((z^{-1}))\to\mf h^\perp$ denotes the projection onto $\mf h^\perp$
with kernel $\mf h$.
\end{lemma}
\begin{proof}
We proceed as in Remark \ref{rem:grado2}.
First, we equate the homogeneous components of degree $0$ (as polynomials in $\mc V(\mf p)$)
in both sides of \eqref{L0_ds}.
We get 
$$
(e^{\ad U(z)[0]}-1)(f+zs)\otimes1=h(z)[0]\,.
$$
Since $\mf h\cap\mf h^\perp=0$, 
it is not hard to prove, inductively on the grading \eqref{decz},
that $U(z)[0]=0$ in $\mf h^\perp\otimes1$
and $h(z)[0]=0$ in $\mf h\otimes1$.
In fact,
the same argument can be used to prove that,
for \emph{any} solution 
$U(z)\in\mf g((z^{-1}))\otimes\mc V(\mf p)$, $h(z)\in\mf h\otimes\mc V(\mf p)$
of equation \eqref{L0_ds}, we have
$h(z)[0]=0$ and $U(z)[0]\in\mf h\otimes 1$
(a fact that was used in the proof of Lemma \ref{20120515:lem}).

Next, equating the homogeneous components of degree $1$, we get that
$h(z)[1]=(\pi_{\mf h}\otimes1)q$, while $U(z)[1]$ satisfies the equation
$$
[(f+zs)\otimes1,U(z)[1]]=(\pi_{\mf h^\perp}\otimes1)q-U^{\prime}(z)[1]\,.
$$
Let $U(z)[1]=\sum_{0\neq i\in\frac12\mb Z_+}U(z)[1]_i$ be the decomposition of $U(z)[1]$
according to the grading \eqref{decz}. We can solve recursively the above 
equation by equating terms of the same degree with respect to the grading \eqref{decz}.
We find that, for $i\geq-\frac12$ in $\frac12\mb Z$,
$$
U(z)[1]_{i+1}
=\sum_{n=0}^{\left[i+\frac12\right]}(-1)^{n}
(\ad(f+zs)^{-n-1}\otimes1)(\pi_{\mf h^\perp_{i-n}})\partial^{n}q\,
$$
where $[i]$ denotes the integer part of $i\in\frac12\mb Z$,
and $\pi_{\mf h^\perp_{i}}$ is the projection on $\mf h^\perp_i=\mf g((z^{-1}))_i\cap\mf h^\perp$.
Equation \eqref{20120523:eq1} is obtained from the above equation
after summing over all possible degrees $i\geq-\frac12$,
and using the identity
$\sum_i(\pi_{\mf h^\perp_i}\otimes1)q=(\pi_{\mf h^\perp}\otimes1)q$.
\end{proof}
\begin{lemma}\label{20120522:lem2}
The linear part of $\frac{\delta g(z)}{\delta q}\in(\mf p\otimes\mc V(\mf p))((z^{-1}))$, 
with respect to the usual differential polynomial grading in $\mc V(\mf p)$, is
\begin{equation}\label{20120523:eq2}
\frac{\delta g(z)}{\delta q}[1]=
\sum_{n\in\mb Z_+}(-1)^n (\pi_{\mf p}\otimes1)
(\ad(f+zs)^{-n-1}\otimes1) [a(z)\otimes1,\partial^n q]\,,
\end{equation}
where, as before, $\pi_{\mf p}:\,\mf g\to\mf p$ denotes the projection onto $\mf p$
with kernel $\mf m$, and it is extended to $\mf g((z^{-1}))$
in the obvious way.
\end{lemma}
\begin{proof}
By Lemma \ref{step2-lemma2} and equation \eqref{20120523:eq1}, 
the linear part of $\frac{\delta g(z)}{\delta q}$ is 
given by 
$$
\begin{array}{l}
\displaystyle{
\frac{\delta g(z)}{\delta q}[1]
=(\pi_{\mf p}\otimes1)[a(z)\otimes1,U(z)[1]]
} \\
\displaystyle{
=\sum_{n\in\mb Z_+}(-1)^n 
(\pi_{\mf p}\otimes1)
\big[
a(z)\otimes1,
(\ad(f+zs)^{-n-1}\otimes1) (\pi_{\mf h^\perp}\otimes1) \partial^n q
\big]
\,.}
\end{array}
$$
Equation \eqref{20120523:eq2} follows from the facts that,
since $a(z)$ is in the center of $\mf h=\Ker\ad(f+zs)$,
we have that $\ad(f+zs)$ and $\ad a(z)$ commute, and 
$\ad a(z)\circ\pi_{\mf h^\perp}=\ad a(z)$.
\end{proof}
\begin{lemma}\label{20120522:lem3}
If $X\in\mf h^\perp$ is non zero, then
$\pi_{\mf p}\ad(f+zs)^{-n-1}X$ 
is different from zero for infinitely many values of $n\in\mb Z_+$.
\end{lemma}
\begin{proof}
First note that, if $A\in\mf m((z^{-1}))\backslash\{0\}$,
then $\ad(f+zs)A=[f,A]$ (since $s\in\Ker\ad(\mf n)$ and $\mf m\subset\mf n$), 
and therefore, 
since $\ad f:\,\mf g_i\to\mf g_{i-1}$ is injective for $i>0$,
we have that $\ad(f+zs)^kA\not\in\mf m((z^{-1}))$
for some $0\leq k\leq\Delta$, where $\Delta$ is the maximal 
$\ad x$ eigenvalue in $\mf g$.
Then, for every $n\in\mb Z_+$ such that $\pi_{\mf p}\ad(f+zs)^{-n-1}X=0$ 
we have that $\pi_{\mf p}\ad(f+zs)^{k-n-1}X\neq0$ for some $0\leq k\leq\Delta$.
The claim follows.
\end{proof}
\begin{lemma}\phantomsection\label{20120522:lem4}
\begin{enumerate}[(a)]
\item
Consider the decomposition of $a(z)\in Z(\mf h)$ with respect to the grading \eqref{decz}:
$a(z)=\sum_{k\in\frac12\mb Z}a_k(z)$, where $a_k(z)\in Z(\mf h)_k$.
Suppose that the basis $\{q^i\}_{i\in P}$ of $\mf m^\perp$ is homogeneous with respect 
to the decomposition \eqref{dec}, and let $h(i)$ be the degree of $q^i$.
Let also 
$\frac{\delta g(z)}{\delta q}[1]=\sum_{-K\leq k\in\frac12\mb Z}\frac{\delta g(z)}{\delta q}[1]_k$
be the decomposition of $\frac{\delta g(z)}{\delta q}[1]\in\mf p((z^{-1}))\otimes\mc V(\mf p)$
according to \eqref{decz}.
Then we have
$$
\frac{\delta g(z)}{\delta q}[1]_k
=\sum_{i\in P}\sum_{n=0}^{N+k+h(i)+1}A_{k,i,n}\otimes q_i^{(n)}\,,
$$
with ``leading term''
$$
A_{k,i,N+k+h(i)+1}=(-1)^{k+N+h(i)+1}\pi_{\mf p}\ad(f+zs)^{-k-N-h(i)-2}[a_{-N}(z),q^i]\,.
$$
\item
Assume that $a(z)$ does not lie in the center of $\mf g((z^{-1}))$.
Let $-N\in\frac12\mb Z$ be the minimal degree such that $a_{-N}(z)\not\in Z(\mf g((z^{-1})))$,
and let $\bar h=h(\bar i)$, where $\bar i\in P$ is such that $[a_{-N}(z),q^{\bar i}]\neq0$.
Then the leading terms $A_{k,\bar i,N+k+\bar h+1}$ are non zero
for infinitely many values of $k\in\frac12\mb Z$.
\item
In particular, for infinitely many values of the degree $k$,
the elements $\frac{\delta g(z)}{\delta q}[1]_k$ are non zero 
and they have distinct differential orders in the variable $q_{\bar i}\in\mc V(\mf p)$.
\end{enumerate}
\end{lemma}
\begin{proof}
From equation \eqref{20120523:eq2} we have
$$
\frac{\delta g(z)}{\delta q}[1]_k=
\sum_{i\in P}\sum_{n\in\mb Z_+}(-1)^n \pi_{\mf p}
\ad(f+zs)^{-n-1} [a_{k-n-h(i)-1}(z),q^i]\otimes q_i^{(n)}\,.
$$
Part (a) follows from the above equation and the fact that, by assumption, $[a_k(z),q^i]$
can be non zero only for $k\geq -N$.
Part (b) follows from part (a) and Lemma \ref{20120522:lem3}.
Part (c) is obvious.
\end{proof}
\begin{lemma}\label{20120522:lem5}
Let $a(z)\in Z(\mf h)\backslash Z(\mf g((z^{-1})))$.
Let $\frac{\delta g(z)}{\delta q}[1]=\sum_{n\in\mb Z_+}\frac{\delta g_n}{\delta q}[1]z^{-n+N}$
be the expansion of $\frac{\delta g(z)}{\delta q}[1]\in\mf p((z^{-1}))\otimes\mc V(\mf p)$
in power series of $z$.
Then, the coefficients $\frac{\delta g_n}{\delta q}[1]$ 
span an infinite dimensional subspace of $\mf p\otimes\mc V(\mf p)$.
\end{lemma}
\begin{proof}
It follows from Lemma \ref{20120522:lem4}(c) and the relation \eqref{decz2} 
between the decompositions of $\mf g((z^{-1}))$
in powers of $z$ and with respect to the grading \eqref{decz}.
\end{proof}
\begin{corollary}\label{20120522:cor}
Let $a(z)\in Z(\mf h)\backslash Z(\mf g((z^{-1})))$,
and $\tint g(z)=\sum_{n\in\mb Z_+}\tint g_nz^{-n+N}\in(\quot{\mc W }{\partial\mc W})((z^{-1}))$.
The coefficients $\tint g_n,\,n\in\mb Z_+$, span an infinite-dimensional
subspace of $\quot{\mc W}{\partial\mc W}$.
\end{corollary}
\begin{proof}
Obvious from Lemma \ref{20120522:lem5}.
\end{proof}

\subsection{Conclusion}\label{sec:4.8}
We can summarize all the results obtained in the previous sections in the following
\begin{theorem}\label{final}
Consider the setup of Section \ref{sec:3.1} and
assume, as in Section \ref{sec:4.2}, that
$s\in\Ker(\ad\mf n)$ is a homogeneous element
such that $\mf g((z^{-1}))$ decomposes as in \eqref{eq:dech}.
Let $a(z)\in Z(\mf h)\backslash Z(\mf g((z^{-1}))$,
and let $U(z)\in\mf g((z^{-1}))_{>0}\otimes\mc V(\mf p)$ and $h(z)\in\mf h_{>-1}\otimes\mc V(\mf p)$
be a solution of equation \eqref{L0_ds}.
Consider the differential subalgebra $\mc W\subset \mc V(\mf p)$
defined in \eqref{20120511:eq2}, with the compatible PVA structures 
$\{\cdot\,_\lambda\,\cdot\}_{H,\rho}$ and $\{\cdot\,_\lambda\,\cdot\}_{K,\rho}$
defined in \eqref{wbrackX}--\eqref{HKds}.
Then, the coefficients of the Laurent series $\tint g(z)=\sum_{n\in\mb Z_+}\tint g_nz^{N-n}$ 
defined in \eqref{gz} span an infinite-dimensional subspace of $\quot{\mc W}{\partial\mc W}$
and they satisfy the Lenard-Magri recursion conditions \eqref{eq:lenard2}. 
Hence, they are in involution with respect to both $H$ and $K$:
$$
\{\tint g_m,\tint g_n\}_{H,\rho}
=\{\tint g_m,\tint g_n\}_{K,\rho}=0
\qquad\text{for all }m,n\in\mb Z_+\,,
$$
and they define an integrable hierarchy of bi-Hamiltonian equations,
called the \emph{generalized Drinfeld-Sokolov hierarchy}:
$$
\frac{dp}{dt_n}
=\{{g_n}_{\lambda}p\}_{H,\rho}\big|_{\lambda=0}
=\{{g_{n+1}}_{\lambda}p\}_{K,\rho}\big|_{\lambda=0}
\,,
\qquad p\in\mc W,\,n\in\mb Z_+\,.
$$
\end{theorem}

\subsection{Examples}\label{sec:4.9}

\begin{example}[The KdV hierarchy]
Let us consider the classical $\mc W$-algebra corresponding to the
Lie algebra $\mf{sl}_2$ constructed in Example \ref{W-sl2} and
consider $a(z)=f+zs$. We get $\tint g_0=\tint w$ and the corresponding
Hamiltonian equation is $\frac{dw}{dt_0}=w^\prime$. The next integral of motion is
$\tint g_1=-\tint\frac{w^2}{4}$ and the corresponding Hamiltonian equation is
the \emph{Korteweg-de Vries equation}
$$
\frac{dw}{dt_1}=\frac{1}{4}(w^{\prime\prime\prime}-6ww^\prime)\,.
$$
\end{example}

\begin{example}[The Boussinesq hierarchy]
Let us consider the classical $\mc W$-algebra corresponding to the Lie algebra $\mf{sl}_3$
and its principal nilpotent element $f=E_{12}+E_{23}$,
constructed in Example \ref{W-sl3-principale}.
Letting $a(z)=(f+zs)^2$ (recall that we are working in the matrix realization),
we get $\tint g_0=\tint w_2$
and the corresponding system of Hamiltonian equations is
$$
\left\{\begin{array}{l}
L_t=2{w_2}^\prime\\
{w_2}_t=-\frac16L^{\prime\prime\prime}+\frac23LL^\prime\,.
\end{array}\right.
$$
Eliminating ${w_2}$ from the system we get that $L$ satisfies the
\emph{Boussinesq equation}
$$
L_{tt}=-\frac{1}{3}L^{(4)}+\frac43(LL^\prime)^\prime\,.
$$
\end{example}

\begin{example}
Let us consider the classical $\mc W$-algebra corresponding to the Lie algebra $\mf{sl}_3$
and its minimal nilpotent element $f=E_{31}$,
constructed in Example \ref{W-sl3-minimal}.
In both cases considered (namely the choice $\mf l=0$ or $\mf l\neq0$) the element
$f+zs$, where $s=E_{13}$, 
is semisimple and we get an integrable hierarchy of bi-Hamiltonian equations
by Theorem \ref{final}.
For example, when $\mf l$ is maximal isotropic, letting $a(z)=f+zs$
we get $\tint g_0=\tint (w_2+w_3)$
and the corresponding system of Hamiltonian equations is
$$
\left\{\begin{array}{l}
L_t=\frac12 (w_2'+w_3')\\
{w_2}_t=L-12w_4^2-3w_4'\\
{w_3}_t=-L+12w_4^2-3w_4'\\
{w_4}_t=\frac12(w_2-w_3)\,.
\end{array}\right.
$$
This system of equations was first studied in \cite{BD91} and is known as \emph{fractional KdV system}.
Eliminating $L$, $w_2$, $w_3$ from the system we get that $w_4$ satisfies the equation
$$
w_4''=-\frac13{w_4}_{tttt}-8(w_4{w_4}_t)_t\,,
$$
which, after rescaling, is the Boussinesq equation with the derivatives with respect
to $x$ and $t$ exchanged.
\end{example}

\subsection{Applicability of the integrability scheme for $\mf{gl}_n$}\label{sec:4.10}

It is natural to ask when the assumptions of Theorem \ref{final} hold,
so that the proposed scheme of integrability can be applied.
In other words, given a reductive Lie algebra $\mf g$,
we want to know for which nilpotent elements $f\in\mf g$ (extended to an $\mf{sl}_2$-triple $f,h=2x,e$), 
we are able to find an isotropic subspace $\mf l\subset\mf g_{\frac12}$ and a homogeneous element
$s\in\Ker(\ad\mf n)$ (where $\mf n=\mf l^{\perp\omega}\oplus\mf g_{\geq1}$)
such that $f+zs$ is a semisimple element in $\mf g((z^{-1}))$.

It is not hard to find a general answer in the case of $\mf{gl}_n$.
In this case, the integrability scheme can be applied successfully
for all nilpotent elements $f\in\mf{gl}_n$ corresponding to the partitions of $n$
of the following type:
\begin{enumerate}[(a)]
\item
$n=r+\dots+r+1+\dots+1$,
\item
or $n=r+(r-1)+\dots+r+(r-1)+1+\dots+1$.
\end{enumerate}
For partitions of $n$ of type
$n=r+r+\dots+r+\epsilon$, where $\epsilon=0$ or $1$,
we can choose $s$ in $\Ker(\ad\mf g_{>0})$
(that is with $\mf l=0$),
such that the corresponding element $f+zs\in\mf{gl}_n((z^{-1}))$ 
is regular, homogeneous, semisimple,
which corresponds to integrable hierarchies of ``type I'', \cite{dGHM92,BdGHM93,FHM92}.
Removing the assumptions that $f+zs$ be regular 
(namely considering ``type II hierarchies'') we allow partitions of $n$
with an arbitrary number of $+1$'s.
Furthermore, 
if the partition of $n$ contains copies of $r+(r-1)$,
we are forced to choose $s$ in $\Ker(\ad\mf n)$, with $\mf n$ strictly included in $\mf g_{>0}$,
namely we need to choose a non-zero isotropic subspace
$\mf l\subset\mf g_{\frac12}$.
For partitions as in types (a) and (b) above, the corresponding homogeneous semisimple element $f+zs\in\mf g((z^{-1}))$ is
$$
f+zs^{(x)}=
\begin{pmatrix}
\Lambda^{DS}_{r,x}(z) &  &  & & & & 0 \\
& \Lambda^{DS}_{r,x}(z) &  &  & & & \\
&  & \ddots & & & & & \\
&  &  & \Lambda^{DS}_{r,x}(z) & & & \\
&  &  & & 0 & & \\
&  &  & & & \ddots & \\
0 &  &  & & & & 0
\end{pmatrix}
\,,
$$
where $x=a$ or $b$, and
$$
\begin{array}{l}
\displaystyle{
\Lambda^{DS}_{r,a}(z)=
\left(\begin{array}{ccccc}
0 & &  &  &  z\\
1 & 0 & &   & \\
  & 1 & 0 &  & \\
  &  & \ddots & \ddots & \\
  &  &  &  1 & 0
\end{array}\right)
\,,} \\
\displaystyle{
\Lambda^{DS}_{r,b}=
\left(\begin{array}{cccc|cccc}
0 & & & &&&& z\\
1 & 0 & & &&&&\\
& \ddots & \ddots & &&&& \\
& & 1 & 0 &&&& \\
\hline
&&& z & 0 & & & \\
&&&& 1 & 0 & & \\
&&&& & \ddots & \ddots & \\
&&&& & & 1 & 0
\end{array}\right)
\begin{array}{l}
\left.
\phantom{
\begin{array}{c} 1 \\ 1 \\ 1 \\ 1 \end{array}
}
\hspace{-20pt}
\right\} r
\\
\left.
\phantom{
\begin{array}{c} 1 \\ 1 \\ 1 \\ 1 \end{array}
}
\hspace{-20pt}
\right\} r-1
\end{array}
\,.
}
\end{array}
$$
%
%
We point out that our restrictions on the nilpotent element $f\in\mf{gl}_n$
are the same as those obtained in \cite{FGMS95,FGMS96},
where they constructed generalized Drinfeld-Sokolov integrable hierarchies associated to a graded element
in a Heisenberg subalgebra of $\mf g((z^{-1}))$.

As a final remark, it is not clear if it is possible to further modify
the setup of the integrability scheme to include other types of nilpotent elements $f\in\mf{gl}_n$.
For example, for $f\in\mf{gl}_6$ corresponding to the partition $6=4+2$,
we have $\mf g_{\frac12}=0$ (hence $\mf l=0$) since the gradation is even,
and one can show that there is no choice of $s\in\Ker(\ad\mf n)$, 
homogeneous in the $\ad x$-eigenspaces decomposition,
for which $f+zs$ is a semisimple element of $\mf{gl}_6((z^{-1}))$.


\end{document}